\documentclass[10pt]{article}
\usepackage{amsfonts,amsmath,amssymb,amsthm}
\usepackage[pdftex]{graphicx}
\usepackage[hmargin=1in,vmargin=1in]{geometry}
\usepackage{natbib} 
\usepackage{setspace}
\usepackage{enumerate}
\usepackage[toc,title,titletoc,header]{appendix}
\usepackage[colorlinks,citecolor=blue,hypertexnames=false]{hyperref}
\usepackage{etoolbox}
\usepackage{booktabs}
\usepackage{bbm}
\usepackage{thmtools}
\usepackage{array,longtable}
\usepackage[font=small]{caption}
\newcolumntype{C}{>{$}c<{$}} %
\allowdisplaybreaks
\usepackage{subcaption}
\usepackage{tikz}
\usetikzlibrary{arrows.meta,patterns}

\newenvironment{packed_enum}{\begin{enumerate} \setlength{\itemsep}{1pt}\setlength{\parskip}{0pt}\setlength{\parsep}{0pt}}{\end{enumerate}}

\def\qed{\rule{2mm}{2mm}}

\newtheorem{theorem}{Theorem}[section]
\newtheorem{lemma}{Lemma}[section]

\theoremstyle{definition}
\newtheorem{example}{Example}[section]

\newtheorem{remark}{Remark}[section]
\newtheorem{assumption}{Assumption}[section]

\AtEndEnvironment{remark}{~\qed}
\AtEndEnvironment{example}{~\qed}

\DeclareMathOperator{\var}{Var}

\DeclareMathOperator{\mvec}{vec}

\DeclareMathOperator{\rec}{rec}
\DeclareMathOperator{\lin}{lin}
\DeclareMathOperator{\cl}{cl}

\DeclareMathOperator{\homog}{hom}

\begin{document}

\title{Inference for Linear Systems with Unknown Coefficients\thanks{We thank Geonwoo Kim for outstanding research assistance. Shaikh acknowledges financial support from National Science Foundation Grant SES-2419008. We thank seminar participants at UCSB, UWO, and Georgetown, as well as conference participants at the CESG for helpful comments.}}

\author{
Yuehao Bai \\
Department of Economics \\
University of Southern California \\
\url{yuehao.bai@usc.edu}
\and
Kirill Ponomarev\\
Department of Economics \\
University of Chicago \\
\url{kponomarev@uchicago.edu}
\and
Max Tabord-Meehan\\
Department of Economics \\
University of Toronto \\
\url{m.tabordmeehan@utoronto.ca}
\and
Andres Santos\\
Department of Economics \\
University of California--Los Angeles\\
\url{andres@econ.ucla.edu}
\and
Azeem M.\ Shaikh \\
Department of Economics \\
University of Chicago \\
\url{amshaikh@uchicago.edu}
\and
Alexander Torgovitsky \\
Department of Economics \\
University of Chicago \\
\url{torgovitsky@uchicago.edu}
}

\bigskip

\maketitle

\vspace{-0.3in}

\begin{spacing}{1.1}
\begin{abstract}
This paper considers the problem of testing whether there exists a solution satisfying known linear constraints to a linear system of equations. Importantly and in contrast to some prior work, we allow all parameters in the system of equations, including the slope coefficients, to be unknown. For this reason, we describe the linear system as having unknown (as opposed to known) coefficients. This hypothesis testing problem arises naturally when constructing confidence sets for possibly partially identified parameters in the analysis of nonparametric instrumental variables models, treatment effect models, and random coefficient models, among other settings. We begin our analysis by providing a novel characterization of the set of coefficient pairs for which the linear system is feasible. Leveraging this characterization, we develop a testing procedure based on sample splitting. We establish its validity under weak and interpretable conditions on the linear system. An important feature of these conditions is that they permit the dimensionality of the problem to grow rapidly with the sample size. A further attractive property of our test is that it does not require simulation to compute suitable critical values. We illustrate the practical relevance of our theoretical results in a simulation study.
\end{abstract}
\end{spacing}



\thispagestyle{empty}
\newpage
\setcounter{page}{1}

\section{Introduction}

Given an independent and identically distributed (i.i.d.) sample $\{Z_i\}_{i = 1}^n$ with $Z_i$ distributed according to $P \in \mathbf P$, this paper studies the hypothesis testing problem
\begin{equation} \label{eq:null}
H_0: P \in \mathbf P_0 \qquad \text{versus} \qquad H_1: P \in \mathbf P \setminus \mathbf P_0~,
\end{equation}
where $\mathbf P$ is a ``large'' set of distributions satisfying conditions described below and
\begin{equation} \label{eq:null-Ab}
\mathbf P_0 := \{P \in \mathbf P: A(P) x = \beta(P) \text{ for some } x \in \mathcal X\}~.
\end{equation}
Here, $A(P)$ is a $p \times d$ matrix, $\beta(P)$ is a $p \times 1$ vector, and $\mathcal X$ is a known, nonempty polyhedron given by
\begin{equation} \label{eq:X}
\mathcal X := \{x \in \mathbb R^d: A^{\mathsf{k}, \mathsf{e}} x = \beta^{\mathsf{k}, \mathsf{e}}, A^{\mathsf{k}, \mathsf{i}} x \leq \beta^{\mathsf{k}, \mathsf{i}} \}~,
\end{equation}
where $A^{\mathsf{k}, \mathsf{e}}$, $A^{\mathsf{k}, \mathsf{i}}$, $\beta^{\mathsf{k}, \mathsf{e}}$, and $\beta^{\mathsf{k}, \mathsf{i}}$ are all known. For example, \eqref{eq:X} could incorporate sign restrictions or prespecified bounds on the entries of $x$, or require the entries to sum to one if they represent probabilities.

As discussed further in Section \ref{sec:examples}, the testing problem described above arises naturally in many settings of empirical interest, including (i) inference for linear functionals of structural functions in nonparametric instrumental variables (NPIV) models with shape restrictions, as in \cite{freyberger2015identification}; (ii) tests of whether a representative consumer can rationalize aggregate consumption, as in \cite{varian1982nonparametric}; (iii) inference for treatment effect parameters that are partially identified through the marginal treatment response (MTR) framework of \cite{mogstad2018using}; (iv) inference under ``synthetic parallel trends'' with convex weights as considered in \cite{liu2025synthetic}; and (v) inference for counterfactual choice probabilities in the distribution-free binary choice model in \cite{gurussell2023joe}. We also note that the null hypothesis \eqref{eq:null} subsumes as a special case the setting where $A(P)$ is \emph{known}, for which there exist many other examples \citep[see, for instance, the examples discussed in][]{fang2023inference}.

The first contribution of the paper is to obtain a novel characterization of the set of coefficient pairs $(A, b)$ for which $A x = b$ has a solution in $\mathcal X$, which guides the construction of our test. As a direct consequence of the separating hyperplane theorem, the system has a solution if and only if no hyperplane strictly separates $b$ from $\{Ax: x \in \mathcal X\}$ (after projecting out the images under $A$ of directions parallel to lines contained in $\mathcal X$). This separation condition yields a family of inequalities indexed by vectors in the unit ball in $\ell_1$, with each nonzero vector representing the normal of a potential separating hyperplane. To handle settings in which $\mathcal X$ is unbounded, we normalize these inequalities and obtain a necessary condition for the null that we can evaluate at each vector, by minimizing linear functions over a compact set determined by the known constraints.

Building on this observation, in Section \ref{sec:test} we develop an inference procedure based on sample splitting. Using the first subsample, we select a direction for which the associated inequality appears most strongly violated. We then use the second subsample to test that inequality. By virtue of this sample-splitting construction, the resulting test is straightforward to implement and involves no simulation: in particular, the selection of the direction requires solving a quadratic program, and the test decision is computed by checking the rejection inequality at the vertices of the known compact polyhedron. Moreover, the uniform asymptotic validity of our test is established under weak and interpretable regularity conditions while allowing the dimensions $p$ and $d$ to grow with the sample size $n$.

Inference procedures for testing \eqref{eq:null} and closely related null hypotheses have gained increasing attention in the literature. \cite{bai2022testing}, \cite{andrews2023inference}, \cite{cox2023simple}, and \cite{fang2023inference} all propose inference procedures which could be applied to \eqref{eq:null} whenever $A(P)$ does not depend on $P$ (i.e., it is known). The test proposed in our paper immediately applies to this setting as a special case. There is also a literature for the closely related problem of inference for the \emph{value} of a linear program: see in particular \cite{freyberger2015identification}, \cite{cho2024simple}, \cite{gafarov2025simple}, \cite{voronin2025linear}, \cite{budanova2026inference}. As explained in \cite{cox2025testing}, these procedures can be used to address the same empirical problems as those that we consider in Section \ref{sec:examples}, but there is a subtle technical difference between the two settings, and in general these procedures are neither tuning-parameter- nor simulation-free.

Two recent papers that propose inference procedures which could be used to test \eqref{eq:null} are \cite{cox2025testing} and \cite{goff2026testing}. Their procedures do not employ sample splitting, and thus in our simulation designs tend to be more powerful than our test at a given sample size (although sometimes the gains are quite modest). However, the theoretical results for uniform validity in both these papers require rank conditions on $A(P)$ which, as argued in \cite{liu2025synthetic}, could be difficult to verify or may even be violated in certain settings of empirical interest. In contrast, we are able to establish the asymptotic validity of our test under interpretable rank conditions which hold for a broad collection of empirically relevant examples. In simulations, we demonstrate that these differences in the rank conditions can have important practical implications for uniform size control. Moreover, neither paper establishes the validity of its test in a regime where the number of estimated coefficients grows with sample size, as we do in this paper.\footnote{\citet[Assumption 4.12 and footnote 15]{goff2026testing} allow the number of \emph{known} components to grow, while the number of estimated components of $A$ and $b$ remains bounded.}

The remainder of the paper is organized as follows. In Section \ref{sec:examples}, we describe several examples of empirical problems of interest that can be accommodated in our framework. Our main results are contained in Sections \ref{sec:farkas} and \ref{sec:test}: Section \ref{sec:farkas} presents our characterization of the feasible coefficient pairs, whereas Section \ref{sec:test} describes the test and establishes its uniform asymptotic validity. In Section \ref{sec:sims}, we perform simulation exercises to compare the finite-sample performance of our proposed test to that of the tests proposed in \cite{goff2026testing} and \cite{cox2025testing}. Proofs of all results are collected in the Appendix.

\section{Examples} \label{sec:examples}
In this section, we present a collection of motivating examples. \cite{goff2026testing} and \cite{cox2025testing} discuss several additional examples that provide further motivation.

\begin{example}[Nonparametric instrumental variables with shape restrictions] \label{eg:npiv}
Consider the NPIV model studied in \cite{freyberger2015identification}, in which
\begin{align*}
Y & = g(X) + U~, \\
E_P[U \mid W] & = 0~,
\end{align*}
with $X$ a possibly endogenous explanatory variable supported on $\{x_1, \dots, x_H\}$, and $W$ an instrumental variable supported on $\mathcal{W} = \{w_1, \dots, w_K\}$, where $K < H$. Define
\begin{align*}
\pi_{h k} (P) & = P \{X = x_h, W = w_k\} \\
m_k(P) & = E_P[Y I \{W = w_k\}] \\
\Pi(P) & = (\pi_{h k} (P))_{1 \leq h \leq H, 1 \leq k \leq K} \\
m(P) & = (m_1(P), \dots, m_K(P))'~.
\end{align*}
Representing $g$ as a vector in $\mathbb R^H$, we may impose shape restrictions on $g$, such as monotonicity, through $S g \leq 0$, where $S$ is a known matrix.
Suppose we wish to test the null hypothesis $H_0: L(g) = L_0 \in \mathbb R$, where $L(g) = c' g$ is a linear functional. This null hypothesis fits into \eqref{eq:null-Ab} with $x = g$, $\mathcal X = \{g \in \mathbb R^H:Sg \leq 0,\ c'g = L_0\}$, $A(P) = \Pi(P)'$, and $\beta(P) = m(P)$.
\end{example}

\begin{example}[Revealed Preference for a Representative Consumer] \label{eg:varian}
\cite{varian1982nonparametric} considers whether a representative consumer can rationalize aggregate consumption. In each period $t \in \{1, 2\}$, households face a common vector of known, strictly positive prices $p_t$. Let $Q_t$ denote a household's consumption bundle and let $q_t = E_P[Q_t]$ denote the population mean bundle. The null hypothesis is that the two mean bundles can be rationalized by a single strictly increasing, concave utility function. That is, in each period $t$, $q_t$ maximizes the utility function over bundles costing at most $p_t' q_t$ at prices $p_t$.

Define $a(P) = p_1' (q_2 - q_1)$ and $b(P) = p_2' (q_1 - q_2)$. By Afriat's theorem \citep{varian1982nonparametric}, the null hypothesis holds if and only if there exist a utility difference $u$ between $q_2$ and $q_1$ and strictly positive budget multipliers $\lambda_1, \lambda_2$ satisfying
\[ -\lambda_2 b(P) \leq u \leq \lambda_1 a(P)~. \]
Rescaling $u$, $\lambda_1$, and $\lambda_2$ by $1 / \min\{\lambda_1, \lambda_2\}$, we can impose $\lambda_1, \lambda_2 \geq 1$ without loss of generality, following the construction in \citet[Algorithm 3]{varian1982nonparametric}. With nonnegative slack variables $s_1, s_2$, the null hypothesis fits into \eqref{eq:null-Ab} with $x = (u, \lambda_1, \lambda_2, s_1, s_2)'$, $\mathcal X = \mathbb R \times [1, \infty)^2 \times [0, \infty)^2$,
\[ A(P) = \begin{pmatrix} 1 & -a(P) & 0 & 1 & 0 \\ -1 & 0 & -b(P) & 0 & 1 \end{pmatrix}~, \]
and $\beta(P) = (0, 0)'$.
\end{example}

\begin{example}[Instrumental variables with heterogeneous treatment effects] \label{eg:mst}
\cite{mogstad2018using} develop an approach to marginal treatment effect analysis \citep{heckman2005structural} that allows for shape constraints and partial identification. A binary treatment $D$ produces two real-valued potential outcomes $Y(0)$, $Y(1)$, and an observed outcome $Y = (1 - D) Y(0) + D Y(1)$. The researcher additionally observes an exogenous instrument $Z$. Treatment assignment satisfies the \cite{imbens1994identification} monotonicity condition, which can be equivalently written with the threshold-crossing model $D = I \{p(Z) \geq U\}$, where $U$ is a uniformly distributed unobservable that is independent of $Z$ and $p(z) := P \{D = 1 \vert Z = z\}$ is the propensity score \citep{vytlacil2002independence}. The marginal treatment response functions are assumed to take a linear-in-parameters form
\begin{equation} \label{eq:mst:linear-in-parameters}
E[Y(d) \vert U = u] = \theta(d)' b(d \vert u)~,
\end{equation}
where $\theta(d)$ are unknown parameters and $b(d \vert u)$ are known basis functions.

The linear parameterization \eqref{eq:mst:linear-in-parameters} implies that common target parameters can also be written as linear functions of the parameters, $\theta := (\theta(0)', \theta(1)')'$, taking the general form
\begin{align}
\nonumber \tau^\star & = \theta(0)' E_P \left [ \int_0^1 b(0 \vert u) \omega^\star(0 \vert u, Z) \, d u \right ] + \theta(1)' E_P \left [ \int_0^1 b(1 \vert u) \omega^\star(1 \vert u, Z) \, d u \right ] \\
\label{eq:mst:target-parameter} & = \sum_{d \in \{0, 1\}} \theta(d)' E_P[t(d \vert Z)]~,
\end{align}
where $\omega^\star(d \vert u, z)$ are scalar weights that are known or identified and $t(d \vert z) = \int_0^1 b(d \vert u) \omega^\star(d \vert u, z) \, d u$. \cite{mogstadtorgovitsky2024hole} observe that if $Y(d)$ is mean independent of $Z$, conditional on $U$, then the linear-in-parameters form \eqref{eq:mst:linear-in-parameters} also has implications for the observed outcome:
\begin{equation} \label{eq:mst:implied-conditional-mean}
E_P[Y \vert D, Z] = \phi(D, Z)' \theta~,
\end{equation}
where $\phi$ is a function of $(D, Z)$ that depends on the basis functions, $b$, and the propensity score, $p(Z)$. Bounds on $\tau^\star$ can be found by considering all values of \eqref{eq:mst:target-parameter} that can be produced by values of $\theta$ that satisfy \eqref{eq:mst:implied-conditional-mean} either for all $(D, Z)$ or for some implied moments. \cite{mogstad2018using} propose using moments of the form $E_P[Y s(D, Z)]$, which can be shown from \eqref{eq:mst:implied-conditional-mean} to also be linear in $\theta$. \cite{sheatorgovitsky2023os} alternatively propose using the normal equations implied by \eqref{eq:mst:implied-conditional-mean}:
\[ E_P[\phi(D, Z) \phi(D, Z)'] \theta = E_P[\phi(D, Z) Y]~, \]
which has the advantage of always being the same dimension as $\theta$ and not requiring one to choose the $s$ functions. In either case, shape constraints can be imposed on the marginal treatment response functions by constraining $\theta$.

Consider testing whether $\tau^\star = \tau_0$ for $\tau_0 \in \mathbb R$. Without shape constraints, the normal equation version of the problem can be phrased as \eqref{eq:null-Ab} by taking $x = \theta$, $\mathcal X = \mathbb R^{d_\theta}$, where $d_\theta$ is the number of coefficients in $\theta$, and setting
\begin{equation} \label{eq:mst-a0-no-constraints}
A(P) = \begin{pmatrix}E_P[\phi(D, Z) \phi(D, Z)'] \\ E_P[t(Z)]'\end{pmatrix} \qquad \text{and} \qquad
\beta(P) = \begin{pmatrix}E_P[\phi(D, Z) Y] \\ \tau_0\end{pmatrix}~,
\end{equation}
where $t(Z) = (t(0 \vert Z)', t(1 \vert Z)')'$. Shape constraints can be imposed directly through the known polyhedron $\mathcal X$. For example, if $Y \in \{0, 1\}$ is binary and $b(d \vert u)$ are Bernstein polynomials, then the implied MTR function can be constrained to lie in $[0, 1]$ by restricting all elements of $\theta$ to lie within $[0, 1]$. In that case, $A(P)$ and $\beta(P)$ remain the same and $\mathcal X = [0, 1]^{d_\theta}$.
\end{example}

\begin{example}[Synthetic parallel trends with convex weights]\label{eg:spt}
Consider the causal panel data setting presented in \citet{liu2025synthetic}. There are $K$ aggregate units indexed by $k \in \{1, \dots, K\}$, observed over periods $t \in \{1, \dots, T_0, T\}$, where $\{1, \ldots, T_0\}$ denotes pre-treatment periods and $T$ denotes a treatment period at which unit $k = 1$ is treated (units $k \geq 2$ are never treated). Let $\tilde \mu_t^k(1)$ denote the potential aggregate outcome for unit $k$ at time $t$ with treatment and $\tilde \mu_t^k(0)$ denote its potential aggregate outcome without treatment. The observed aggregate outcome is given by
\[ \mu_t^k(P) = \tilde \mu_t^k(0) + \big ( \tilde \mu_t^k(1) - \tilde \mu_t^k(0) \big ) I \{k = 1, t = T\}~. \]
The target parameter is the effect on the treated unit at time $T$:
\[ \tau = \tilde \mu_T^1(1) - \tilde \mu_T^1(0)~. \]
\cite{liu2025synthetic} maintains the assumption of (convex) \emph{synthetic parallel trends} (SPT); that is, there exists a set of weights $(\omega_k: 2 \leq k \leq K) \in \mathbb R^{K - 1}$ with $\sum_{2 \leq k \leq K} \omega_k = 1$, $\omega_k \geq 0$ for all $k$ such that for every $t \in \{2, \dots, T_0, T\}$,
\[ \sum_{k = 2}^K \omega_k \Delta \tilde \mu_t^k(0) = \Delta \tilde \mu_t^1(0)~, \]
where $\Delta \tilde \mu_t^k(0) = \tilde \mu_t^k(0) - \tilde \mu_{t - 1}^k(0)$ for $2 \leq t \leq T_0$ and $\Delta \tilde \mu_T^k(0) = \tilde \mu_T^k(0) - \tilde \mu_{T_0}^k(0)$. Let $\Delta \mu_t^k(P) = \Delta \tilde \mu_t^k(0)$ for all $k \geq 2$ and $t \in \{2, \dots, T_0, T\}$. For the treated unit, define $\Delta \mu_t^1(P) = \mu_t^1(P) - \mu_{t - 1}^1(P)$ for $2 \leq t \leq T_0$. Suppose we wish to test the null hypothesis $H_0: \tau = \tau_0$. Then, under the convex SPT assumption, this problem fits into \eqref{eq:null-Ab} with $x = (\omega_2, \dots, \omega_K)'$, $\mathcal X = \{x \in \mathbb R_+^{K - 1}:\mathbf 1_{K - 1}' x = 1\}$,
\[
A(P) = \begin{pmatrix}\Delta \mu_2^2(P) & \cdots & \Delta \mu_2^K(P) \\ \vdots & \ddots & \vdots \\ \Delta \mu_{T_0}^2(P) & \cdots & \Delta \mu_{T_0}^K(P) \\ \Delta \mu_T^2(P) & \cdots & \Delta \mu_T^K(P)\end{pmatrix}~, \qquad \text{and} \qquad
\beta(P) = \begin{pmatrix}\Delta \mu_2^1(P) \\ \vdots \\ \Delta \mu_{T_0}^1(P) \\ \mu_T^1(P) - \tau_0 - \mu_{T_0}^1(P)\end{pmatrix}~.
\]
\cite{liu2025synthetic} also analyzes the setting where we drop the assumption of convexity, so that $\omega_k$ are not restricted to be non-negative. In this case, $A(P)$ and $\beta(P)$ stay the same and $\mathcal X = \{x \in \mathbb R^{K - 1}:\mathbf 1_{K - 1}' x = 1\}$.
\end{example}

\begin{example}[Distribution-free binary choice] \label{eg:binchoice}
\cite{gurussell2023joe} consider binary choice models of the form
\[ Y = I \{\varphi(D, Z, U) \geq 0\}~, \]
where $Y$ is a binary outcome, $\varphi$ is an unknown function, $D$ is an endogenous regressor, $Z$ is a vector of exogenous regressors, and $U$ is a vector of unobservables. The distribution of $U$ is not restricted to lie in a parametric family. The authors observe that if $D$ and $Z$ have finite supports, then the conditional distribution of $Y$ implied by the model is determined by the mass placed on a finite partition of the support of $U$ into sets $\mathcal U_j$, $j = 1, \ldots, d_U$. In particular,
\begin{equation} \label{eq:binary-response:obs-eq}
P \{Y = 1 \vert D = d, Z = z\} = \sum_{1 \leq j \leq d_U} I \{j \in \mathcal J(d, z) \} \theta_j(d, z)~,
\end{equation}
where $\theta_j(d, z)$ is the mass that the distribution of $U$ places on $\mathcal U_j$, conditional on $D = d, Z = z$, and the set $\mathcal J(d, z)$ collects the appropriate indices for sets that lead to $Y = 1$ when $D = d$ and $Z = z$. If the instrument $Z$ is independent of $U$, then also
\begin{equation} \label{eq:binary-response:independence}
\sum_d \theta_j(d, z) P \{D = d \vert Z = z\} = \sum_d \theta_j(d, z') P \{D = d \vert Z = z' \}\text{ for all }z, z', j~.
\end{equation}
A natural target parameter in this model is the counterfactual choice probability $\pi(\theta; P) := \mathrm{Pr} \{Y(d^\star) = 1\} = \mathrm{Pr} \{\varphi(d^\star, Z, U) \geq 0\}$ at some fixed $d^\star$, which can also be expressed as a linear function of the $\theta_j(d, z)$ if the sets $\mathcal U_j$ have been constructed to be sufficiently fine:
\begin{align}
\nonumber \pi(\theta; P) & = \sum_{1 \leq j \leq d_U} E_P \left [ I \{j \in \mathcal J_\pi(d^\star, Z) \} \theta_j(D, Z) \right ] \\
\label{eq:binary-response:target-param} & = \sum_{1 \leq j \leq d_U} \sum_{d, z} I \{j \in \mathcal J_\pi(d^\star, z) \} \theta_j(d, z) P \{D = d, Z = z\}~.
\end{align}
Here, $\mathcal J_\pi(d^\star, z)$ collects the indices of partition cells that lead to $Y(d^\star) = 1$ when $Z = z$. Similar observations have been used for multinomial choice models by \cite{manski2007ier}, \cite{tebalditorgovitskyyang2023e}, and \cite{gurussellstringham2026}.

Suppose we wish to test the null hypothesis $H_0: \pi(\theta; P) = \pi_0$ for a vector of probabilities $\{\theta_j(d, z) \}_{j, d, z}$ that satisfies \eqref{eq:binary-response:obs-eq} and \eqref{eq:binary-response:independence} when $D$ and $Z$ both have finite supports. This problem fits into \eqref{eq:null-Ab} with $x$ taken to be the $\theta_j(d, z)$ arranged in a vector across $(j, d, z)$,
\[ \mathcal X = \left\{x:\theta_j(d,z) \geq 0,\ \sum_{j=1}^{d_U}\theta_j(d,z)=1\text{ for every }(d,z)\right\}~, \]
and $A(P)$ and $\beta(P)$ from the linear functions \eqref{eq:binary-response:obs-eq}--\eqref{eq:binary-response:target-param}, with \eqref{eq:binary-response:target-param} set equal to $\pi_0$. The rows of $A(P)$ corresponding to \eqref{eq:binary-response:independence}--\eqref{eq:binary-response:target-param} both depend on $P$.
\end{example}

\section{A Useful Characterization of \texorpdfstring{$\mathbf P_0$}{P0}} \label{sec:farkas}

In this section, we obtain a characterization of the set of coefficients $(A(P), \beta(P))$ satisfying the null hypothesis in \eqref{eq:null-Ab} that directly motivates the test in Section \ref{sec:test}. 
To this end, we define the set
\begin{equation}\label{def:C0}
 \mathbf C_0 := \{(A, b) \in \mathbb R^{p \times d} \times \mathbb R^p: A x = b \text{ for some } x \in \mathcal X\}~, 
\end{equation}
so that the null hypothesis may be equivalently written as $\mathbf P_0 = \{P \in \mathbf P: (A(P), \beta(P)) \in \mathbf C_0\}$. 

An important concept for our derivations is the lineality space of $\mathcal X$, which we denote by $\lin(\mathcal X)$ \citep[p.\ 65]{rockafellar1970convex}.
Intuitively, $\lin(\mathcal X)$ consists of the directions along which we may move (up or down) arbitrarily far without leaving the polyhedron $\mathcal X$.
Formally, given the formulation of $\mathcal X$ in \eqref{eq:X}, $\lin(\mathcal X)$ equals
\begin{align} \label{eq:lineality}
\lin(\mathcal X) & :=  \{v\in \mathbb R^d : x+ sv \in \mathcal X \text{ for all } x\in \mathcal X, ~ s\in \mathbb R\} \notag \\
& \phantom{:}= \{v \in \mathbb R^d: A^{\mathsf{k}, \mathsf{e}} v = 0, ~A^{\mathsf{k}, \mathsf{i}} v = 0\}~.    
\end{align}
Our first result employs $\lin(\mathcal X)$ and the separating hyperplane theorem to obtain an equivalent characterization of $\mathbf C_0$. 
In its statement, we denote the orthogonal complement of $\lin(\mathcal X)$ by $\lin(\mathcal X)^\perp := \{w \in \mathbb R^d : w^\prime v= 0 \text{ for all } v\in \lin(\mathcal X)\}$, $A(\lin(\mathcal X)) = \{Av: v \in \lin (\mathcal{X})\}$, and set $\|y\|_1 := \sum_i |y_i|$ for any vector $(y_1,\ldots, y_p)\in \mathbb R^p$. 

\begin{lemma} \label{lem:annihilator}
For each matrix $A$, let $M_A$ denote the orthogonal projection matrix that projects onto the orthogonal complement of the vector space $A(\lin(\mathcal X))$. Then, it follows that
\begin{equation}  \label{eq:farkas1}
\mathbf C_0 %
= \bigg\{ (A, b): \sup_{\|y\|_1 \leq 1} \inf_{x \in \mathcal X \cap \lin(\mathcal X)^\perp} y' M_A (A x - b) = 0 \bigg\}~.  
\end{equation}
\end{lemma}

Lemma \ref{lem:annihilator} provides us with an alternative characterization of the null hypothesis in \eqref{eq:null-Ab}.
Specifically, since $P\in \mathbf P_0$ if and only if $(A(P),\beta(P))\in \mathbf C_0$, Lemma \ref{lem:annihilator} implies the null hypothesis holds if and only if
\begin{equation*}%
\sup_{\|y\|_1\leq 1} \inf_{x\in \mathcal X \cap\lin(\mathcal X)^\perp} y^\prime M_{A(P)}(A(P)x-\beta(P)) = 0 ~.
\end{equation*}
Unfortunately, constructing a test based on a sample analog to this characterization is challenging  when the set $\mathcal X \cap \lin(\mathcal X)^\perp$ is unbounded. 
In such cases, for example, estimation error can be arbitrarily magnified by large values of $x$ in the inner infimum, leading to asymptotic approximations failing to hold uniformly.
We next address this problem by developing a formulation in which the inner infimum is taken over a compact polyhedron instead.
To this end, we introduce the following normalizing function
\begin{equation}\label{def:littleh}
 h(x) := 1 + \mathbf 1' (\beta^{\mathsf{k}, \mathsf{i}} - A^{\mathsf{k}, \mathsf{i}} x)~, 
 \end{equation}
where $\mathbf 1$ denotes the vector of ones with the same dimension as $\beta^{\mathsf{k}, \mathsf{i}}$. If there are no known inequality constraints in $\mathcal X$, then $\mathbf 1' (\beta^{\mathsf{k}, \mathsf{i}} - A^{\mathsf{k}, \mathsf{i}} x)$ is understood to be zero and $h(x)$ to equal one for all $x$.

The normalizing function $h$ has two important properties. First, it satisfies $h(x)\geq 1$ for any $x\in \mathcal X$ (because $A^{\mathsf{k}, \mathsf{i}} x \leq \beta^{\mathsf{k}, \mathsf{i}}$ for all $x\in \mathcal X$).
Second, as we show in Lemma \ref{lem:normalization} below, $x/h(x)$ is bounded on $\mathcal X \cap \lin(\mathcal X)^\perp$.
Together with Lemma \ref{lem:annihilator}, the first property (i.e., $h(x) \geq 1$) implies the set inclusion
\begin{equation} \label{eq:normalized-condition}
\mathbf C_0 \subseteq \bigg \{(A, b): \sup_{\|y\|_1 \leq 1} \inf_{x \in \mathcal X \cap \lin(\mathcal X)^\perp} \frac{y' M_A (A x - b)}{h(x)} = 0 \bigg \}~.
\end{equation}
The second property (i.e. $x/h(x)$ being bounded on $\mathcal X \cap \lin(\mathcal X)^\perp$) ensures that the infimum in \eqref{eq:normalized-condition} is now effectively taken over bounded sets, thus delivering a formulation that is amenable to use in a test.
However, this normalization has come at the cost of the inner optimization problem being nonlinear in $x$.
As our final reformulation step, we recover linearity by re-parametrizing the problem through the Charnes--Cooper transformation \citetext{\citealp{charnes1962programming}}.
Specifically, setting $z = x/h(x)$ and $t = 1/h(x)$ we obtain 
\begin{equation}\label{sec:farkas2}
\frac{y^\prime M_A(Ax-b)}{h(x)} = y^\prime M_A(Az - bt)~.
\end{equation}
 We denote the closure of the set of values for $(z,t)$ that we can obtain through this transformation by 
\begin{equation} \label{eq:normalized-domain}
\mathcal H := \cl \bigg \{ \bigg ( \frac{x}{h(x)}, \frac{1}{h(x)} \bigg ): x \in \mathcal X \cap \lin(\mathcal X)^\perp \bigg \} ~,
\end{equation}
where ``$\cl$'' denotes the closure in the Euclidean topology on $\mathbb R^d \times \mathbb R$.
Therefore, employing the re-parametrization in \eqref{sec:farkas2} and the definition of $\mathcal H$, we may write the nonlinear infimum in \eqref{eq:normalized-condition} in the equivalent formulation
\begin{equation} \label{eq:normalized-equivalent-prelim}
\inf_{x \in \mathcal X \cap \lin(\mathcal X)^\perp} \frac{y' M_A (A x - b)}{h(x)} = \inf_{(z, t) \in \mathcal H} y' M_A (A z - b t)~.    
\end{equation}

The utility of the final reformulation in \eqref{eq:normalized-equivalent-prelim} hinges on the properties of $\mathcal H$.
Our next result shows that $\mathcal H$ is in fact a compact polyhedron, which implies that the right-hand side of \eqref{eq:normalized-equivalent-prelim} is a bounded linear program.

\begin{lemma} \label{lem:normalization}
Let $\mathcal X$ and $\mathcal H$ be as defined in \eqref{eq:X} and \eqref{eq:normalized-domain} respectively.
Then, it follows that 
\begin{equation} \label{eq:normalization}
\mathcal H = \bigg \{  (z, t)\in \mathbb R^d \times \mathbb R:
\begin{array}{c}
t \geq 0, \qquad z \in \lin(\mathcal X)^\perp, \qquad t + \mathbf 1' (t \beta^{\mathsf{k}, \mathsf{i}} - A^{\mathsf{k}, \mathsf{i}} z) = 1 \\[2pt]
A^{\mathsf{k}, \mathsf{e}} z = t \beta^{\mathsf{k}, \mathsf{e}}, \qquad A^{\mathsf{k}, \mathsf{i}} z \leq t \beta^{\mathsf{k}, \mathsf{i}}
\end{array} \bigg\}~, 
\end{equation}
and that $\mathcal H$ is a nonempty compact polyhedron.
Therefore, for every $(A, b) \in \mathbb R^{p \times d} \times \mathbb R^p$, it follows that 
\begin{equation}\label{eq:normalized-equivalent}
\sup_{\|y\|_1\leq 1} \inf_{x \in \mathcal X \cap \lin(\mathcal X)^\perp} \frac{y' M_A (A x - b)}{h(x)} = \max_{\|y\|_1\leq 1} \min_{(z, t) \in \mathcal H} y' M_A (A z - b t)~.   
\end{equation}
\end{lemma}

The first claim of Lemma \ref{lem:normalization} derives an explicit characterization of $\mathcal H$  and establishes its compactness.
Building on this result, the second claim exploits compactness of $\mathcal H$ to establish attainment of the inner infimum and, through the theorem of the maximum, of the outer supremum as well.
The main implication of Lemma \ref{lem:normalization} and our preceding discussion (recall, in particular,  \eqref{eq:normalized-condition}) is a reformulation of the null hypothesis that is amenable to the construction of a test through sample analogs.
Specifically, for $\bar {\mathbf C}_0$ given by
\begin{equation} \label{eq:C0bar}
\bar{\mathbf C}_0 := \bigg \{ (A, b) \in \mathbb R^{p \times d} \times \mathbb R^p: \max_{\|y\|_1 \leq 1} \min_{(z, t) \in \mathcal H} y' M_A (A z - b t) = 0 \bigg \}~,
\end{equation}
our results suggest testing the null hypothesis by examining whether $(A(P),\beta(P))\in \bar {\mathbf C}_0$. 
This is precisely the approach we take to develop our test in Section \ref{sec:test}.

A potential concern with basing a test on whether $(A(P),\beta(P))$ belongs to $\bar {\mathbf C}_0$ is that the null hypothesis in \eqref{eq:null-Ab} in fact specifies that $(A(P),\beta(P))$ belongs to $\mathbf C_0$  (recall definition  \eqref{def:C0}).
The set inclusion $\mathbf C_0 \subseteq \bar {\mathbf C}_0$ immediately implies that an (asymptotically) level-$\alpha$ test for whether $(A(P),\beta(P))\in \bar {\mathbf C}_0$ is also an (asymptotically) level-$\alpha$ test for the null hypothesis of interest.
However, whenever the set inclusion $\mathbf C_0 \subseteq \bar{\mathbf C}_0$ is strict, a test based on $\bar {\mathbf C}_0$ may fail to be consistent against alternatives $P$ for which $(A(P),\beta(P))\in \bar{\mathbf C}_0 \setminus \mathbf C_0$.
Our final result in this section examines this concern by characterizing the difference between $\mathbf C_0$ and $\bar {\mathbf C}_0$.
The statement of this result requires additional notation, and to this end we introduce the set
\begin{equation} \label{eq:CRD:def}
\mathbf C^{\rm R D} := \{(A, b) \in \mathbb R^{p \times d} \times \mathbb R^p: \dim(A(\lin(\mathcal X))) < \dim(\lin(\mathcal X)) \}~,
\end{equation}
where $\dim(\mathbb V)$ denotes the dimension of a vector space $\mathbb V$ and RD stands for ``rank deficiency'' (of $A$). 
We also require the recession cone of $\mathcal X$ \citep{rockafellar1970convex}, denoted  $\rec(\mathcal X)$, which represents the directions along which we may move up (but not necessarily down) without leaving $\mathcal X$.
Formally, $\rec(\mathcal X)$ is given by
\begin{align}\label{def:rec}
\rec(\mathcal X) & := \{v\in \mathbb R^d : x + sv \in \mathcal X \text{ for all } x\in \mathcal X,~ s\geq 0\} \notag \\
& = \{v\in \mathbb R^d : A^{\mathsf{k,e}} v = 0, ~ A^{\mathsf{k}, \mathsf{i}} v \leq 0\} ~.   
\end{align}

Our next result characterizes the difference between $\mathbf C_0$, $\bar {\mathbf C}_0$, and the closure of $\mathbf C_0$.

\begin{lemma} \label{lem:closure}
Let $\cl(\mathbf C_0)$ denote the closure of $\mathbf C_0$ in the Euclidean topology on $\mathbb R^{p \times d} \times \mathbb R^p$. Then:
\begin{packed_enum}
    \item[(i)] If $\mathcal X$ is compact, then $\mathbf C_0 = \bar {\mathbf C}_0 = \cl(\mathbf C_0)$.
    \item[(ii)] $\bar{\mathbf C}_0 = \mathbf C_0 \cup \{(A, b) \in \mathbb R^{p \times d} \times \mathbb R^p: Av = 0 \text{ for some } v \in \rec(\mathcal X) \setminus \lin(\mathcal X)\}$.
    \item[(iii)] $\cl(\mathbf C_0) = \bar {\mathbf C}_0 \cup \mathbf C^{\rm RD}$.
    \item[(iv)] $\mathbf C_0 \subseteq \bar {\mathbf C}_0 \subseteq \cl(\mathbf C_0)$.    
\end{packed_enum}
\end{lemma}

The first part of Lemma \ref{lem:closure} establishes that $\mathbf C_0$, $\bar {\mathbf C}_0$, and $\cl(\mathbf C_0)$ are all equal whenever the polyhedron $\mathcal X$ is compact.
Therefore, testing the null hypothesis in \eqref{eq:null-Ab} is equivalent to testing whether $P$ belongs to
\begin{equation}\label{eq:barP0def}
\bar{\mathbf P}_0 := \{P\in \mathbf P : (A(P),\beta(P))\in \bar {\mathbf C}_0\}    
\end{equation}
in such settings.
This equivalence applies, among other instances, in the instrumental variable models of \cite{freyberger2015identification} and \cite{mogstad2018using} when outcomes are restricted to be bounded by $\mathcal X$ (see Examples \ref{eg:npiv} and \ref{eg:mst}), and in applications in which $\mathcal X$ encodes a set of probabilities such as in the synthetic parallel trends approach of \cite{liu2025synthetic}, and the distribution-free binary choice model of \cite{gurussell2023joe} (see Examples \ref{eg:spt} and \ref{eg:binchoice}).
Thus, in most of the examples that motivate us, testing the null hypothesis in \eqref{eq:null-Ab} is fully equivalent to testing whether $P\in \bar {\mathbf P}_0$, which is the approach we pursue in Section \ref{sec:test}.

The second part of Lemma \ref{lem:closure} more generally characterizes the difference between $\bar{\mathbf C}_0$ and $\mathbf C_0$.
In particular, Lemma \ref{lem:closure}(ii) implies that the difference between the null hypothesis that $P$ belongs to $\mathbf P_0$ (as in \eqref{eq:null-Ab}) and the null hypothesis that $P$ belongs to $\bar{\mathbf P}_0$ is characterized by the set
\begin{equation*} %
\bar{\mathbf P}_0 \setminus \mathbf P_0 = \{P \in \mathbf P \setminus \mathbf P_0: A(P) v = 0 \text{ for some } v \in \rec(\mathcal X) \setminus \lin(\mathcal X)\}~.    
\end{equation*}
For example, $\mathcal X$ being bounded implies $\rec(\mathcal X) = \lin(\mathcal X) = \{0\}$ and hence that $\mathbf P_0$ equals $\bar {\mathbf P}_0$ --- a conclusion that also follows from Lemma \ref{lem:closure}(i).
As we show in Appendix \ref{sec:coefficient-preprocessing}, however, Lemma \ref{lem:closure}(ii) can also be used to show $\mathbf P_0$ equals $\bar{\mathbf P}_0$ in applications with unbounded $\mathcal X$, such as a special case of Example \ref{eg:npiv} with an unbounded outcome.
Moreover, we emphasize that $\mathbf P_0$ need not equal $\bar{\mathbf P}_0$ in order for both testing problems to be equivalent.
Classical impossibility results imply that any test that has uniform asymptotic level $\alpha \in (0, 1)$ over $\mathbf P_0$ must also have asymptotic power at most $\alpha$ at every $P \in \bar{\mathbf P}_0 \setminus \mathbf P_0$ whenever $\bar {\mathbf P}_0\setminus \mathbf P_0$ is contained in the closure of $\mathbf P_0$ under total variation \citep{romano2004non}.
These impossibility results can alternatively also be combined with parts (iii) and (iv) of Lemma \ref{lem:closure}, which establish that $\bar{\mathbf C}_0$ is sandwiched between $\mathbf C_0$ and its closure.

We conclude this section by illustrating the geometry of our results through two stylized examples.

\begin{figure}[t!]
\centering
\includegraphics{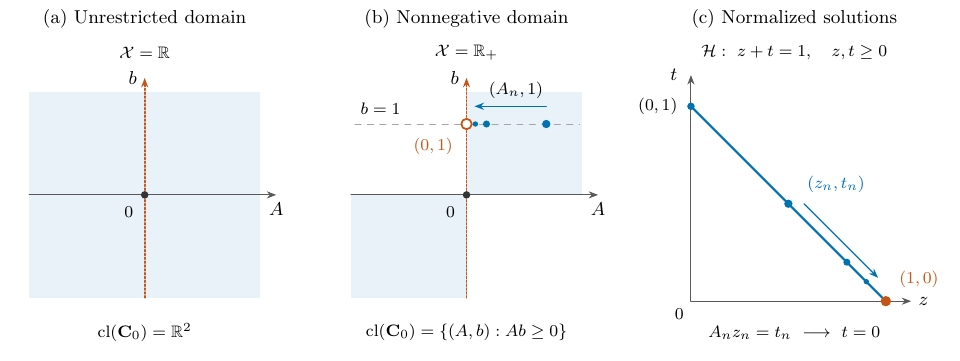}
\caption{Feasible coefficient pairs and normalized solutions for $A x = b$. In panels (a) and (b), the shaded regions and the horizontal axis constitute $\mathbf C_0$; the dotted vertical axis excluding the origin is contained in $\cl(\mathbf C_0) \setminus \mathbf C_0$. In panel (b), the points $(A_n, 1)$, with $A_n \downarrow 0$, approach an infeasible linear system with $(A,b) = (0, 1)$. Panel (c) shows the normalized solutions approaching the included endpoint $(1, 0)$ of $\mathcal H$, where $t = 0$ and the limiting equation holds.}
\label{fig:scalar-closure}
\end{figure}

\begin{example}\label{eg:scalar:p1}
Suppose that $A$ and $b$ are scalars ($A,b\in \mathbb R$) and $\mathcal X$ equals the real line.
In this case, the linear system $Ax = b$ has a solution in $x\in \mathcal X =\mathbb R$ if and only if $A\neq 0$ or $b = 0$ --- i.e., $\mathbf C_0 = \{(A,b) \in \mathbb R \times \mathbb R: A\neq 0 \text{ or } b = 0\}$.
Figure \ref{fig:scalar-closure}(a) depicts this setting, with $\mathbf C_0$ consisting of the shaded region and the horizontal axis.
Since in this example $\rec(\mathcal X) = \lin(\mathcal X) = \mathbb R$, Lemma \ref{lem:closure}(ii) implies $\bar{\mathbf C}_0 = \mathbf C_0$.
The equivalence between $\bar{\mathbf C}_0$ and $\mathbf C_0$ can alternatively be seen by directly computing $\bar{\mathbf C}_0$.
Since $\lin(\mathcal X)^\perp = \{0\}$ and $A(\lin(\mathcal X)) = A(\mathbb R)$ it follows that $\mathcal H = \{(0,1)\}$ (see definition \eqref{eq:normalization}) and that $M_A = I\{A = 0\}$.
Therefore, we obtain
\begin{equation}\label{eg:scalar:p1:1}
\sup_{|y|\leq 1} \min_{(z,t)\in \mathcal H} y M_A(Az-bt) = |M_A b| = \left\{\begin{array}{cl} 0 & \text{ if } A\neq 0 \\ |b| & \text{ if } A =0\end{array}\right. ~,
\end{equation}
which, by definition \eqref{eq:C0bar}, implies $\bar{\mathbf C}_0 = \{(A,b) : A\neq 0 \text{ or } b = 0\} = \mathbf C_0$. 
Moreover, as is evident from Figure \ref{fig:scalar-closure}(a), in this example $\cl(\mathbf C_0) = \mathbb R\times \mathbb R$.
This conclusion can also be obtained from Lemma \ref{lem:closure}(iii) after noting $\dim(A(\lin(\mathcal X))) < \dim(\lin(\mathcal X))$ only when $A = 0$ and hence $\mathbf C^{\rm R D} = \{(0, b):b \in \mathbb R\}$. 
\end{example}

\begin{example}\label{eg:scalar:p2}
We next modify Example \ref{eg:scalar:p1} by letting $\mathcal X = \{x \in \mathbb R :x\geq 0\}$ be the positive orthant.
The linear system $Ax= b$ now has a solution if and only if $Ab > 0$ or $b = 0$, which implies $\mathbf C_0 = \{(A,b) \in \mathbb R \times \mathbb R : Ab > 0 \text{ or } b = 0\}$; see Figure \ref{fig:scalar-closure}(b). %
Moreover, since $\rec(\mathcal X) = \mathcal X$ and $\lin(\mathcal X) = \{0\}$, it follows that $\rec(\mathcal X) \setminus \lin(\mathcal X) = \{x : x > 0\}$ and Lemma \ref{lem:closure}(ii) implies that $\bar {\mathbf C}_0 = \{(A,b) \in \mathbb R \times \mathbb R: Ab \geq 0\}$. 
As is evident from Figure \ref{fig:scalar-closure}(b), $\bar {\mathbf C}_0$ coincides with the closure of $\mathbf C_0$ --- a conclusion that can also be obtained through Lemma \ref{lem:closure}(iii). %
To see the role of $\mathcal H$, note that $h(x) = 1+x$ because there is a single inequality constraint ($-x\leq 0$) defining $\mathcal X$.
The normalizing transformation $x\mapsto (x/h(x),1/h(x))$ then maps $\mathcal X$ to the set with closure
\begin{equation*}%
\mathcal H %
 = \{(z,t) \in \mathbb R\times \mathbb R : z\geq 0, ~ t\geq 0, ~ z+t = 1\} ~,
\end{equation*}
which is compact. 
As a result, the image of every sequence $x_n\in\mathcal X$ has a convergent subsequence in $\mathcal H$.
For example, the sequence $(A_n,b_n)$ with $A_n >0$ converging to zero and $b_n =1$ belongs to $\mathbf C_0$ because $A_n x_n = b_n$ for $x_n =1/A_n$. 
The sequence $x_n$ diverges, but applying the transformation $x\mapsto (x/h(x), 1/h(x))$ we obtain
\begin{equation*}%
(z_n, t_n) := \bigg ( \frac{x_n}{1 + x_n}, \frac{1}{1 + x_n} \bigg ) \to (1, 0) ~;
\end{equation*}
see Figures \ref{fig:scalar-closure}(b)(c). 
In this way, while the linear system $A_n x = b_n$ does not have a solution at the limit point $(A,b) = (0,1)$, the normalization ensures that the equalities $A_n z_n = t_n b_n$ hold in the limit as well. 
\end{example}

\section{The Test} \label{sec:test}
\subsection{Description of the Test} \label{sec:test-description}
Let $\{Z_i\}_{i = 1}^n$ be i.i.d.\ with $Z_i$ distributed according to $P \in \mathbf P$. Recall that the null hypothesis in \eqref{eq:null-Ab} may be equivalently written as $\mathbf P_0 = \{P \in \mathbf P: (A(P), \beta(P)) \in \mathbf C_0\}$. Further recall from Lemma \ref{lem:closure}(iv) that $\mathbf C_0 \subseteq \bar{\mathbf C}_0$, where $\bar{\mathbf C}_0$ is defined in \eqref{eq:C0bar} and only involves optimization over compact sets; we therefore base the test on $\bar{\mathbf C}_0$. Towards that end, define $b_{z, t}(P) := A(P) z - \beta(P) t$ for $(z, t) \in \mathcal H$, and denote $M_{A(P)}$ by $M(P)$. Because the maximum in \eqref{eq:C0bar} is at least zero (by taking $y = 0$), we see $(A(P), \beta(P)) \in \bar{\mathbf C}_0$ if and only if
\begin{equation} \label{eq:null_final}
\min_{(z, t) \in \mathcal H} b_{z, t}(P)' M(P) y \leq 0 \text{ for all } y \in \mathbb R^p \text{ such that } \|y\|_1 \leq 1~.
\end{equation}
The inequalities in \eqref{eq:null_final} provide the basis for our one-sided test.

Our test uses a sample-splitting procedure in which one sample split is used to select a $y \in \mathbb R^p$ and a second sample split is used to test whether the inequality in \eqref{eq:null_final} holds at the selected $y$. We next present a high-level description of the test and defer the details of the construction of its specific components to Sections \ref{sec:properties} and \ref{sec:yhat}. To construct the test, we split the data into two samples $\{Z_i\}_{i \in I_{1, n}}$, $\{Z_i\}_{i \in I_{2, n}}$ of sizes $n_1$ and $n_2$, where $I_{1, n} \cup I_{2, n} = \{1, \ldots, n\}$ and $I_{1, n} \cap I_{2, n} = \emptyset$. In what follows, we always assume that $n_2 \to \infty$ as $n \to \infty$ and allow $n_1$ to be fixed. Throughout, we use the superscript $(k)$ to denote when a given quantity is a function of only the $k$th split. Using the first sample split $\{Z_i\}_{i \in I_{1, n}}$, we construct a vector $\hat y_n^{(1)} \in \mathbb R^p$ such that $\|\hat y_n^{(1)}\|_1 \leq 1$ which represents a direction in which the inequality in \eqref{eq:null_final} appears to be violated; we discuss how to construct such a vector in Section \ref{sec:yhat}.

Let $L$ be a matrix whose columns form an orthonormal basis of $\lin(\mathcal X)$, denoting its number of columns by $\ell$, and define $A_0(P) := A(P) L$ (note from \eqref{eq:lineality} that $\lin(\mathcal X)$ equals the common null space of $A^{\mathsf{k},\mathsf{e}}$ and $A^{\mathsf{k},\mathsf{i}}$, and an orthonormal basis of it can be computed by QR decomposition). For $k = 1, 2$, let $\hat A_n^{(k)}$ and $\hat\beta_n^{(k)}$ denote estimators of $A(P)$ and $\beta(P)$ computed in the $k$-th sample split, and define
\begin{align*}
\hat A_{0, n}^{(k)} & := \hat A_n^{(k)} L \\
\hat b_{z, t, n}^{(k)} & := \hat A_n^{(k)} z - \hat\beta_n^{(k)} t \\
\hat M_{n}^{(k)} & := \mathbf I_p - \hat A_{0, n}^{(k)} (\hat A_{0, n}^{(k)})^\dagger~.
\end{align*}
Here, $\dagger$ denotes the Moore--Penrose pseudoinverse. When $M(P)$ is known, we take $\hat M_n^{(k)} = M(P)$ on both splits. This includes the case $\lin(\mathcal X) = \{0\}$, for which $M(P) = \mathbf I_p$. 

Next, given the estimators $\hat b_{z, t, n}^{(2)}$ and $\hat M_{n}^{(2)}$ for $b_{z, t}(P)$ and $M(P)$ computed in the second sample split $\{Z_i\}_{i \in I_{2, n}}$, we define the test statistic
\begin{equation} \label{eq:test_stat}
T_n := \min_{(z, t) \in \mathcal H} \frac{\sqrt{n_2} (\hat b_{z, t, n}^{(2)})' \hat M_{n}^{(2)} \hat y_n^{(1)}}{\hat\sigma_{z, t, n}^{(2)}(\hat y_n^{(1)})}~,
\end{equation}
where $\hat\sigma_{z, t, n}^{(2)}(y)$ is a truncated estimator for the asymptotic standard deviation of $\sqrt{n_2} (\hat b_{z, t, n}^{(2)})' \hat M_{n}^{(2)} y$ specified in \eqref{eq:sigma_hat} below. Finally, we set
\begin{equation} \label{eq:test}
\phi_n := I \{T_n > z_{1 - \alpha}\}~,
\end{equation}
where, for $0 < \alpha < 1/2$, $z_{1 - \alpha}$ denotes the $1 - \alpha$ quantile of a standard normal distribution. If $\hat y_n^{(1)} = 0$, then we set $T_n = 0$ and do not reject. Appendix \ref{sec:LP} shows how to compute the test in \eqref{eq:test} by evaluating a finite set of inequalities at the vertices of $\mathcal H$.

\begin{remark} \label{rem:div0}
It may occur that for some $(z, t) \in \mathcal H$ the asymptotic variance of $\sqrt{n_2} (\hat b_{z, t, n}^{(2)})' \hat M_{n}^{(2)} \hat y_n^{(1)}$ is zero. To avoid degeneracy of the corresponding standard error, we define $\hat\sigma_{z, t, n}^{(2)}(\hat y_n^{(1)})$ using a small truncation, as described in Section \ref{sec:properties}. The restriction $\alpha < 1/2$ ensures that if the numerator of \eqref{eq:test_stat} is nonpositive for any $(z,t) \in \mathcal H$, then we fail to reject for any choice of truncation.
\end{remark}

\begin{remark}\label{rem:LP}
Although the vertices of $\mathcal{H}$ are often known (for example, if $\mathcal{X} = \mathbb{R}^d_{+}$, then $\mathcal{H}$ is the simplex in $\mathbb{R}^{d+1}$) or can be explicitly enumerated as described in Appendix \ref{sec:LP} (which is feasible in all of the simulation designs considered in this paper), in some applications such an enumeration could be computationally prohibitive. In these cases, an additional benefit of our formulation \eqref{eq:C0bar} is that it permits an implementation based on linear programming: for the selected direction $\hat{y}^{(1)}_n$, let $\bar \sigma_{z, t, n}(\hat{y}^{(1)}_n)$ be an upper bound on $\hat\sigma^{(2)}_{z,t,n}(\hat{y}^{(1)}_n)$ that is affine over $(z, t) \in \mathcal{H}$. Then, a test that rejects whenever
\[\min_{(z,t)\in\mathcal{H}}\left\{\sqrt{n_2} (\hat b_{z, t, n}^{(2)})' \hat M_{n}^{(2)} \hat{y}^{(1)}_n - z_{1 - \alpha}\bar \sigma_{z, t, n}(\hat{y}^{(1)}_n)\right\} > 0 \]
can be shown to be uniformly valid and can be computed using linear programming. 
\end{remark}

\subsection{Properties of the Test} \label{sec:properties}
In this section, we present results establishing the asymptotic validity of our test and a detailed construction of the standard deviation in \eqref{eq:test_stat}. When stating our assumptions, we will suppress the superscript $(k)$, with the understanding that all assumptions stated for the entire sample will also hold when applied to the sample split $\{Z_i\}_{i \in I_{2, n}}$, under suitable scaling. Throughout the section, the regularity conditions simplify substantially whenever $M(P)$ is known, so we are careful to explicitly state the relevant differences whenever they arise. All stochastic orders below are uniform over $P \in \mathbf P_0$ and deterministic $(z, t) \in \mathcal H$. In particular,
$R_n(P; z, t) = o_P(1)$ uniformly means that for every $\epsilon > 0$,
\[ \sup_{P \in \mathbf P_0} \sup_{(z, t) \in \mathcal H} P \{|R_n(P; z, t)| > \epsilon\} \to 0~.\]

Our first assumption imposes conditions on the estimators $\hat A_{0, n}$ for $A_0(P)$ and $\hat b_{z, t, n}$ for $b_{z, t}(P)$ with $(z, t) \in \mathcal H$. In its statement, $\|\cdot\|_2$ denotes the Euclidean norm, $\|\cdot\|_{2, 2}$ denotes the operator norm of a matrix when the domain and range are endowed with $\|\cdot\|_2$, and $a \vee b = \max(a, b)$ for any $a, b \in \mathbb R$.

\begin{assumption} \label{ass:prelim-estimators}
Let $\{Z_i\}_{i = 1}^n$ be i.i.d.\ with marginal distribution $P \in \mathbf P$. There are $\Psi_A(Z_i, P) \in \mathbb R^{p \times d}$ with $E_P[\Psi_A(Z_i, P)] = 0$ and $\varphi_\beta(Z_i, P) \in \mathbb R^p$ with $E_P[\varphi_\beta(Z_i, P)] = 0$, such that for
\begin{align*}
\Psi_0(Z_i, P) & := \Psi_A(Z_i, P) L \\
\varphi_{z, t}(Z_i, P) & := \Psi_A(Z_i, P) z - t \varphi_\beta(Z_i, P)~,
\end{align*}
the following conditions hold:
\begin{packed_enum}
\item[(a)] There is a deterministic $a_n \geq 0$ for which $a_n / \sqrt n \to 0$ such that uniformly in $P \in \mathbf P_0$ and $(z, t) \in \mathcal H$,
\begin{align}
\label{eq:A0_IF} \Big \| \sqrt n(\hat A_{0, n} - A_0(P)) - \frac{1}{\sqrt n} \sum_{1 \leq i \leq n} \Psi_0(Z_i, P) \Big \|_{2, 2} & = O_P(a_n / \sqrt n) \\
\label{eq:B_IF} \Big \| \sqrt n(\hat b_{z, t, n} - b_{z, t}(P)) - \frac{1}{\sqrt n} \sum_{1 \leq i \leq n} \varphi_{z, t}(Z_i, P) \Big \|_2 & = O_P(a_n / \sqrt n)~.
\end{align}
\item[(b)] For each $p \geq 1$, there are $1 \leq K_{0, p}, K_{1, p} < \infty$ such that $\|\Psi_0(Z_i, P)\|_{2, 2} \leq K_{0, p}$ with probability one and
\begin{align*}
& \sup_{P \in \mathbf P_0} \Big ( \|E_P[\Psi_0(Z_i, P) \Psi_0(Z_i, P)']\|_{2, 2} \vee \|E_P[\Psi_0(Z_i, P)' \Psi_0(Z_i, P)]\|_{2, 2} \\
& \hspace{6cm} \vee \sup_{(z, t) \in \mathcal H} E_P[\varphi_{z, t}(Z_i, P)' \varphi_{z, t}(Z_i, P)] \Big ) \leq K_{1, p}~.
\end{align*}
\end{packed_enum}
\end{assumption}
\noindent Assumption \ref{ass:prelim-estimators}(a) requires our estimators for $A_0(P)$ and $b_{z, t}(P)$ to be asymptotically linear with influence functions whose moments are disciplined by Assumption \ref{ass:prelim-estimators}(b). Assumption \ref{ass:prelim-estimators}(a) is automatically satisfied with $a_n = 0$ whenever the entries of $A(P)$ and $\beta(P)$ are expectations and their estimators are the corresponding sample means. Note that when $M(P)$ is known, only the conditions on $\hat b_{z,t,n}$ and its influence function $\varphi_{z,t}$ are needed in Assumption \ref{ass:prelim-estimators}. Later, to construct the standard error of our test statistic, we also require estimators $\hat{\Psi}_A(Z_i)$, $\hat{\varphi}_{\beta}(Z_i)$ of $\Psi_A(Z_i, P)$ and $\varphi_\beta(Z_i, P)$.

Our second assumption imposes boundedness and non-degeneracy conditions on $A_0(P)$ and $b_{z, t}(P)$. In its statement, $\underline s(A_0(P))$ denotes the minimum singular value of $A_0(P)$.
\begin{assumption} \label{ass:singular-value}
$A_0(P)$ and $b_{z, t}(P)$ are such that
\begin{packed_enum}
\item[(a)] If $M(P)$ must be estimated, then $A_0(P)$ has full column rank for every $P \in \mathbf P_0$, and $\inf_{P \in \mathbf P_0} \underline s(A_0(P)) \geq \underline s > 0$ for some $\underline s$ not depending on $p$.
\item[(b)] $\sup_{P \in \mathbf P_0} \sup_{(z, t) \in \mathcal H} \|b_{z, t}(P)\|_2 \leq K_{2, p}$ for some $K_{2, p}$ satisfying $1 \leq K_{2, p} < \infty$.
\end{packed_enum}
\end{assumption}
\noindent When $M(P)$ is estimated, Assumption \ref{ass:singular-value}(a) ensures that $A_0(P)' A_0(P)$ is bounded away from degeneracy over $P \in \mathbf P_0$. This rules out the set of rank-deficient matrices $\mathbf C^{\rm R D}$ in our characterization of the closure of the null hypothesis, as defined in Lemma \ref{lem:closure}. If $M(P)$ is known (as is the case when, for instance, $\lin(\mathcal X) = \{0\}$), then our assumptions do not impose any rank restrictions on $A(P)$. Assumption \ref{ass:singular-value}(b) requires that the Euclidean norm of $b_{z, t}(P)$ be bounded uniformly in $P \in \mathbf P_0$ and $(z, t) \in \mathcal H$, with the bound possibly depending on $p$.

Even when $M(P)$ is estimated, Assumption \ref{ass:singular-value}(a) could be dropped if we let $n_1$ grow sufficiently quickly that $\hat A_{0,n}^{(1)}$ is uniformly consistent for $A_0(P)$ in operator norm over $P\in\mathbf P_0$, and apply the test $\phi_n$ in \eqref{eq:test} only when $\underline s(\hat A_{0,n}^{(1)})>\tau$ for some prespecified $\tau>0$ (and do not reject otherwise). We emphasize, moreover, that even if we maintain Assumption \ref{ass:singular-value}(a), our assumptions impose no requirement on the rank of $A(P)$ restricted to $\lin(\mathcal X)^\perp$. In contrast, the assumptions underlying uniform validity of the procedures in \cite{cox2025testing} and \cite{goff2026testing} restrict the rank of $A(P)$ on $\lin(\mathcal X)$ and $\lin(\mathcal X)^\perp$. 

Assumptions \ref{ass:prelim-estimators} and \ref{ass:singular-value} are instrumental in obtaining an asymptotic expansion for our test statistic.
In particular, when $M(P)$ is estimated, we will show that for every $(z, t) \in \mathcal H$ the vector $\sqrt n(\hat b_{z, t, n}' \hat M_{n} - b_{z, t}(P)' M(P))'$ is asymptotically linear with influence function
\begin{equation} \label{eq:xi}
\xi_{z, t}(Z_i, P) := M(P) \varphi_{z, t}(Z_i, P) - M(P) \Psi_0(Z_i, P) A_0^\dagger(P) b_{z, t}(P) - A_0^\dagger(P)' \Psi_0(Z_i, P)' M(P) b_{z, t}(P)~.
\end{equation}
When $M(P)$ is known, the influence function simplifies to $\xi_{z, t}(Z_i, P) := M(P) \varphi_{z, t}(Z_i, P)$.

Our third assumption imposes moment restrictions on the influence function $\xi_{z, t}(Z_i, P)$.
\begin{assumption} \label{ass:bdd-moments-IF}
There is a constant $K_\xi < \infty$ not depending on $p$ such that
\[ \sup_{P \in \mathbf P_0} \sup_{\|y\|_1 \leq 1} \sup_{(z, t) \in \mathcal H} E_P[|\xi_{z, t}(Z_i, P)' y|^3]^{1/3} \leq K_\xi~. \]
\end{assumption}
\noindent Assumption \ref{ass:bdd-moments-IF} ensures that we are able to couple the normalized sum of influence functions for our test statistic to a Gaussian random variable uniformly in $P \in \mathbf P_0$ and $(z, t) \in \mathcal H$. The requirement of Assumption \ref{ass:bdd-moments-IF} is satisfied, for example, if the third absolute moments of the entries of $\xi_{z, t}(Z_i, P)$ are uniformly bounded across $(z, t) \in \mathcal H$ and $P \in \mathbf P_0$.

We now turn to the construction of the estimator $\hat\sigma_{z,t,n}(y)$ employed in our test statistic. To estimate the influence function of $\hat b_{z,t,n}$ for every $(z,t) \in \mathcal H$, we compute estimated influence functions $\hat\Psi_{A, i} := \hat \Psi_A(Z_i)$ and $\hat\varphi_{\beta,i} := \hat \varphi_{\beta}(Z_i)$ and set $\hat\varphi_{z, t, i}:=\hat\Psi_{A,i} z - t \hat\varphi_{\beta,i}$. 

We first begin with the case where $M(P)$ is estimated, where we further define $\hat\Psi_{0,i}:=\hat\Psi_{A,i}L$, the estimated influence function of $\hat A_{0,n}=\hat A_nL$. Let $\mvec$ denote the vectorization operator and $\otimes$ denote the Kronecker product \citep[see, for instance, Chapter 2 of][]{magnus2019matrix}. Define the asymptotic covariance matrix
\[ V_{z, t}(P) := \var_P \bigg [ \begin{pmatrix} \mvec \Psi_0(Z_i, P) \\ \varphi_{z, t}(Z_i, P) \end{pmatrix} \bigg ] \]
and let $\hat V_{z, t, n}$ denote the sample covariance matrix of the corresponding estimated influence functions. For each fixed $y \in \mathbb R^p$ and $(z, t) \in \mathcal H$, we apply the Delta method to the function $(A_0(P), b_{z, t}(P)) \mapsto b_{z, t}(P)' M(P) y$. Accordingly, defining the gradient
\[
D_{z, t}(P; y) := \begin{pmatrix} - \big ( A_0^\dagger(P) y \otimes M(P) b_{z, t}(P) + A_0^\dagger(P) b_{z, t}(P) \otimes M(P) y \big ) \\ M(P) y \end{pmatrix}~,
\]
we can show that the asymptotic variance of $\sqrt n \hat b_{z, t, n}' \hat M_{n} y$ equals $\sigma_{z, t}^2(P; y) := D_{z, t}(P; y)' V_{z, t}(P) \allowbreak D_{z, t}(P; y)$.
By analogy, we estimate $D_{z, t}(P; y)$ by setting
\[ \hat D_{z, t, n}(y) := \begin{pmatrix} - \big ( \hat A_{0, n}^\dagger y \otimes \hat M_{n} \hat b_{z, t, n} + \hat A_{0, n}^\dagger \hat b_{z, t, n} \otimes \hat M_{n} y \big ) \\ \hat M_{n} y \end{pmatrix}~. \]
When $M(P)$ is known, we can omit the first block of the covariance matrix, so that $V_{z, t}(P) = \var_P[\varphi_{z, t}(Z_i, P)]$, $D_{z, t}(P; y) = \hat D_{z, t, n}(y) = M(P) y$. In either case, we set 
\begin{equation} \label{eq:sigma_hat}
\hat\sigma_{z, t, n}^2(y) = \big \{ \hat D_{z, t, n}(y)' \hat V_{z, t, n} \hat D_{z, t, n}(y) \big \} \vee \underline\sigma^2~,
\end{equation}
where, as previously discussed in Remark \ref{rem:div0}, $\underline\sigma > 0$ is a small positive constant which we include to address potential (near) degeneracies in $\sigma_{z, t}^2(P; y)$.

Our fourth assumption imposes that our estimator $\hat V_{z, t, n}$ is suitably uniformly consistent.
\begin{assumption} \label{ass:variance-estimator}
Uniformly in $P\in\mathbf P_0$ and $(z,t)\in\mathcal H$,
\[
\|\hat V_{z,t,n}-V_{z,t}(P)\|_{2,2}
=
\begin{cases}
o_P(K_{2,p}^{-2}), & \text{if $M(P)$ is estimated},\\
o_P(1), & \text{if $M(P)$ is known}.
\end{cases}
\]
\end{assumption}
\noindent Assumption \ref{ass:variance-estimator} enables us to show that the standard errors $\hat\sigma_{z, t, n}(y)$ are suitably uniformly consistent in both $P \in \mathbf P_0$ and $(z, t) \in \mathcal H$.
When $M(P)$ is estimated, it requires that $\hat V_{z, t, n}$ converge to $V_{z, t}(P)$ at a rate faster than $K_{2, p}^{-2}$, where $K_{2, p}$ depends on $p$ and is specified in Assumption \ref{ass:singular-value}(b). We state Assumption \ref{ass:variance-estimator} as a high-level condition because the structure of $V_{z, t}(P)$ is dictated by the influence functions of our estimators (as introduced in Assumption \ref{ass:prelim-estimators}).
In applications for which our estimators are sample means, sufficient conditions for Assumption \ref{ass:variance-estimator} are readily available from the literature; see, e.g., Chapter 6 in \cite{wainwright2019high}.

Our final assumption imposes conditions on the rates of convergence of our moment and error bounds.
\begin{assumption} \label{ass:rates}
If $M(P)$ must be estimated, the following rate restrictions hold as $n \to \infty$:
\begin{packed_enum}
\item[(a)] $K_{2, p}(K_{0, p} \vee K_{1, p}) \log(1 + p) = o(\sqrt{n_2})$.
\item[(b)] $K_{2, p} a_{n_2} = o(\sqrt{n_2})$.
\end{packed_enum}
\end{assumption}

When $M(P)$ is known, Assumption \ref{ass:prelim-estimators}(a) already ensures that $a_{n_2}=o(\sqrt{n_2})$, and no additional rate restriction is needed in Assumption \ref{ass:rates}. In Remark \ref{rem:rates} below, we discuss high-level sufficient conditions for Assumption \ref{ass:rates}.

Our next result establishes that our proposed test is uniformly consistent in level over $\mathbf P_0$.
\begin{theorem} \label{thm:test}
Suppose Assumptions \ref{ass:prelim-estimators}--\ref{ass:rates} hold. Then, for $0 < \alpha < 1/2$,
\[ \limsup_{n \to \infty} \sup_{P \in \mathbf P_0} E_P[\phi_n] \leq \alpha~. \]
\end{theorem}

\begin{remark} \label{rem:rates}
While the constants specified by Assumptions \ref{ass:prelim-estimators} and \ref{ass:singular-value} are application specific, in many instances we can expect the bounds $K_{0, p} \vee K_{1, p} \lesssim p$, $K_{2, p} \lesssim \sqrt{p}$, and $a_n \lesssim p$ (or $a_n = 0$) to hold.
Under these conditions, when $M(P)$ must be estimated, Assumption \ref{ass:rates} holds provided that $p^3 = o(n_2)$ (up to logs), while Assumption \ref{ass:variance-estimator} can also be shown to hold when $p^3 = o(n_2)$ (up to logs) under suitable moment restrictions.
These rates simplify when $M(P)$ is known. In this case, under suitable moment restrictions it suffices for the conclusion of Theorem \ref{thm:test} that $p = o(n_2)$ and $a_{n_2} = o(\sqrt{n_2})$, with the latter requirement automatically holding when there is no linearization error or being implied by $p^2 = o(n_2)$ when $a_n \lesssim p$. When $M(P)$ is estimated, Assumption \ref{ass:singular-value}(a) requires $\ell = \dim(\lin(\mathcal X))\leq p$. However, the theorem imposes no separate rate restriction on $d - \ell$, which may grow even with $p$ fixed, provided the moment bounds and estimation rates in the assumptions continue to hold uniformly over $P \in \mathbf P_0$ and $(z, t) \in \mathcal H$. 
\end{remark}

\begin{remark}\label{rem:parametrization}
In light of the theoretical results presented in this section and our simulation experiments, we provide some recommendations on the choice of the initial parametrization in \eqref{eq:null-Ab} that may improve the power of the resulting test. First, in our experience it is helpful to specify $\mathcal X$ to incorporate all available known linear restrictions on $x$. Where possible, we also recommend choosing a parametrization for which $M(P)$ is known, thereby avoiding the estimation error associated with estimating this projection. Finally, in some cases, a careful choice of the initial parametrization of \eqref{eq:null-Ab} can help exploit model-specific information about the coefficients of $A(P)$; Appendix \ref{sec:coefficient-preprocessing} states conditions under which we can use known restrictions on $A$ to ``pre-process'' the system $Ax=b$, $x\in\mathcal X$ into an equivalent system whose variables range over a known bounded set.
\end{remark}

\subsection{Procedure for selecting \texorpdfstring{$\hat y_n^{(1)}$}{the direction}} \label{sec:yhat}
In this section, we describe a procedure for selecting $\hat y_n^{(1)}$. Given suitable estimators $\hat b_{z, t, n}^{(1)}$ for each $(z, t) \in \mathcal H$ and $\hat M_{n}^{(1)}$, we let $\tilde y_n^{(1)}$ denote a solution to
\begin{equation} \label{eq:yhat-problem}
\max_{y \in \mathbb R^p: \|y\|_2 \leq 1} \min_{(z, t) \in \mathcal H} (\hat b_{z, t, n}^{(1)})' \hat M_{n}^{(1)} y~.
\end{equation}
A solution to \eqref{eq:yhat-problem} always exists because of the theorem of the maximum and the compactness of $\{y \in \mathbb R^p: \|y\|_2 \leq 1\}$. In Appendix \ref{sec:LP}, we explain how to compute a solution to \eqref{eq:yhat-problem} by quadratic programming.
To guarantee $\|\hat y_n^{(1)}\|_1 \leq 1$, we define
\begin{equation} \label{eq:yhat}
\hat y_n^{(1)} :=
\begin{cases}
\tilde y_n^{(1)} / \|\tilde y_n^{(1)}\|_1, &\text{ if } \tilde y_n^{(1)} \neq 0 \\
0, &\text{ if } \tilde y_n^{(1)} = 0~.
\end{cases}
\end{equation}

\section{Simulations} \label{sec:sims}
We now evaluate the finite-sample performance of our test and compare it to that of the generalized conditional chi-squared procedure of \cite{cox2025testing} and the bootstrap procedure of \cite{goff2026testing} in four designs. The first two designs are taken directly from the simulation exercises in \cite{goff2026testing} and \cite{cox2025testing}. The third design is a two-budget revealed preference model based on Example \ref{eg:varian}. The fourth design is a nonparametric instrumental variables model calibrated to the dataset in \cite{card1995geographic}, along with a perturbed version of this design which is intended to evaluate uniform size control. Appendix \ref{sec:additional-sims} reports three additional designs which we briefly comment on at the end of Section \ref{sec:card-sims}.

All tests have nominal significance level $5\%$. Our procedure uses $n_1 = \lfloor n / 4 \rfloor$ observations in the first sample split to select the direction $\hat y_n^{(1)}$ and the remaining $n_2 = n - n_1$ observations to test the null hypothesis. The procedures in \cite{cox2025testing} and \cite{goff2026testing} use the full sample, so the power comparisons incorporate the cost of sample splitting for our procedure.

For each design, we plot the rejection rates over grids of hypothesized values of the parameter of interest or, for the revealed preference example in Section \ref{sec:varian-sims}, the population expenditure gap. Rejection rates in this section are computed over $1{,}000$ Monte Carlo replications per grid point. The shaded region in each plot marks the values for which the null hypothesis holds.

\subsection{Mogstad--Santos--Torgovitsky Design} \label{sec:goff}
We first revisit a simple instance of Example \ref{eg:mst} proposed by \citet[Section 5.1]{goff2026testing}, which is loosely based on the application in \cite{mogstadsantostorgovitsky2017nwp} to the bed net data used by \cite{dupas2014e}. The outcome $Y$ is a binary indicator of using a new type of anti-malarial bed net, the treatment $D$ is an indicator for purchasing the bed net, and the instrument $Z$ is a randomly assigned price for the bed net.
The context implies that $Y(0) = 0$, so that the average treatment effect equals $E[Y(1)]$. Suppose that the instrument is binary and that there are no covariates. The researcher assumes that $E[Y(1) \vert U = u] = \theta_{0} + \theta_{1}u + \theta_2u^2$ is a weakly decreasing, quadratic function of $u$, which must be contained within $[0,1]$ because $Y(1)$ is binary. These shape restrictions are equivalently represented by the following linear inequalities:
\begin{equation} \label{eq:mtr-goff-mbakop-constraints}
\begin{aligned}
0 \leq \theta_0 & \leq 1 \\
\theta_0 + \theta_1 + \theta_2 & \geq 0 \\
\theta_1 & \leq 0 \\
\theta_1 + 2 \theta_2 & \leq 0~.
\end{aligned}
\end{equation}
Under the assumptions in Example \ref{eg:mst}, the researcher matches the moments $E_P[Y D \mid Z = z]$ for $z \in \{0,1\}$, creating two restrictions:
\[ p(z) \theta_0 + \frac{p(z)^2}{2} \theta_1 + \frac{p(z)^3}{3} \theta_2 = E_P[Y D \mid Z = z] \text{ for } z \in \{0, 1\}~. \]
The null hypothesis of interest is $H_{0}: E[Y(1)] = \tau_{0}$, where
\[E[Y(1)] = \int_0^1 (\theta_0 + \theta_1u + \theta_2u^2) \, d u = \theta_0 + \frac{1}{2} \theta_1 + \frac{1}{3} \theta_2~.\]
This null hypothesis fits into \eqref{eq:null-Ab} with $x = (\theta_0, \theta_1, \theta_2)'$, $\mathcal X = \mathcal X(\tau_0) = \{x: \eqref{eq:mtr-goff-mbakop-constraints} \text{ holds and } \allowbreak \theta_0 + \frac{1}{2} \theta_1 + \frac{1}{3} \theta_2 = \tau_0\}$,
\[ A(P) = \begin{pmatrix}
    p(0) & p(0)^2 / 2 & p(0)^3 / 3 \\
    p(1) & p(1)^2 / 2 & p(1)^3 / 3
\end{pmatrix}~, \qquad \text{and} \qquad \beta(P) = \begin{pmatrix}
    E_P[Y D \mid Z = 0] \\
    E_P[Y D \mid Z = 1]
\end{pmatrix}~. \]
Since $\mathcal X$ is compact, $\lin(\mathcal X) = \{0\}$ and $M(P) = \mathbf I_2$. Therefore, we do not employ the preprocessing step in Appendix \ref{sec:coefficient-preprocessing}.

The simulation design is specified as $P\{Z = 1\} = 0.5 = P\{Z = 0\}$, $p(0) = 1/3$, $p(1) = 2/3$, $Y(0) = 0$, $Y(1) = I \{V \le \theta_{0} + \theta_{1} U + \theta_{2} U^2\}$, where $V \mid Z,U \sim U[0,1]$ and $(\theta_{0},\theta_{1},\theta_{2}) = (1, -1, 0.5)$. The identified set for $E[Y(1)]$ is $[19/33, 2/3]$. $\hat A_n$ and $\hat \beta_n$ are given by sample analogs.

\begin{figure}[t!]
\centering
\includegraphics{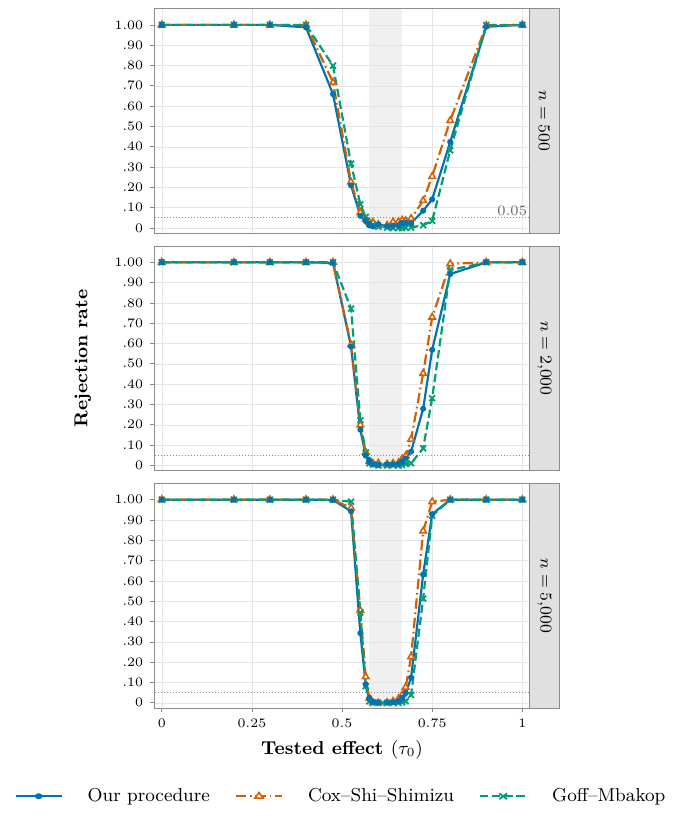}
\caption{Mogstad--Santos--Torgovitsky design: rejection rates over $1{,}000$ replications for sample sizes $n \in \{500, 2{,}000, 5{,}000\}$. Rows correspond to sample sizes. The horizontal axis shows the hypothesized average treatment effect $\tau_0$. The shaded region marks the identified set $[19/33, 2/3]$. The solid blue line with filled circles represents our procedure, the dash-dotted orange line with triangles represents Cox--Shi--Shimizu, and the dashed green line with crosses represents Goff--Mbakop. The dotted horizontal line marks the nominal rejection rate $0.05$.}
\label{fig:GM}
\end{figure}

Figure \ref{fig:GM} reports the results for $n \in \{500, 2000, 5000\}$. All three procedures reject less than $5\%$ of the time for every grid point in the identified set. At $n = 500$, our rejection rate is, on average across the grid points outside of the identified set, $2.68$ percentage points lower than that of \cite{cox2025testing} and $0.51$ percentage points lower than that of \cite{goff2026testing}, although the differences narrow as $n$ increases. All three procedures reject with high probability at distant values of $\tau_0$.

\subsection{Interval Outcome Instrumental Variables} \label{sec:interval-sims}
Next, we consider the interval outcome IV design in Section 5.2 in \cite{cox2025testing}, based on the application in \cite{gandhi2023estimating}. The model in \cite{gandhi2023estimating} describes demand for a product when observed market shares are noisy measures of purchase probabilities. Let $X_1$ denote an endogenous regressor, $X_2$ denote an exogenous regressor, $W$ denote a $d_W$-vector of exogenous controls, and $Z_e$ denote an excluded instrument. The latent market share is given by a logit specification 
\[
s^* = \frac{\exp(\theta_1X_1 + \theta_2X_2 + W'\gamma + \epsilon)}{1 + \exp(\theta_1X_1 + \theta_2X_2 + W'\gamma + \epsilon)}~.
\]The latent log odds thus satisfy
\begin{equation}\label{eq:int_model} 
Y^* = \log\left(\frac{s^*}{1 - s^*}\right) = \theta_1 X_1 + \theta_2 X_2 + W'\gamma + \epsilon~, \qquad E[\epsilon \mid X_2, W, Z_e] = 0~. 
\end{equation}
Let $Z$ be the $K$-vector of indicators for the support points of $(X_2, W, Z_e)$. We construct variables $Y^L$ and $Y^U$ generated from the observed market shares whose conditional means bound that of $Y^*$: following \cite{gandhi2023estimating}, if we observe market share $s_N \sim \text{Binomial}(N,s^*)/N$, where $N$ is the number of participants in the market, then under the simulation calibration considered in \cite{cox2025testing}, it can be shown that with $N = 100$, 
\[
Y^L := \log\!\left(\frac{s_N + 0.00125}{1 - s_N + 2/N}\right)
\qquad \text{and} \qquad
Y^U := \log\!\left(\frac{s_N + 2/N}{1 - s_N + 0.00125}\right)
\]
satisfy $E_P[Y^L|Z] \le E[Y^*|Z] \le E_P[Y^U|Z]$. The restrictions given by \eqref{eq:int_model} then imply
\[ E_P[Z Y^L] \leq E_P[Z X_1]\theta_1 + E_P[Z X_2]\theta_2 + E_P[Z W']\gamma \leq E_P[Z Y^U]~. \]
The null hypothesis of interest is $\theta_2 = \theta_{2, 0}$, and we treat $\delta := (\theta_1, \gamma')' \in \mathbb R^{d_W+1}$ as unrestricted nuisance parameters. This null hypothesis fits into \eqref{eq:null-Ab} with $x = (\theta_1, \gamma', \ell', u')'$, $\mathcal X = \mathbb R^{d_W+1} \times \mathbb R_+^{2K}$,
\[ A(P) = \begin{pmatrix}
E_P[ZX_1] & E_P[ZW'] & - \mathbf I_K & 0 \\
E_P[ZX_1] & E_P[ZW'] & 0 & \mathbf I_K
\end{pmatrix}~, \text{ and } \beta(P) = \beta(P, \theta_{2, 0}) = \begin{pmatrix}
E_P[ZY^L] - \theta_{2,0} E_P[ZX_2] \\
E_P[ZY^U] - \theta_{2,0} E_P[ZX_2]
\end{pmatrix}~, \]
where $\ell, u \in \mathbb R_+^K$ denote slack variables for the lower and upper bounds. We implement the procedures in \cite{cox2025testing} and \cite{goff2026testing} based on this parametrization. For our procedure, we follow the recommendations in Remark \ref{rem:parametrization} and exploit the fact that the lower and upper slack variables sum across instrument cells to $E_P[Y^U - Y^L]$, which is strictly positive by construction. We can thus re-parametrize the problem following Appendix \ref{sec:coefficient-preprocessing}.\footnote{We could also implement the procedures in \cite{cox2025testing} and \cite{goff2026testing} using the re-parametrization. However, for \cite{goff2026testing}, the original parametrization delivered higher average power in an independent comparison, and for \cite{cox2025testing}, our implementation encountered covariance singularities under the re-parametrization.}

The simulation design is specified as $W_1=1$, with $X_2$, $Z_e$, and the remaining components of $W$ being independent Bernoulli$(1/2)$ variables. Independently, we draw $V\sim N(0,1)$ and set $\epsilon = \min\{\max\{V,-4\},4\}$ and $X_1 = I \{Z_e + \epsilon / 2 > 0\}$, and set $(\theta_1, \theta_2, \gamma) = (-1, -1, 0)$. For $d_W \in \{1, 2, 3\}$, there are $K \in \{4, 8, 16\}$ instrument cells and, in the notation of \cite{cox2025testing}, $d_\delta := d_W + 1 \in \{2, 3, 4\}$ nuisance coefficients, and $d_C := 2 K \in \{8, 16, 32\}$ inequalities. The identified set for $\theta_2$ is approximately $[-1.0991,-0.8607]$ in every configuration.

\begin{figure}[t!]
\centering
\includegraphics{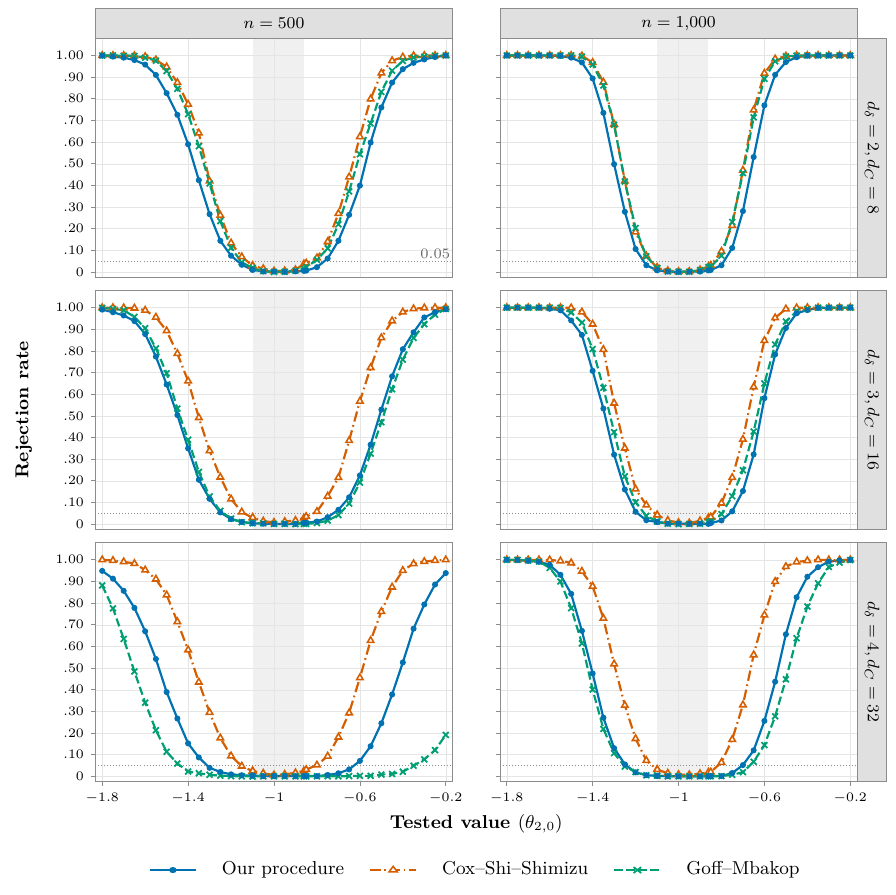}
\caption{Interval outcome IV design: rejection rates over $1{,}000$ replications for sample sizes $n \in \{500, 1{,}000\}$. Rows correspond to the numbers of nuisance coefficients $d_\delta$ and inequalities $d_C$, and columns correspond to sample sizes. The horizontal axis shows the hypothesized value $\theta_{2, 0}$. The shaded region marks the identified set, approximately $[-1.0991, -0.8607]$. The solid blue line with filled circles represents our procedure, the dash-dotted orange line with triangles represents Cox--Shi--Shimizu, and the dashed green line with crosses represents Goff--Mbakop. The dotted horizontal line marks the nominal rejection rate $0.05$.}
\label{fig:interval}
\end{figure}
 
Figure \ref{fig:interval} reports the results for $n\in\{500,1000\}$. All three procedures reject less than $5\%$ of the time for every grid point in the identified set. At $n=500$, our procedure's average rejection rate across grid points outside the identified set is $8.5$, $14.9$, and $23.9$ percentage points below that of \cite{cox2025testing} for $d_\delta=2,3,4$. These gaps narrow to $5.4$, $9.2$, and $16.5$ percentage points at $n=1{,}000$. Relative to \cite{goff2026testing}, our average rejection rate across grid points outside the identified set at $n=500$ is $5.7$ percentage points lower when $d_\delta=2$, but $0.3$ and $21.9$ percentage points higher when $d_\delta=3$ and $4$. At $n=1{,}000$, it is $5.0$ and $3.5$ percentage points lower when $d_\delta=2$ and $3$, and $4.6$ percentage points higher when $d_\delta=4$. At $n = 1000$ all three procedures reject with high probability for hypothesized values $\theta_{2, 0}$ far from the identified set.

\subsection{Two-Budget Revealed Preference Design} \label{sec:varian-sims}
We revisit the representative-consumer problem in Example \ref{eg:varian}. We construct a two-good design in which the first-period mean bundle is strictly affordable in period two.

Across designs, we vary $\delta = a(P)$, the difference between the cost of the second-period mean bundle at first-period prices and first-period expenditure. We set $p_1 = (1, 1)'$, $p_2 = (2, 1)'$, $q_1 = (2, 2)'$, and $q_2(\delta) = (3 - \delta, 1 + 2 \delta)'$. For each value of $\delta$, we draw an i.i.d.\ sample of $n$ households and observe each household in both periods. Household $i$'s purchases are given by
\begin{align*}
Q_{1i} & = q_1 + H_{1i} (1, -1)' \\
Q_{2i} & = q_2(\delta) + H_{2i} (1, -2)'~,
\end{align*}
where the $H_{ti}$ independently take the values $-0.2$ and $0.2$ with equal probability across households and periods. Prices and expenditures are fixed across designs: every household spends $4$ in period one and $7$ in period two. Alternatively, $H_{ti}$ can be interpreted as mean-zero measurement errors in observed quantities \citep[see][]{varian1985measurement}.
At second-period prices, $q_1$ costs $6$ and $q_2(\delta)$ costs $7$, so $b(P) = -1$. The bounds in Example \ref{eg:varian} thus reduce to $\lambda_2 \leq u \leq \lambda_1 \delta$. The null holds exactly when $\delta > 0$, in which case every feasible solution to the linear system satisfies $\lambda_1 / \lambda_2 \geq 1 / \delta$. Figure \ref{fig:budget} illustrates why the null fails at $\delta = 0$. In period two, choosing $q_2(0)$ when $q_1$ was strictly affordable reveals $q_2(0)$ is strictly preferred to $q_1$. Yet $q_2(0) = (3, 1)'$ lies on the first budget line, so choosing $q_1$ in period one reveals $q_1$ is weakly preferred to $q_2(0)$, a contradiction.
Meanwhile, we note that at $\delta = 0$, the vector $v = (0, 1, 0, 0, 0)'$ belongs to $\rec(\mathcal X) \setminus \lin(\mathcal X)$ and satisfies $A(P) v = 0$, so Lemma \ref{lem:closure}(ii) implies that the observed distribution $P$ at $\delta = 0$ belongs to $\bar{\mathbf P}_0 \setminus \mathbf P_0$, and in particular $\bar{\mathbf P}_0 \neq \mathbf P_0$.

\begin{figure}[t!]
  \centering
  \includegraphics{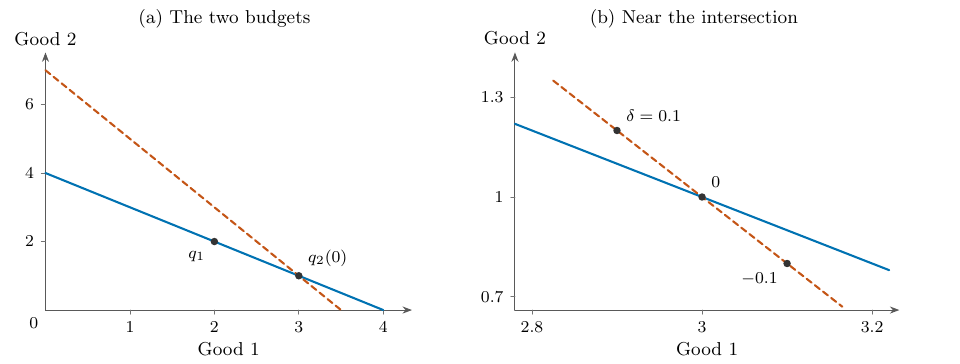}
  \caption{Two-budget revealed preference design: budget lines and mean bundles. The solid line shows the first budget, and the dashed line shows the second budget. Their intersection, $q_2(0) = (3, 1)'$, is the second-period mean bundle when $\delta = 0$. The right panel shows the second-period mean bundles for $\delta \in \{-0.1, 0, 0.1\}$.}
  \label{fig:budget}
\end{figure}

We implement all three procedures using the original parametrization in Example \ref{eg:varian}, estimating $a(P)$ and $b(P)$ by their sample analogs. We do not apply the preprocessing step in Appendix \ref{sec:coefficient-preprocessing} because one of the required assumptions fails.

\begin{figure}[t!]
  \centering
  \includegraphics{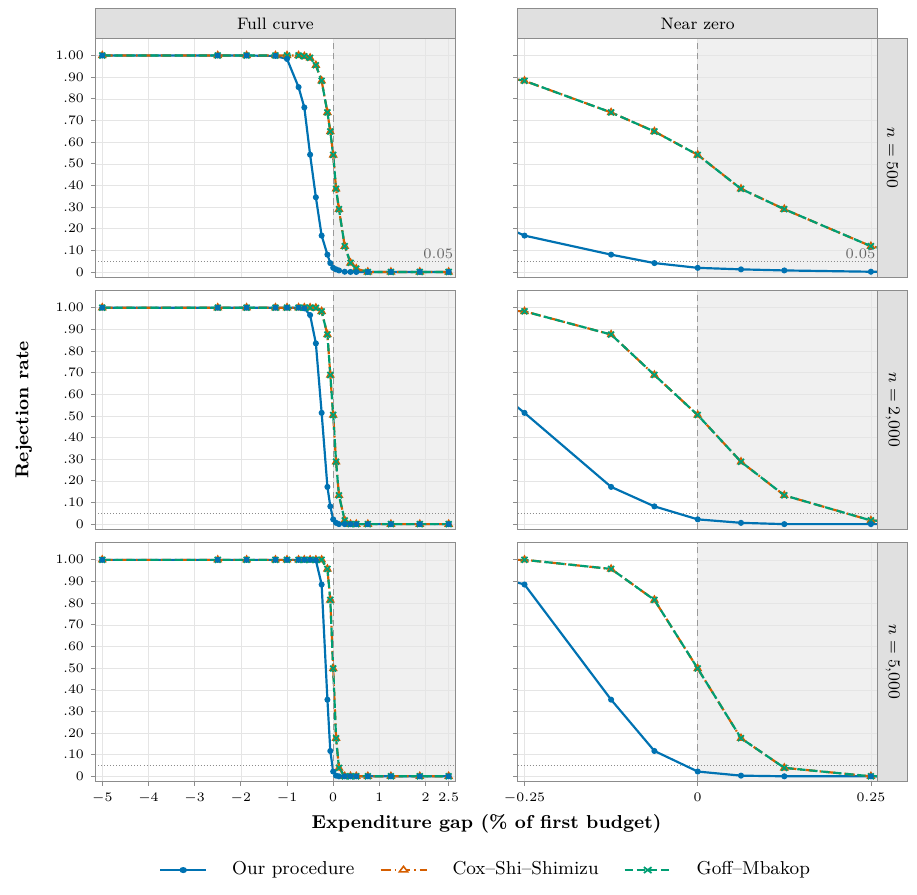}
  \caption{Two-budget revealed preference design: rejection rates over $1{,}000$ replications for sample sizes $n \in \{500, 2{,}000, 5{,}000\}$. Rows correspond to sample sizes. The left column shows the full curves, and the right column shows an enlarged view near $\delta = 0$. The horizontal axis shows the population expenditure gap as a percentage of the first budget ($100 \delta / 4$). The shaded region marks designs satisfying the null hypothesis ($\delta > 0$). The solid blue line with filled circles represents our procedure, the dash-dotted orange line with triangles represents Cox--Shi--Shimizu, and the dashed green line with crosses represents Goff--Mbakop. The dotted horizontal line marks the nominal rejection rate $0.05$. The Cox--Shi--Shimizu and Goff--Mbakop curves coincide. The Monte Carlo standard error is at most $1.6$ percentage points.}
  \label{fig:rates}
\end{figure}

Figure \ref{fig:rates} reports rejection rates for $n \in \{500, 2{,}000, 5{,}000\}$. Our procedure rejects less than $5\%$ of the time at every grid point with $\delta > 0$. Both comparison procedures overreject at the smallest strictly positive gap, $\delta = 0.0025$, rejecting $38.4\%$ of the time at $n = 500$, $28.8\%$ at $n = 2{,}000$, and $17.6\%$ at $n = 5{,}000$.
To understand why the other two procedures overreject, consider a sequence of null distributions with $\delta \downarrow 0$. The limiting inequalities require $1 \leq \lambda_2 \leq u \leq 0$, violating the feasibility requirement in Assumption 1(vi) of \cite{cox2025testing}. For \cite{goff2026testing}, the coefficient rows for $u \leq \delta \lambda_1$, $u \geq \lambda_2$, and $\lambda_2 \geq 1$ form a submatrix that is invertible when $\delta > 0$ but rank-deficient when $\delta = 0$, so the ``condition number'' in their Assumption 4.17(i) diverges, violating their Assumption 4.17(i). Moreover, since $\lambda_2 \geq 1$, every feasible solution satisfies $\lambda_1 \geq 1 / \delta$, violating the uniform boundedness requirement in Assumption 4.17(ii). 

\subsection{Nonparametric Instrumental Variables Calibrated to Card's Data} \label{sec:card-sims}
Finally, we calibrate the nonparametric instrumental variables model in Example \ref{eg:npiv} to the empirical distribution of schooling and wage data of \cite{card1995geographic}. The outcome $Y$ is negative log wages, and $W \in \{0, 1\}$ indicates proximity to a four-year college. $X$  takes values $1, \ldots, 6$ which index schooling categories in increasing order: at most nine years, ten or eleven years, twelve years, thirteen through fifteen years, sixteen years, and at least seventeen years. We impose that the structural function $g$ is weakly decreasing and the null hypothesis of interest is $H_0: L(g) = g(1) = L_0$. As a result, $\mathcal X = \mathcal{X}(L_0) = \{g: g(1) = L_0, \ g(1) \geq g(2) \geq \cdots \geq g(6)\}$.
Following the recommendations in Remark \ref{rem:parametrization}, here we exploit the fact that probabilities are non-negative and the restriction $P \{X = 6\} > 0$ to re-parametrize the problem following Appendix \ref{sec:coefficient-preprocessing}. In this section, we implement all three procedures using this parametrization.\footnote{Both comparison procedures can also be made to overreject under the original parametrization in Example \ref{eg:npiv}, using a different adversarial perturbation at Wasserstein distance 0.0496 from the empirical distribution. In this case, the procedures in \cite{cox2025testing} and \cite{goff2026testing} reject at the lower endpoint of the identified set with estimated probabilities 18.66 percent and 16.38 percent.}

We report simulation results using the original empirical distribution along with a perturbation of this distribution that assigns probability $w_i > 0$ to $(Y_i + b_i, X_i, W_i)$, where $b_i$ is an observation-specific outcome shift and $\sum_i w_i = 1$. We construct these weights and shifts using a neural generator, adapting the adversarial search developed in \cite{heydt2026gan}.\footnote{We thank Lasse von der Heydt for helping us implement his methodology.} The search adjusts the network parameters to minimize the transport distance from the empirical distribution, subject to fixed lower bounds on (approximations to) the rejection probabilities of \cite{cox2025testing} and \cite{goff2026testing} at the upper endpoint of the identified set. We measure transport distance from the empirical distribution by the minimum average transport cost, using the sum of absolute differences in outcome, schooling-bin mean, and instrument coordinates, each standardized using the empirical distribution. The Wasserstein distance is $0.1054$, and the largest change in a joint schooling-by-instrument cell probability is $0.49$ percentage points. The observation-specific outcome shifts correspond to multiplying wages by factors ranging from $0.8041$ to $1.4482$; the median absolute percentage change in wages across the original observations is $1.37\%$.\footnote{We also applied a similar numerical search to our procedure, with the objective of inflating its rejection probability under the same transport-distance budget. The search did not find evidence of size distortions: the largest rejection rate of our procedure at either endpoint was \(2.23\%\).}

\begin{figure}[t!]
\centering
\includegraphics{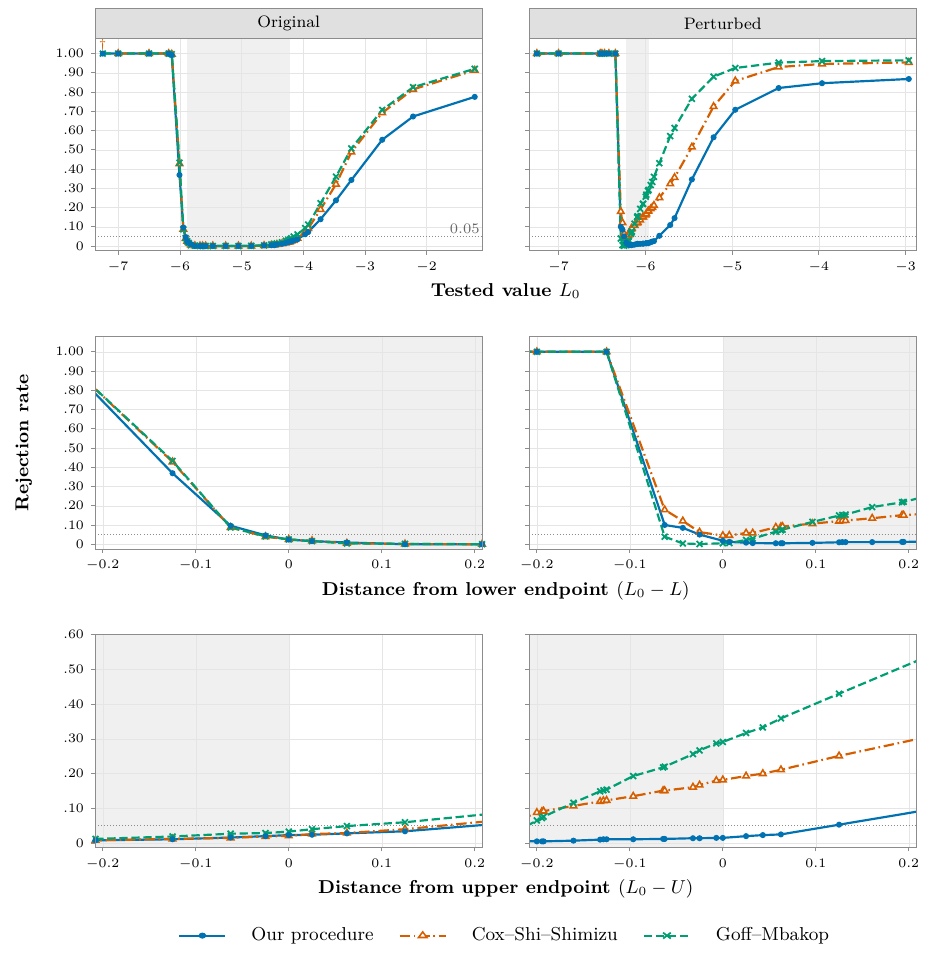}
\caption{Card-calibrated NPIV designs: rejection rates over $1{,}000$ replications for sample size $n = 3{,}010$. The left column shows the original calibration, and the right column shows the perturbed calibration. The horizontal axis shows the hypothesized value $L_0$ in the top row, $L_0 - L$ in the middle row, and $L_0 - U$ in the bottom row, where $L$ and $U$ are the lower and upper endpoints of the identified set. The shaded region marks the identified set. The top row uses different horizontal ranges across columns, and the bottom row uses an expanded rejection-rate scale. The solid blue line with filled circles represents our procedure, the dash-dotted orange line with triangles represents Cox--Shi--Shimizu, and the dashed green line with crosses represents Goff--Mbakop. The dotted horizontal line marks the nominal rejection rate $0.05$. The dagger marks $74$ undefined Cox--Shi--Shimizu evaluations at $L_0 = -7.25$ under the original calibration; the corresponding rejection rate is omitted.}
\label{fig:card}
\end{figure}

\begin{table}[ht!]
\centering
\begin{tabular}{llrrr}
\toprule
Calibration & Endpoint & Our procedure & Cox--Shi--Shimizu & Goff--Mbakop \\
\midrule
Original & Lower & 2.40 & 2.50 & 2.50 \\
 & Upper & 2.30 & 2.10 & 3.30 \\
\midrule
Perturbed & Lower & 1.80 & 4.50 & 0.50 \\
 & Upper & 1.50 & 18.20 & 29.10 \\
\bottomrule
\end{tabular}
\caption{Card-calibrated NPIV designs: rejection rates over $1{,}000$ replications for sample size $n = 3{,}010$. Rates are reported in percent at the lower and upper endpoints of each identified set. All three procedures use the same system in Example \ref{eg:pp-npiv} for each calibration.}
\label{tab:card-endpoints}
\end{table}

Figure \ref{fig:card} reports results for $n=3,010$ independent draws from each distribution. The left column uses the original empirical calibration and the right column uses the perturbed calibration. Under the original empirical calibration, all three procedures reject less than $5\%$ of the time at every grid point in the identified set. Our procedure's average rejection rate across grid points outside the identified set is $4.36$ percentage points lower than that of \cite{cox2025testing} and 5.3 percentage points lower than that of \cite{goff2026testing}. Under the perturbed calibration, our procedure still rejects less than $5\%$ at every grid point in the identified set. In contrast, both the procedures of \cite{goff2026testing} and \cite{cox2025testing} begin to overreject within the identified set for this design: Table \ref{tab:card-endpoints} reports the rejection rates for all three procedures at the lower and upper endpoints of the identified set. At the upper endpoint, the procedure in \cite{cox2025testing} rejects $18.2\%$ of the time and that in \cite{goff2026testing} rejects $29.1\%$, despite receiving the same preprocessed system as our procedure, which rejects $1.5\%$ of the time.

Comparing the original and perturbed calibrations helps explain these results. For \cite{goff2026testing}, the perturbation increases the ``condition number'' in their Assumption 4.17(i) by approximately thirty-fold at the upper endpoint. For \cite{cox2025testing}, the perturbation moves the design close to a distribution at which the stable-rank condition in their Assumption 2 fails, which substantially reduces the accuracy of the preliminary nuisance estimator used to construct the test. In Appendix \ref{sec:additional-sims}, we provide additional simulation evidence of overrejection of either procedure in other designs in which their conditions on ranks of matrices fail or are close to failing.

\clearpage
\bibliography{randomA}

\clearpage
\appendix

\section{Proofs of Main Results}

Throughout the appendix, for any vector $(v_1,\ldots, v_m) = v \in \mathbb R^m$ we set $\|v\|_q : = (\sum_i |v_i|^q)^{1/q}$ for any $1\leq q < \infty$, with the usual modification when $q = \infty$. 
For any $1\leq r,s \leq \infty$, the norms $\|\cdot\|_r$ and $\|\cdot\|_s$ induce the $(r, s)$-operator norm on a matrix $B$, which we denote by $\|B\|_{r, s} := \sup_{\|u\|_r \leq 1} \|B u\|_s$.

The proofs for the results in Section \ref{sec:farkas} rely on  concepts that we did not discuss in the main text for ease of exposition, but we now introduce.
For any polyhedron $\mathcal Q\subseteq \mathbb R^q$, its recession cone $\rec(\mathcal Q)$ is
\begin{equation*}%
\rec(\mathcal Q) := \{v\in \mathbb R^q : x + sv \in \mathcal Q \text{ for all } x\in \mathcal Q,~ s\geq 0\} ~;
\end{equation*}
i.e., $\rec(\mathcal Q)$ represents the directions along which we may move in $\mathcal Q$ without leaving $\mathcal Q$ \citep[pp. 61--62]{rockafellar1970convex}.
We also rely on the homogenization of a polyhedron, which is defined as 
\begin{equation}\label{def:homo}
\homog(\mathcal Q) := \{(z,t)\in \mathbb R^q\times \mathbb R : t > 0, z/t\in \mathcal Q\} \cup (\rec(\mathcal Q) \times \{0\}) ~.
\end{equation}
Homogenization turns a polyhedron into a cone by replacing the polyhedron at height $t = 1$ (note $\{z : (z,1)\in \hom(\mathcal Q)\} = \mathcal Q$) with all positive multiples of these points, and including its recession cone at $t = 0$ \citep[Section 1.5]{ziegler1995lectures}.
Finally, for any set $C\subseteq \mathbb R^q$, we set its orthogonal complement to be 
\begin{equation*}%
C^\perp := \{v \in \mathbb R^q : v^\prime c = 0 \text{ for all } c\in C\}~.
\end{equation*}

\subsection{Proof of Lemma \ref{lem:annihilator}}

We establish the result by deriving two characterizations of $\mathbf C_0$. First we show that $\mathbf C_0 = \mathbf C_1$ for
\begin{equation}\label{lem:annihilator1}
\mathbf C_1 := \{(A, b): M_A A x = M_A b \text{ for some } x \in \mathcal X \cap \lin(\mathcal X)^\perp\}~. 
\end{equation}
To this end, note that if $\lin(\mathcal X) = \{0\}$, then $M_A$ equals the identity matrix and $\lin(\mathcal X)^\perp = \mathbb R^d$, which immediately implies $\mathbf C_0 = \mathbf C_1$.
We therefore assume $\lin(\mathcal X) \neq \{0\}$ and let $L$ denote a $d\times \ell$ orthonormal matrix whose columns span $\lin(\mathcal X)$.
By Lemma \ref{lem:basic}(a), it then follows that every $x \in \mathcal X$ can be decomposed as $x = \tilde x + L w$ for some $\tilde x\in \mathcal X \cap \lin(\mathcal X)^\perp$ and $w \in \mathbb R^\ell$. 
As a result, if $(A, b) \in \mathbf C_0$, then $A L w + A \tilde x = b$ for some $w \in \mathbb R^\ell$ and $\tilde x \in \mathcal X \cap \lin(\mathcal X)^\perp$. 
Because $M_A A L = 0$ due to $M_A$ being a projection matrix onto $[A(\lin(\mathcal X))]^\perp$, it follows that $M_A A \tilde x = M_A b$, which establishes that $\mathbf C_0 \subseteq \mathbf C_1$. 
For the reverse inclusion, suppose there exists a $\tilde x \in \mathcal X \cap \lin(\mathcal X)^\perp$ such that $M_A A \tilde x = M_A b$. 
Then, it follows that $M_A(b - A \tilde x) = 0 $ and therefore that $(b-A\tilde x) \in A(\lin(\mathcal X))$.
In particular, by definition of $L$, there is a $w \in \mathbb R^\ell$ such that $A L w = b - A \tilde x$, which implies $A (\tilde x + L w) = b$. 
Moreover, since $L w \in \lin(\mathcal X)$ and $\tilde x\in \mathcal X$, it follows that $\tilde x + L w \in \mathcal X$ by definition of $\lin(\mathcal X)$. We can therefore conclude that $\mathbf C_1 \subseteq \mathbf C_0$, and hence that $\mathbf C_0 = \mathbf C_1$.

To establish the lemma, we next show the equality $\mathbf C_1 = \mathbf C_2$ for the set $\mathbf C_2$ given by
\begin{equation*}%
 \mathbf C_2 := \bigg \{ (A, b): \sup_{\|y\|_1 \leq 1} \inf_{x \in \mathcal X \cap \lin(\mathcal X)^\perp} y' M_A (A x - b) = 0 \bigg \}~. 
\end{equation*}
To this end, note that $\mathcal X \cap \lin(\mathcal X)^\perp$ is a nonempty polyhedron by Lemma \ref{lem:basic}(a).
Therefore, Theorem 19.3 in \cite{rockafellar1970convex} implies that for each $(A, b) \in \mathbb R^{p \times d} \times \mathbb R^p$, the set $\mathbf V(A, b) := \{M_A (A x - b): x \in \mathcal X \cap \lin(\mathcal X)^\perp\}$ is a nonempty polyhedron, and in particular is a closed and convex set. 
Next, suppose that $(A, b) \notin \mathbf C_1$, which implies that $0 \notin \mathbf V(A, b)$. 
Therefore, by the separating hyperplane theorem (see, e.g., Corollary 11.4.2 and Theorem 11.1(c) in \cite{rockafellar1970convex}), there exists a $ y_0 \in \mathbb R^p$ with $\|y_0\|_1 =1 $ such that
\begin{equation*}%
0 < \inf_{v \in \mathbf V(A,b)} y_0^\prime v =  \inf_{x \in \mathcal X \cap \lin(\mathcal X)^\perp} y_0' M_A (A x - b) ~. 
\end{equation*}
In particular, since $\|y_0\|_1 = 1$, it follows that if $(A,b)\notin \mathbf C_1$, then $(A,b) \notin \mathbf C_2$, and therefore we conclude $\mathbf C_2 \subseteq \mathbf C_1$.
For the reverse inclusion, we now suppose that $(A, b) \in \mathbf C_1$.
By definition of $\mathbf C_1$, there must be an $x_0 \in \mathcal X \cap \lin(\mathcal X)^\perp$ such that $M_A (A x_0 - b) = 0$. For every $y \in \mathbb R^p$, it then follows that
\begin{equation}\label{lem:annihilator4}
\inf_{x \in \mathcal X \cap \lin(\mathcal X)^\perp} y' M_A (A x - b) \leq y' M_A(A x_0 - b) = 0~. 
\end{equation}
On the other hand, the infimum in \eqref{lem:annihilator4} equals 0 when $y = 0$, and we therefore obtain that
\begin{equation*}%
\sup_{\|y\|_1 \leq 1} \inf_{x \in \mathcal X \cap \lin(\mathcal X)^\perp} y' M_A (A x - b) = 0~.
\end{equation*}
It therefore follows that $(A,b) \in \mathbf C_2$, which establishes that $\mathbf C_1 \subseteq \mathbf C_2$, and hence $\mathbf C_1 = \mathbf C_2$. \qed

\subsection{Proof of Lemma \ref{lem:normalization}}

We begin by defining the linear function $\rho(z,t) := t + {\bf 1}^\prime (t \beta^{\mathsf{k}, \mathsf{i}} - A^{\mathsf{k}, \mathsf{i}} z)$ and the set $\tilde {\mathcal H}$ to be given by
\begin{equation} \label{eq:tildeH}
\tilde{\mathcal H} := \{(z, t) \in \homog(\mathcal X \cap \lin(\mathcal X)^\perp): \rho(z, t) = 1\}~.
\end{equation}
We first show $\mathcal H = \tilde{\mathcal H}$.
To this end, note that the function $h$ (as defined in \eqref{def:littleh}) satisfies $h(x) = \rho(x,1)$. 
Therefore, $(x,1)/\rho(x,1) = (x/h(x),1/h(x))$, which together with Lemma \ref{lem:basic}(f) implies that the map $x\mapsto (x/h(x),1/h(x))$ is a bijection from $\mathcal X \cap \lin(\mathcal X)^\perp$ onto $\{(z, t) \in \tilde{\mathcal H}: t > 0\}$.
We therefore obtain that
\begin{equation}\label{lem:normalization1}
\{(z, t) \in \tilde{\mathcal H}: t > 0\} = \{  ( x/h(x), 1/h(x) ): x \in \mathcal X \cap \lin(\mathcal X)^\perp  \}~.
\end{equation}
Since $\{(z, t) \in \tilde{\mathcal H}: t > 0\}$ is dense in $\tilde{\mathcal H}$ in the Euclidean topology on $\mathbb R^d \times \mathbb R$ by Lemma \ref{lem:basic}(g), result \eqref{lem:normalization1} implies $\mathcal H = \tilde{\mathcal H}$ as desired.
The representation of $\mathcal H$ in \eqref{eq:normalization} then follows from $\mathcal H = \tilde {\mathcal H}$, Lemma \ref{lem:basic}(b), and the definition of $\rho$.
Moreover, since $\tilde {\mathcal H}$ is a nonempty compact polyhedron by Lemma \ref{lem:basic}(e), it therefore follows that $\mathcal H$ is as well, as claimed.

To conclude, fix an arbitrary $(A, b) \in \mathbb R^{p \times d} \times \mathbb R^p$ and $y \in \mathbb R^p$. 
Since, as already argued, the function $x\mapsto (x/h(x),1/h(x))$ is a bijection from $\mathcal X \cap \lin(\mathcal X)^\perp$ onto $\{(z,t)\in \tilde {\mathcal H} : t > 0\}$, we obtain
\begin{equation}\label{lem:normalization2}
\inf_{x \in \mathcal X \cap \lin(\mathcal X)^\perp} \frac{y' M_A (A x - b)}{h(x)} = \inf_{\substack{(z, t) \in \tilde{\mathcal H} \\ t > 0}} y' M_A (A z - b t) 
 = \min_{(z, t) \in \mathcal H} y' M_A (A z - b t)~,
\end{equation}
where the second equality follows from the closure of $\{(z, t) \in \tilde{\mathcal H}: t > 0\}$ being equal to $\mathcal H$ (as already argued), the function $(z, t) \mapsto y' M_A (A z - b t)$ being continuous, and $\mathcal H$ being compact. 
Result \eqref{eq:normalized-equivalent} then follows from \eqref{lem:normalization2} holding for any $y \in \mathbb R^p$ and the minimum in \eqref{lem:normalization2} being continuous in $y$ by the Theorem of the maximum (see, e.g., Theorem 17.31 in \cite{aliprantis2006infinite}). \qed

\subsection{Proof of Lemma \ref{lem:closure}}

Part (ii) of the lemma follows from Lemma \ref{lem:auxrem}.
Part (iii) of the lemma follows from Lemmas \ref{lem:homogeneous-closure} and \ref{lem:homogeneous-projection} implying that $(A, b) \in \cl(\mathbf C_0)$ if and only if $(A, b) \in \bar{\mathbf C}_0 \cup \mathbf C^{\rm R D}$.
Part (iv) of the lemma is immediate from part (ii) implying $\mathbf C_0 \subseteq \bar{\mathbf C}_0$ and part (iii) yielding $\bar {\mathbf C}_0 \subseteq \cl(\mathbf C_0)$.
To establish part (i) select an arbitrary $(A,b)\in \cl(\mathbf C_0)$. 
By Lemma \ref{lem:homogeneous-closure} there are sequences $A_n \to A$ and $x_n\in \mathcal X$ satisfying $A_n x_n = b$ for all $n$.
Since $\mathcal X$ is compact by hypothesis, there is some subsequence $x_{n_k}$ satisfying $x_{n_k}\to x$ for some $x \in \mathcal X$. Therefore
\begin{equation*}%
b = \lim_{n \to \infty} A_n x_n = \lim_{k\to \infty} A_{n_k} x_{n_k} = Ax    
\end{equation*}
implies $Ax = b$ and hence that $(A,b)\in \mathbf C_0$.
Hence, we have shown $\cl(\mathbf C_0) \subseteq \mathbf C_0$ and part (i) immediately follows from $\mathbf C_0 \subseteq \cl(\mathbf C_0)$. \qed

\subsection{Proof of Theorem \ref{thm:test}} \label{sec:proof-test}
We prove the results for the case where $M(P)$ must be estimated, since the arguments when $M(P)$ is known are similar but more straightforward.
For $(z, t) \in \mathcal H$ and $y \in \mathbb R^p$ with $\|y\|_1 \leq 1$, define
\begin{equation} \label{thm:test1}
\hat{\mathbb G}_{z, t, n}^{(2)} (y) := \frac{\sqrt{n_2} \big ( (\hat b_{z, t, n}^{(2)})' \hat M_{n}^{(2)} y - b_{z, t}(P)' M(P) y \big )}{\hat\sigma_{z, t, n}^{(2)} (y)}~.
\end{equation}
For each $P \in \mathbf P_0$, since $\mathbf C_0 \subseteq \bar{\mathbf C}_0$, result \eqref{eq:C0bar-equality} in Lemma \ref{lem:homogeneous-projection} implies there exists $(z(P), t(P)) \in \mathcal H$ such that $M(P) b_{z(P), t(P)}(P) = 0$. Fix one $(z(P), t(P))$ for each $P \in \mathbf P_0$. If $T_n > z_{1 - \alpha}$, then
\[ z_{1 - \alpha} < T_n \leq \frac{\sqrt{n_2} (\hat b_{z(P), t(P), n}^{(2)})' \hat M_n^{(2)} \hat y_n^{(1)}}{\hat\sigma_{z(P), t(P), n}^{(2)}(\hat y_n^{(1)})} = \hat{\mathbb G}_{z(P), t(P), n}^{(2)}(\hat y_n^{(1)})~, \]
where the equality uses $M(P) b_{z(P), t(P)}(P) = 0$. Therefore,
\begin{equation} \label{thm:test2}
I \{T_n > z_{1 - \alpha} \} \leq I \{\hat{\mathbb G}_{z(P), t(P), n}^{(2)} (\hat y_n^{(1)}) > z_{1 - \alpha} \}~.
\end{equation}
Next, note that result \eqref{thm:test2} together with $\hat y_n^{(1)}$ being a function of only the first split $\{Z_i\}_{i \in I_{1, n}}$ and therefore being independent of $\hat{\mathbb G}_{z, t, n}^{(2)} (y)$ for any $(z, t) \in \mathcal H$ and $y \in \mathbb R^p$ with $\|y\|_1 \leq 1$ allows us to conclude that
\begin{align}
\label{thm:test3} \nonumber \sup_{P \in \mathbf P_0} P \{T_n > z_{1 - \alpha} \} & \leq \sup_{P \in \mathbf P_0} E_P \left [ P \left \{ \hat{\mathbb G}_{z(P), t(P), n}^{(2)} (\hat y_n^{(1)}) > z_{1 - \alpha} \Big | \{Z_i\}_{i \in I_{1, n}} \right \} \right ] \\
\nonumber & \leq \sup_{P \in \mathbf P_0} E_P \left [ \sup_{(z, t) \in \mathcal H} \sup_{\|y\|_1 \leq 1} P \left \{ \hat{\mathbb G}_{z, t, n}^{(2)} (y) > z_{1 - \alpha} \Big | \{Z_i\}_{i \in I_{1, n}} \right \} \right ] \\
& \leq \sup_{P \in \mathbf P_0} \sup_{(z, t) \in \mathcal H} \sup_{\|y\|_1 \leq 1} P \left \{ \hat{\mathbb G}_{z, t, n}^{(2)} (y) > z_{1 - \alpha} \right \}~,
\end{align}
where in the final inequality we again used that $\hat{\mathbb G}_{z, t, n}^{(2)} (y)$ is independent of the first split $\{Z_i\}_{i \in I_{1, n}}$.

We next aim to apply Lemma \ref{lem:studentized} to show $\hat{\mathbb G}_{z, t, n}^{(2)} (y) - \mathbb G_{z, t, n}^{(2)} (P; y)$ converges in probability to zero uniformly, where
\[ \mathbb G_{z, t, n}^{(2)} (P; y) := \frac{1}{\sqrt{n_2}} \sum_{i \in I_{2, n}} \frac{y^\prime\xi_{z, t}(Z_i, P)}{\sigma_{z, t} (P; y) \vee \underline\sigma}~. \]
To verify the conditions of Lemma \ref{lem:studentized} note Assumption \ref{ass:bdd-moments-IF} and Jensen's inequality imply $\mathrm{Var}_P[y^\prime\xi_{z, t}(Z, P)]$ is bounded uniformly in $P \in \mathbf P_0$, $(z, t) \in \mathcal H$, and $\|y\|_1 \leq 1$. Also note Assumption \ref{ass:rates} implies $K_{1, p} K_{2, p}^2(K_{0, p} \vee K_{1, p}) \log(1 + p) = o(n_2)$ and $K_{1, p}^{1/2} K_{2, p} a_{n_2} = o(n_2)$. By Assumptions \ref{ass:prelim-estimators}, \ref{ass:singular-value}, and \ref{ass:variance-estimator} we may therefore apply Lemma \ref{lem:studentized}, which together with Assumption \ref{ass:rates} implies
\[ \limsup_{n_2 \to \infty} \sup_{P \in \mathbf P_0} \sup_{(z, t) \in \mathcal H} \sup_{\|y\|_1 \leq 1} P \left \{ \Big | \hat{\mathbb G}_{z, t, n}^{(2)} (y) - \mathbb G_{z, t, n}^{(2)} (P; y) \Big | > \epsilon \right \} = 0 \]
for any $\epsilon > 0$. Hence, we obtain for any $\epsilon > 0$ that
\begin{align}
\label{thm:test6} \nonumber & \limsup_{n_2 \to \infty} \sup_{P \in \mathbf P_0} \sup_{(z, t) \in \mathcal H} \sup_{\|y\|_1 \leq 1} P \left \{ \hat{\mathbb G}_{z, t, n}^{(2)} (y) > z_{1 - \alpha} \right \} \\
& \leq \limsup_{n_2 \to \infty} \sup_{P \in \mathbf P_0} \sup_{(z, t) \in \mathcal H} \sup_{\|y\|_1 \leq 1} P \left \{ \mathbb G_{z, t, n}^{(2)} (P; y) > z_{1 - \alpha} - \epsilon \right \}~.
\end{align}
Next set $\eta_{n_2} (\delta) := ((K_\xi/\delta) \vee 1)^3\log(n_2)/\sqrt{n_2}$ and note that for any $\epsilon, \delta > 0$ such that $z_{1 - \alpha} - \epsilon - 3\delta/\underline{\sigma} > 0$, Lemma \ref{lem:coupling} implies that there is a universal constant $C < \infty$ and a standard normal $\mathbb Z$ such that
\begin{align}
\label{thm:test7} \nonumber & \limsup_{n_2 \to \infty} \sup_{P \in \mathbf P_0} \sup_{(z, t) \in \mathcal H} \sup_{\|y\|_1 \leq 1} P \left \{ \mathbb G_{z, t, n}^{(2)} (P; y) > z_{1 - \alpha} - \epsilon \right \} \\
\nonumber & \leq \limsup_{n_2 \to \infty} \sup_{P \in \mathbf P_0} \sup_{(z, t) \in \mathcal H} \sup_{\|y\|_1 \leq 1} \Big \{ 1 - P \left \{ \sigma_{z, t} (P; y) \mathbb Z \leq (z_{1 - \alpha} - \epsilon) (\sigma_{z, t} (P; y) \vee \underline\sigma) - 3\delta \right \} + C\eta_{n_2} (\delta) \Big \} \\
& \leq P \Big \{ \mathbb Z > z_{1 - \alpha} - \epsilon - \frac{3\delta}{\underline\sigma} \Big \}~,
\end{align}
where in the final inequality we used $z_{1 - \alpha} - \epsilon - 3\delta/\underline{\sigma} > 0$ and $\sigma_{z, t} (P; y) \leq \sigma_{z, t} (P; y) \vee \underline\sigma$. Since $\epsilon, \delta > 0$ are arbitrary and $P \{\mathbb Z > z_{1 - \alpha} \} = \alpha$, the claim of the theorem follows from results \eqref{thm:test2}, \eqref{thm:test3}, \eqref{thm:test6}, and \eqref{thm:test7}. \qed

\section{Auxiliary Lemmas for Section \ref{sec:farkas}} \label{sec:farkas-lemmas}

\begin{lemma} \label{lem:basic}
Let $\mathcal X$ be as defined in \eqref{eq:X}, set $\rho(z,t) := t + {\bf 1}^\prime(t \beta^{\mathsf{k,i}} - A^{\mathsf{k,i}}z)$ for any $(z,t)\in \mathbb R^d\times \mathbb R$, and define  the set $\tilde {\mathcal H} :=\{(z,t) \in \hom(\mathcal X\cap \lin(\mathcal X)^\perp) : \rho(z,t) = 1\}$.
Then, the following statements hold:
\begin{packed_enum}
\item[(a)] $\mathcal X \cap \lin(\mathcal X)^\perp$ is a nonempty polyhedron, $\mathcal X = (\mathcal X \cap \lin(\mathcal X)^\perp) + \lin(\mathcal X)$, and for each $x\in \mathcal X$ there are unique $x_1\in \mathcal X \cap \lin(\mathcal X)^\perp$ and $x_2 \in \lin(\mathcal X)$ satisfying $x = x_1 + x_2$.

\item[(b)] $\hom(\mathcal X \cap \lin(\mathcal X)^\perp) = \{(z,t)\in \mathbb R^d\times \mathbb R:  t\geq 0, ~ z\in \lin(\mathcal X)^\perp, ~ A^{\mathsf{k,e}}z = t \beta^{\mathsf{k,e}}, ~ A^{\mathsf{k,i}} z \leq t \beta^{\mathsf{k,i}}\}$.

\item[(c)] $\rho(z, t) > 0$ for every $(z, t) \in \homog(\mathcal X \cap \lin(\mathcal X)^\perp) \setminus \{(0, 0)\}$.

\item[(d)] If $z \in \rec(\mathcal X) \cap \lin(\mathcal X)^\perp \setminus \{0\}$, then it follows that $( z, 0) / \rho(z, 0) \in \tilde{\mathcal H}$.

\item[(e)] $\tilde{\mathcal H}$ is a nonempty compact polyhedron.

\item[(f)] The function $x \mapsto (x, 1) / \rho(x, 1)$ is a bijection from $\mathcal X \cap \lin(\mathcal X)^\perp$ onto $\{(z, t) \in \tilde{\mathcal H}: t > 0\}$ and its inverse is given by the map $(z, t) \mapsto z / t$.

\item[(g)] $\{(z, t) \in \tilde{\mathcal H}: t > 0\}$ is dense in $\tilde{\mathcal H}$ in the Euclidean topology on $\mathbb R^d \times \mathbb R$.
\end{packed_enum}
\end{lemma}

\begin{proof}
For part (a), note that the statement is immediate if $\lin(\mathcal X) = 0$ because then $\lin(\mathcal X)^\perp = \mathbb R^d$.
We therefore assume $\lin(\mathcal X) \neq 0$ and let $L$ be a matrix whose columns form an orthonormal basis for $\lin(\mathcal X)$.
To show that $\mathcal X \cap \lin(\mathcal X)^\perp$ is nonempty,  first fix an arbitrary $x \in \mathcal X$. 
By orthonormality of $L$, the matrix $LL^\prime$ is a projection matrix onto $\lin(\mathcal X)$.
Therefore, it follows that $L L' x \in \lin(\mathcal X)$ and $(\mathbf I_d - L L') x \in \lin(\mathcal X)^\perp$.  
Moreover, by definition of $\lin(\mathcal X)$ and $LL^\prime x \in \lin(\mathcal X)$ we also have that $(x - LL^\prime x) \in \mathcal X$.
Hence, we obtain that $(\mathbf I_d - LL^\prime)x\in \mathcal X\cap \lin(\mathcal X)^\perp$, which implies that $\mathcal X \cap \lin(\mathcal X)^\perp$ is nonempty.
The fact that $\mathcal X \cap \lin(\mathcal X)^\perp$ is a polyhedron follows from it being the intersection of a polyhedron ($\mathcal X$) and a linear space ($\lin(\mathcal X)^\perp$).
Moreover, since $x\in \mathcal X$ was arbitrary, it follows from the decomposition $x = (\mathbf I_d-LL')x + LL'x$ that
\begin{equation}\label{lem:basic1}
\mathcal X \subseteq (\mathcal X \cap \lin(\mathcal X)^\perp) + \lin(\mathcal X) \subseteq \mathcal X + \lin(\mathcal X) \subseteq \mathcal X~ ,
\end{equation}
where the second result follows by set inclusion and the third by definition of $\lin(\mathcal X)$.
Finally, note that the claim that the decomposition is unique follows from $\lin(\mathcal X)$ and $\lin(\mathcal X)^\perp$ being orthogonal spaces.

To establish part (b), first note that $\rec(\mathcal X) = \{z\in \mathbb R^d : A^{\mathsf{k,e}}z = 0, ~ A^{\mathsf{k,i}}z\leq 0\}.$
Since $\rec(\mathcal X \cap \lin(\mathcal X)^\perp) = \rec(\mathcal X) \cap \rec(\lin(\mathcal X)^\perp)$ by Corollary 8.3.3 in \cite{rockafellar1970convex} and $\rec(\lin(\mathcal X)^\perp) = \lin(\mathcal X)^\perp$, it follows
\begin{equation}\label{lem:basic2}
\rec(\mathcal X \cap \lin(\mathcal X)^\perp) = \{z\in \mathbb R^d :  A^{\mathsf{k,e}}z = 0, ~ A^{\mathsf{k,i}}z\leq 0, ~ z \in \lin(\mathcal X)^\perp\} ~.  
\end{equation}
Moreover, $\lin(\mathcal X)^\perp$ being a vector space implies $z/t\in \lin(\mathcal X)^\perp$ for any $z\in \lin(\mathcal X)^\perp$ and $t > 0$, which implies
\begin{multline}\label{lem:basic3}
\{(z,t) \in \mathbb R^d \times \mathbb R : t > 0, z/t\in \mathcal X \cap \lin(\mathcal X)^\perp\} \\ 
= \{(z,t) \in \mathbb R^d \times \mathbb R : t > 0, ~ z \in \lin(\mathcal X)^\perp, ~A^{\mathsf{k,e}}z = t \beta^{\mathsf{k,e}}, ~ A^{\mathsf{k,i}} z \leq t \beta^{\mathsf{k,i}}\}    ~.
\end{multline}
Part (b) of the lemma then follows by combining results \eqref{lem:basic2} and \eqref{lem:basic3} with definition \eqref{def:homo}.

For part (c), note that part (b) of this lemma implies $t\geq 0$ and $A^{\mathsf{k,i}}z\leq t \beta^{\mathsf{k,i}}$ for any $(z,t) \in \homog(\mathcal X \cap \lin(\mathcal X)^\perp)$, which in turn implies $\rho(z,t)\geq 0$ for all $(z,t) \in \homog(\mathcal X \cap \lin(\mathcal X)^\perp)$ by definition of $\rho$.
Next suppose $\rho(z_0,t_0) = 0$ for some $(z_0,t_0)\in \homog(\mathcal X \cap \lin(\mathcal X)^\perp)$.
Since $t_0\geq 0$ and $A^{\mathsf{k,i}}z_0\leq t_0 \beta^{\mathsf{k,i}}$ due to $(z_0,t_0)\in \homog(\mathcal X \cap \lin(\mathcal X)^\perp)$ and part (b) of this lemma, the definition of $\rho$ implies that $t_0 = 0$ and hence also that $A^{\mathsf{k}, \mathsf{i}} z_0 = 0$.
Moreover, $z_0$ also satisfies $A^{\mathsf{k}, \mathsf{e}} z_0 = t_0 \beta^{\mathsf{k}, \mathsf{e}}$ by part (b) of this lemma, and therefore $A^{\mathsf{k}, \mathsf{e}} z_0 = 0$ as well.
However, $z_0$ satisfying $A^{\mathsf{k}, \mathsf{e}} z_0 = 0$ and $A^{\mathsf{k}, \mathsf{i}} z_0 = 0$ implies $z_0 \in \lin(\mathcal X)$, and since $(z_0, t_0) \in \homog(\mathcal X \cap \lin(\mathcal X)^\perp)$ also implies $z_0 \in \lin(\mathcal X)^\perp$ by part (b) of this lemma, we conclude that $z_0 = 0$.
Thus, we have shown the only $(z_0,t_0) \in \homog(\mathcal X \cap \lin(\mathcal X)^\perp)$ satisfying $\rho(z_0,t_0) = 0$ is $(z_0,t_0) = (0,0)$, which, together with $\rho(z,t) \geq 0$ for all $(z,t)\in \homog(\mathcal X \cap \lin(\mathcal X)^\perp)$, establishes part (c).

To establish part (d), first note that Corollary 8.3.3 in \cite{rockafellar1970convex} and $\lin(\mathcal X)^\perp$ being a linear space imply that $\rec(\mathcal X \cap \lin(\mathcal X)^\perp) = \rec(\mathcal X)\cap \lin(\mathcal X)^\perp$.
Therefore, if $z\in \rec(\mathcal X)\cap \lin(\mathcal X)^\perp\setminus \{0\}$, then $(z,0)\in \homog(\mathcal X \cap \lin(\mathcal X)^\perp)$ by definition \eqref{def:homo}.
On the other hand, $(z,0) \in \homog(\mathcal X \cap \lin(\mathcal X)^\perp)$ with $z\neq 0$ and part (c) of this lemma yield that $\rho(z,0) > 0$. 
Thus, because $\homog(\mathcal X \cap \lin(\mathcal X)^\perp)$ is a cone (see, e.g., part (b) of this lemma), it follows that  $(z,0)/\rho(z,0) \in \homog(\mathcal X \cap \lin(\mathcal X)^\perp)$ as well.
Since in addition $\rho$ is linear, we have $\rho((z,0)/\rho(z,0)) = 1$, which allows us to conclude that $(z,0)/\rho(z,0) \in \tilde {\mathcal H}$ as desired.

For part (e), first note that $\mathcal X \cap \lin(\mathcal X)^\perp$ is nonempty by part (a) of this lemma.
Select an $x_0\in \mathcal X \cap \lin(\mathcal X)^\perp$ and then note that $(x_0,1) \in \homog(\mathcal X \cap \lin(\mathcal X)^\perp)$ (by definition \eqref{def:homo}) and part (c) of this lemma imply that $\rho(x_0,1) > 0$.
Hence, since $\homog(\mathcal X \cap \lin(\mathcal X)^\perp)$ is a cone, and $\rho$ is linear, it follows that $(x_0,1)/\rho(x_0,1) \in \homog(\mathcal X \cap \lin(\mathcal X)^\perp)$ and satisfies $\rho((x_0,1)/\rho(x_0,1)) = 1$, which establishes that $(x_0,1)/\rho(x_0,1) \in \tilde{\mathcal H}$ and therefore that $\tilde {\mathcal H}$ is non-empty.
The claim that $\tilde {\mathcal H}$ is a polyhedron then follows from $\homog(\mathcal X\cap \lin(\mathcal X)^\perp)$ being a polyhedron (see, e.g., part (b) of this lemma) and $\rho(z,t) = 1$ being a linear constraint in $(z,t)$.
To show $\tilde{\mathcal H}$ is bounded, note that $(x_0,1)/\|(x_0,1)\|_2$ belongs to both $\homog(\mathcal X \cap \lin(\mathcal X)^\perp)$ and the unit sphere, which implies their intersection is non-empty. 
We may therefore define the constant
\begin{equation}\label{lem:basic4}
c_\ast := \inf_{(x,t)}    \rho(x,t) \text{ s.t. } (x,t) \in \homog(\mathcal X \cap \lin(\mathcal X)^\perp), ~ \|(x,t)\|_2 = 1 ~,
\end{equation}
which we note satisfies $c_\ast > 0$ by part (c) of this lemma and the infimum in \eqref{lem:basic4} being attained by continuity of $\rho$ and compactness of the constraint set. 
Next, note that $\rho(z, t) = 1$ for all $(z,t)\in \tilde{\mathcal H}$ and part (c) of this lemma imply that $(0,0)\notin \tilde {\mathcal H}$. 
Therefore, for any $(z,t)\in \tilde {\mathcal H}$ we can conclude that
\begin{equation*}%
1 = \rho(z, t) = \|(z, t)\|_2 \rho \bigg ( \frac{(z, t)}{\|(z, t)\|_2} \bigg ) \geq c_\ast \|(z, t)\|_2~, 
\end{equation*}
where the second equality follows from linearity of $\rho$ and the inequality from \eqref{lem:basic4} and $(z, t) / \|(z, t)\|_2$ belonging to the intersection of $\homog(\mathcal X \cap \lin(\mathcal X)^\perp)$ with the unit sphere. 
Consequently, $\|(z, t)\|_2 \leq 1 / c_\ast$ for every $(z, t) \in \tilde{\mathcal H}$, which establishes that $\tilde{\mathcal H}$ is bounded. 
Since $\tilde{\mathcal H}$ is a polyhedron, it is also closed and therefore compact.

For part (f), note that in the proof of part (e) we already showed $x \mapsto (x, 1) / \rho(x, 1)$ takes $\mathcal X \cap \lin(\mathcal X)^\perp$ into $\{(z, t) \in \tilde{\mathcal H}: t > 0\}$. 
To show the map is onto, fix $(z_0, t_0) \in \tilde{\mathcal H}$ with $t_0 > 0$ and note that $(z_0,t_0) \in \homog(\mathcal X \cap \lin(\mathcal X)^\perp)$ (by definition of $\tilde {\mathcal H}$) implies $x_0 := z_0/t_0 \in \mathcal X \cap \lin(\mathcal X)^\perp$ (see definition \eqref{def:homo}).
Moreover, 
\begin{equation}\label{lem:basic6}
\frac{(x_0,1)}{\rho(x_0,1)} = \frac{(z_0/t_0,1)}{\rho(z_0/t_0,1)} = \frac{(z_0,t_0)}{\rho(z_0,t_0)} = (z_0,t_0) ~,
\end{equation}
where the first equality follows from the definition of $x_0$, the second from linearity of $\rho$ and the third from $\rho(z_0,t_0) = 1$ due to $(z_0,t_0)\in \tilde {\mathcal H}$.
Since $(z_0,t_0)$ was arbitrary, it follows that the map $x\mapsto (x,1)/\rho(x,1)$ is indeed onto $\{(z,t)\in \tilde{\mathcal H} : t > 0\}$ and that its inverse is $(z,t)\mapsto z/t$ as claimed.

Finally, for part (g), note that parts (a) and (f) of this lemma imply that there is a $(z_0, t_0) \in \tilde{\mathcal H}$ with $t_0 > 0$. 
Since $\tilde {\mathcal H}$ is a polyhedron by part (e), for every $(z, 0) \in \tilde{\mathcal H}$ and $0 < \epsilon < 1$ convexity implies $(z_\epsilon,t_\epsilon) : = (1 - \epsilon) (z, 0) + \epsilon (z_0, t_0)$ belongs to $\tilde{\mathcal H}$. 
The point $(z_\epsilon,t_\epsilon)$ converges to $(z, 0)$ in the Euclidean topology on $\mathbb R^d \times \mathbb R$ as $\epsilon \downarrow 0$ and satisfies $t_\epsilon > 0$ because $t_0 > 0$. 
Therefore, it follows that $\{(z, t) \in \tilde{\mathcal H}: t > 0\}$ is dense in $\tilde{\mathcal H}$, which establishes the final claim of the lemma.
\end{proof}

\begin{lemma} \label{lem:homogeneous-closure}
Let $\hom(\mathcal X)$ denote the homogenization of $\mathcal X$ (see definition \eqref{def:homo}). For every $(A, b) \in \mathbb R^{p \times d} \times \mathbb R^p$, we have $(A, b) \in \cl(\mathbf C_0)$ if and only if there exists $(u, t) \in \homog(\mathcal X) \setminus \{(0, 0)\}$ satisfying $A u = b t$.
Moreover, for every $(A, b) \in \cl(\mathbf C_0)$, there exist $A_n \to A$ and $x_n \in \mathcal X$ such that $A_n x_n = b$ for every $n$.
\end{lemma}

\begin{proof}
First suppose that $(A, b) \in \cl(\mathbf C_0)$. 
Then, there exists a sequence $\{(A_n, b_n)\} \subseteq \mathbf C_0$ converging to $(A, b)$ in the Euclidean topology on $\mathbb R^{p \times d} \times \mathbb R^p$. 
Further note that, by definition of $\mathbf C_0$, for every $n$ there is an $x_n \in \mathcal X$ satisfying $A_n x_n = b_n$.
Next, define $(u_n,t_n)\in \mathbb R^d\times \mathbb R$ to be given by
\begin{equation*}%
 (u_n, t_n) := \frac{(x_n, 1)}{\sqrt{1 + \|x_n\|_2^2}}~. 
\end{equation*}
Because $t_n > 0$ and $u_n / t_n = x_n \in \mathcal X$, we have $(u_n, t_n) \in \homog(\mathcal X)$. 
Moreover, $(u_n,t_n)$ satisfies for all $n$
\begin{equation} \label{eq:hom-unit-sphere}
A_n u_n - b_n t_n = \frac{A_n x_n - b_n}{\sqrt{1 + \|x_n\|_2^2}} = 0~,
\end{equation}
where the final equality follows from $A_nx_n = b_n$.
Since $\|(u_n,t_n)\|_2 =1$ for all $n$, the compactness of the unit sphere implies a subsequence of $\{(u_n, t_n)\}$ converges to some $(u, t)$ with $\|(u, t)\|_2 = 1$. 
By Proposition 1.14(i) in \cite{ziegler1995lectures}, $\homog(\mathcal X)$ is closed as a polyhedron and therefore $(u, t) \in \homog(\mathcal X)$. 
Taking the limit along this subsequence in \eqref{eq:hom-unit-sphere}, we obtain by continuity that $A u = b t$.
Hence, we have shown that if $(A,b)\in \cl(\mathbf C_0)$, then there exists $(u,t)\in \hom(\mathcal X) \setminus\{(0,0)\}$ satisfying $Au = bt$.

For the converse direction, suppose that $(A,b)$ is such that there is a $(u, t) \in \homog(\mathcal X) \setminus \{(0, 0)\}$ satisfying $A u = b t$. 
We will show that $(A,b)\in \cl(\mathbf C_0)$ by constructing a sequence $(A_n,x_n)$ with $A_n \to A$ and $x_n \in \mathcal X$ satisfying $A_n x_n = b$ for every $n$. 
To this end, note that if $t > 0$, then we may set $(A_n,x_n) = (A,u/t)$, which then satisfies $x_n \in \mathcal X$ by definition of $\hom(\mathcal X)$ and $A_n x_n = b$ because $Au = bt$ by hypothesis.
We therefore assume without loss of generality that $t = 0$.
Then observe that $(u,t) \in \hom(\mathcal X)\setminus\{(0,0)\}$ and $t = 0$ imply, by definition of $\hom(\mathcal X)$, that $0\neq u \in \rec(\mathcal X)$ (see definition \eqref{def:rec}).
Select an arbitrary $\bar x \in \mathcal X$ and define $x_n := \bar x + n u$, which we note satisfies $x_n \in \mathcal X$ for all $n$ due to $u\in \rec(\mathcal X)$. 
Moreover, $A x_n = A \bar x$ because $A u = b t= 0 $ (since $t = 0$), and $\|x_n\|_2 \geq n \|u\|_2 - \|\bar x\|_2 \to \infty$ by the triangle inequality and $u \neq 0$.
Therefore, for all $n$ sufficiently large, it follows that $x_n \neq 0$ and we can define the matrices 
\begin{equation*}%
A_n := A + \frac{(b - A \bar x) x_n'}{\|x_n\|_2^2}~. 
\end{equation*}
Next note that $x_n := \bar x + n u$, $Au = 0$, and the definition of $A_n$ imply (for all $n$ sufficiently large)  that
\begin{equation*}%
 A_n x_n = A x_n + (b - A \bar x) \frac{x_n' x_n}{\|x_n\|_2^2} = A \bar x + (b - A \bar x) = b~. 
\end{equation*}
Furthermore, $\|A_n - A\|_{2,2} = \|b-A\bar x\|_2/\|x_n\|_2$ and therefore $\|x_n\|_2 \to \infty$ implies that $\|A_n - A\|_{2,2}= o(1)$.
Hence, we have shown that $(A_n,b) \to (A,b)$ and $A_n x_n = b$ for some $x_n \in \mathcal X$ for all $n$ large enough, and therefore that $(A,b) \in \cl(\mathbf C_0)$. 
The second claim of the lemma is immediate from our construction.
\end{proof}

\begin{lemma} \label{lem:homogeneous-projection}
Let $\bar{\mathbf C}_0$, $\mathbf C^{\rm R D}$ and $\homog(\mathcal X)$ be as defined in \eqref{eq:C0bar}, \eqref{eq:CRD:def}, and \eqref{def:homo} respectively. Then:
\begin{equation} \label{eq:C0bar-equality}
\bar{\mathbf C}_0 = \{(A, b) \in \mathbb R^{p \times d} \times \mathbb R^p: M_A(A z - bt) = 0 \text{ for some } (z, t) \in \mathcal H\}~.
\end{equation}
Furthermore, for every $(A, b) \in \mathbb R^{p \times d} \times \mathbb R^p$ it follows that $(A, b) \in \bar{\mathbf C}_0 \cup \mathbf C^{\rm R D}$ if and only if there exists $(u, t) \in \homog(\mathcal X) \setminus \{(0, 0)\}$ satisfying $A u = b t$.
\end{lemma}

\begin{proof}
We first establish result \eqref{eq:C0bar-equality}. 
To this end, for every $(A, b) \in \mathbb R^{p \times d} \times \mathbb R^p$ define the set $\mathcal C(A, b) := \{M_A (A z - b t): (z, t) \in \mathcal H\}.$
Then note that by definition of $\bar {\mathbf C}_0$ it follows that
\begin{align}\label{lem:homogeneous-projection1}
\bar{\mathbf C}_0 & = \{(A,b) \in \mathbb R^{p\times d}\times \mathbb R^p : \max_{\|y\|_1\leq 1} \min_{c\in \mathcal C(A,b)} y^\prime c = 0  \} \notag \\
& = \{(A,b) \in \mathbb R^{p\times d}\times \mathbb R^p : \max_{\|y\|_1\leq 1} \min_{c\in \mathcal C(A,b)} y^\prime c \leq 0  \} ~,
\end{align}
where the second equality holds because it is always possible to attain zero by setting $y = 0$.
Next note that, by Lemma \ref{lem:normalization}, $\mathcal H$ is a nonempty compact polyhedron. 
Therefore, Theorem 19.3 in \cite{rockafellar1970convex} implies that its linear image $\mathcal C(A, b)$ is also a nonempty compact polyhedron. 
Next use that $y\mapsto \max_{c\in \mathcal C(A,b)} y^\prime c$ is the support function of $\mathcal C(A,b)$ to obtain by Theorem 13.1 in \cite{rockafellar1970convex} that
\begin{equation}\label{lem:homogeneous-projection2}
\{(A,b) \in \mathbb R^{p\times d}\times \mathbb R^p : \max_{\|y\|_1\leq 1} \min_{c\in \mathcal C(A,b)} y^\prime c \leq 0  \} = \{(A,b) \in \mathbb R^{p\times d}\times \mathbb R^p : 0 \in \mathcal C(A,b)\}
\end{equation}
Result \eqref{eq:C0bar-equality} therefore follows from combining results \eqref{lem:homogeneous-projection1} and \eqref{lem:homogeneous-projection2}.

To establish the second claim of the lemma, first suppose that $(u, t) \in \homog(\mathcal X) \setminus \{(0, 0)\}$ satisfies $A u = b t$. 
Let $r$ denote the projection of $u$ onto $\lin(\mathcal X)$ and set $z := u - r$, which we note satisfies $z\in \lin(\mathcal X)^\perp$.
Moreover, since $r\in \lin(\mathcal X)$ implies that $A^{\mathsf{k,e}}r = 0$ and $A^{\mathsf{k,i}}r = 0$ (see \eqref{eq:lineality}) we obtain from $u = z+ r$ that
\begin{align}\label{lem:homogeneous-projection3}
A^{\mathsf{k}, \mathsf{e}} z & = A^{\mathsf{k}, \mathsf{e}} u = t \beta^{\mathsf{k}, \mathsf{e}} \notag \\
A^{\mathsf{k}, \mathsf{i}} z & = A^{\mathsf{k}, \mathsf{i}} u \leq t \beta^{\mathsf{k}, \mathsf{i}} 
\end{align}
due to $(u,t) \in \homog(\mathcal X)$ (see definition \eqref{def:homo}) and $\rec(\mathcal X) = \{v \in \mathbb R^d : A^{\mathsf{k,e}}v=0, ~ A^{\mathsf{k,i}}v \leq 0\}$ (see definition \eqref{def:rec}).
Hence, result \eqref{lem:homogeneous-projection3}, $z\in \lin(\mathcal X)^\perp$ and $t \geq 0$ due to $(u,t)\in \hom(\mathcal X)$ together with Lemma \ref{lem:basic}(b) imply $(z,t)\in \hom(\mathcal X\cap \lin(\mathcal X)^\perp)$.
Moreover, since $M_A(Ar) = 0$ due to $r\in \lin(\mathcal X)$ and $M_A$ projecting onto the orthogonal complement of $A(\lin(\mathcal X))$, we obtain from $Au - bt = 0$ by hypothesis and $u = r+z$ that
\begin{equation}\label{lem:homogeneous-projection4}
M_A(Az - bt) = M_A(A(z+r) - bt) = M_A(Au - bt) = 0    ~.
\end{equation}
If $(z, t) \neq (0, 0)$, then Lemma \ref{lem:basic}(c) implies $\rho(z, t) > 0$, and therefore $z \in \lin(\mathcal X)^\perp$, result \eqref{lem:homogeneous-projection3}, and Lemma \ref{lem:normalization} yield that $(z, t) / \rho(z, t) \in \mathcal H$. Since \eqref{lem:homogeneous-projection4} further implies $M_A(Az/\rho(z,t) - bt/\rho(z,t)) = 0$, result \eqref{eq:C0bar-equality} implies $(A,b) \in \bar{\mathbf C}_0 $ as desired.
On the other hand, if $(z, t) = (0, 0)$, then $t = 0$ and $(u,t) \neq (0,0)$ imply $u \neq 0$.
However, since $u = r+ z$ and $z = 0$, we must also have $u = r\neq 0$, which together with $Au = b t = b\cdot 0 = 0$ implies $Ar = 0$.
Since $r \in \lin(\mathcal X)$, we therefore obtain that $\dim(A(\lin(\mathcal X))) < \dim(\lin(\mathcal X))$ and hence that $(A,b)\in \mathbf C^{\rm RD}$.

For the reverse inclusion, we next show that if $(A,b) \in \bar{\mathbf C}_0\cup \mathbf C^{\rm RD}$, then $Au = bt$ for some $(u,t) \in \hom(\mathcal X)\setminus\{(0,0)\}$.
To this end,  first suppose $(A, b) \in \bar{\mathbf C}_0$. 
By result \eqref{eq:C0bar-equality}, there then exists a $(z, t) \in \mathcal H$ satisfying $M_A (A z - b t) = 0$. 
Since $M_A$ projects onto the orthogonal complement of $A(\lin(\mathcal X))$, it follows that $b t - A z \in A(\lin(\mathcal X))$ and hence that $bt - Az = Aw$ for some $w\in \lin(\mathcal X)$.
Define $u := z + w$ and note that $A^{\mathsf{k}, \mathsf{e}} w = 0$ and $A^{\mathsf{k}, \mathsf{i}} w = 0$ due to $w\in \lin(\mathcal X)$ (see \eqref{eq:lineality}) and $(z,t) \in \mathcal H$ imply that 
\begin{equation}\label{lem:homogeneous-projection5}
A^{\mathsf{k}, \mathsf{e}} u = t \beta^{\mathsf{k}, \mathsf{e}} \hspace{0.5 in} A^{\mathsf{k}, \mathsf{i}} u \leq t \beta^{\mathsf{k}, \mathsf{i}}
\end{equation}
by Lemma \ref{lem:normalization}.
Since $t \geq 0$, result \eqref{lem:homogeneous-projection5} and $\rec(\mathcal X) = \{v \in \mathbb R^d : A^{\mathsf{k,e}}v=0, ~ A^{\mathsf{k,i}}v \leq 0\}$ (see definition \eqref{def:rec}) imply $(u, t) \in \homog(\mathcal X)$.
Moreover, since $z\in \lin(\mathcal X)^\perp$ due to $(z,t)\in \mathcal H$ and Lemma \ref{lem:normalization}, $w\in \lin(\mathcal X)$ implies
\begin{equation*}%
\|(u,t)\|_2^2 = \|z\|_2^2 + \|w\|_2^2 + t^2 \geq \|(z,t)\|_2^2 > 0    
\end{equation*}
where the strict inequality follows from $(z,t) \in \mathcal H$.
Hence, we conclude that $(u,t)\in \hom(\mathcal X)\setminus \{(0,0)\}$.
Since, in addition, $Aw = bt-Az$ by construction, we finally obtain that $Au = Az + Aw = bt$, and hence we can conclude that if $(A,b)\in \bar{\mathbf C}_0$, then there is a $(u,t)\in \hom(\mathcal X)\setminus\{(0,0)\}$ satisfying $Au= bt$.
On the other hand, if $(A, b) \in \mathbf C^{\rm R D}$, then there exists $0 \neq w \in \lin(\mathcal X)$ such that $A w = 0$. 
Since $\lin(\mathcal X) \subseteq \rec(\mathcal X)$, we have $(w,0)\in \hom(\mathcal X)\setminus\{(0,0)\}$ (see definition \eqref{def:homo}), which satisfies $Aw = 0 = b\cdot 0$, as desired.
\end{proof}

\begin{lemma}\label{lem:auxrem}
Let $\mathbf C_0$, $\bar{\mathbf C}_0$, and $\rec(\mathcal X)$ be as defined in \eqref{def:C0}, \eqref{eq:C0bar},  and \eqref{def:rec} respectively. Then:
$$\bar {\mathbf C}_0 = \mathbf C_0 \cup \{(A,b) \in \mathbb R^{p\times d} \times \mathbb R^p : Av = 0 \text{ for some } v\in \rec(\mathcal X)\setminus \lin(\mathcal X)\}.$$
\end{lemma}

\begin{proof}
Select any $(A,b) \in \bar{\mathbf C}_0$, and note that Lemma \ref{lem:homogeneous-projection} then implies that $M_A(Az - bt) = 0$ for some $(z,t) \in \mathcal H$.
First suppose $t > 0$.
Then, Lemma \ref{lem:normalization} and $(z,t)\in \mathcal H$ imply that $z/t\in \mathcal X \cap \lin(\mathcal X)^\perp$.
Moreover, since $M_A(Az-bt) = 0$ and $M_A$ projects onto the orthogonal complement of $A(\lin(\mathcal X))$, it follows that there is some $w\in \lin(\mathcal X)$ such that $b-A(z/t) = Aw$.
It then follows that $x := z/t+w \in \mathcal X$ (because $z/t\in \mathcal X$ and $w\in \lin(\mathcal X)$) and in addition $Ax = A(z/t) + Aw = b$, which implies $(A,b)\in \mathbf C_0$.
On the other hand, if $t = 0$ then $M_A(Az - bt)= 0$ implies $M_A(Az) = 0$, while $(z,0) \in \mathcal H$, Lemma \ref{lem:normalization}, and $\rec(\mathcal X) = \{v \in \mathbb R^d : A^{\mathsf{k,e}}v=0, ~ A^{\mathsf{k,i}}v \leq 0\}$ (see definition \eqref{def:rec}) imply $0\neq z\in \rec(\mathcal X) \cap \lin(\mathcal X)^\perp$.
Hence, we obtain that
\begin{equation} \label{lem:auxrem1}
M_A A z = 0 \text{ for some } z \in \rec(\mathcal X) \cap \lin(\mathcal X)^\perp \setminus \{0\}~.
\end{equation}
Since $M_A$ projects onto the orthogonal complement of $A(\lin(\mathcal X))$, it follows that $Az \in A(\lin(\mathcal X))$ and therefore that there exists a $w\in \lin(\mathcal X)$ such that $Aw = Az$.
Setting $v:= z-w$, we then note that $Av = 0$, $v\notin \lin(\mathcal X)$ because $w\in \lin(\mathcal X)$ and $0\neq z\in \lin(\mathcal X)^\perp$, and $v\in \rec(\mathcal X)$ because $z\in \rec(\mathcal X)$ and $w\in \lin(\mathcal X)$.
Hence, we have shown that if $t = 0$, there is a $v\in \rec(\mathcal X)\setminus \lin(\mathcal X)$ satisfying $Av = 0$.
Therefore, it follows that
\begin{equation}\label{lem:auxrem2}
\bar {\mathbf C}_0 \subseteq \mathbf C_0 \cup \{(A,b) \in \mathbb R^{p\times d} \times \mathbb R^p : Av = 0 \text{ for some } v\in \rec(\mathcal X)\setminus \lin(\mathcal X)\} ~. 
\end{equation}

For the reverse inclusion, let $(A,b)\in \mathbf C_0$ and note that Lemma \ref{lem:annihilator} implies $M_A(Ax - b) = 0$ for some $x\in \mathcal X \cap \lin(\mathcal X)^\perp$ (see, e.g., \eqref{lem:annihilator1}).
Setting $(z,t) = (x/h(x),1/h(x))$ we then obtain that $(z,t)\in \mathcal H$ by definition of $\mathcal H$ (see \eqref{eq:normalized-domain}) and $M_A(Az - t b) = 0$.
Lemma \ref{lem:homogeneous-projection} then yields $(A,b)\in \bar {\mathbf C}_0$ and therefore we have shown $\mathbf C_0 \subseteq \bar {\mathbf C}_0$.
Next, suppose $(A,b)$ is such that $Av = 0$ for some $v\in \rec(\mathcal X)\setminus \lin(\mathcal X)$.
Let $w$ denote the projection of $v$ onto $\lin(\mathcal X)$ and set $r:= v - w$, which we note satisfies $r\in \lin(\mathcal X)^\perp$.
Further note $r\neq 0$ because $v\in \rec(\mathcal X)\setminus \lin(\mathcal X)$ and $r\in \rec(\mathcal X)$ because $v\in \rec(\mathcal X)$ and $w\in \lin(\mathcal X)$.
Hence, we conclude that $r\in \rec(\mathcal X)\cap \lin(\mathcal X)^\perp \setminus \{0\}$ and, since 
$\rec(\mathcal X) = \{v \in \mathbb R^d : A^{\mathsf{k,e}}v=0, ~ A^{\mathsf{k,i}}v \leq 0\}$ (see definition \eqref{def:rec}), Lemma \ref{lem:basic}(b) implies $(r,0)\in \hom(\mathcal X\cap \lin(\mathcal X)^\perp)$.
Next use $r\neq 0$ and Lemma \ref{lem:basic}(c) to conclude that $\rho(r,0) > 0$ for $\rho(r,t) := t+{\bf 1}^\prime(t\beta^{\mathsf{k,i}} - A^{\mathsf{k,i}}r)$.
Therefore, setting $z:= r/\rho(r,0)$ we obtain from $r \in \rec(\mathcal X)\cap \lin(\mathcal X)^\perp \setminus\{0\}$ and Lemma \ref{lem:normalization} that $(z,0)\in \mathcal H$.
Finally, note that $Av = 0$ by hypothesis and $v = r+w$ implies that $Ar = - Aw$.
Since $Aw\in A(\lin(\mathcal X))$ and $M_A$ projects onto the orthogonal complement of $A(\lin(\mathcal X))$ it follows that $0 = M_A(Aw) = M_A(Az)$.
Thus, $(z,0) \in \mathcal H$ satisfies $M_A(Az - b\cdot 0) = 0$, which together with Lemma \ref{lem:homogeneous-projection} yields that $(A,b) \in \bar {\mathbf C}_0$.
We are therefore able to conclude that
\begin{equation*}%
 \mathbf C_0 \cup \{(A,b) \in \mathbb R^{p\times d} \times \mathbb R^p : Av = 0 \text{ for some } v\in \rec(\mathcal X)\setminus \lin(\mathcal X)\} \subseteq \bar {\mathbf C}_0 ~,
\end{equation*}
which together with result \eqref{lem:auxrem2} establishes the claim of the lemma.
\end{proof}

\section{Auxiliary Lemmas for Section \ref{sec:test}}
For notational simplicity, in this section we suppress the superscript $(k)$ indicating the sample split. Throughout this section, we prove the results for the case where $M(P)$ must be estimated, since the arguments when $M(P)$ is known are similar but more straightforward.

\begin{lemma} \label{lem:studentized}
Suppose Assumptions \ref{ass:prelim-estimators} and \ref{ass:singular-value} hold, $\sup_{\|y\|_1 \leq 1} {\rm Var}_P[y^\prime \xi_{z, t}(Z, P)] < B < \infty$ for all $P \in \mathbf P_0$, $(z, t) \in \mathcal H$ and $p \geq 1$, and that $\hat V_{z, t, n}$ is such that uniformly in $P \in \mathbf P_0$ and $(z, t) \in \mathcal H$ we have
\[ \|\hat V_{z, t, n} - V_{z, t}(P) \|_{2, 2} = O_P(\delta_n) ~. \]
Further suppose that $K_{1, p} K_{2, p}^2 (K_{0, p} \vee K_{1, p}) \log(1 + p) = o(n)$ and $K_{1, p}^{1/2} K_{2, p} a_n = o(n)$, and define
\begin{align*}
r_n & := \frac{K_{2, p} (\log(1 + p) (K_{0, p} \vee K_{1, p}) + a_n)}{\sqrt n} \\
q_n & := \sqrt{K_{1, p}} K_{2, p} \left ( \sqrt{\frac{(K_{0, p} \vee K_{1, p}) \log(1 + p)}{n}} + \frac{a_n}{n} \right ) + K_{2, p}^2 \delta_n ~ .
\end{align*}
It then follows that uniformly in $P \in \mathbf P_0$, $(z, t) \in \mathcal H$, and $y \in \mathbb R^p$ with $\|y\|_1 \leq 1$ we have that
\[ \Big | \frac{\sqrt n(\hat b_{z, t, n}^\prime \hat M_{n} y - b_{z, t}(P)^\prime M(P) y)}{\hat \sigma_{z, t, n} (y)} - \frac{1}{\sqrt n} \sum_{i = 1}^n \frac{y^\prime \xi_{z, t}(Z_i, P)}{(\sigma_{z, t} (P; y) \vee \underline{\sigma})} \Big | = O_P \Big ( \frac{r_n}{\underline{\sigma}} + \frac{q_n}{\underline{\sigma}^3} \Big ) ~ . \]
\end{lemma}

\begin{proof}
To begin, note that Lemma \ref{lem:rate_final} implies that uniformly in $P \in \mathbf P_0$ and $(z, t) \in \mathcal H$ we have
\begin{equation} \label{eq:studentized1}
\sup_{\|y\|_2 \leq 1}|\sqrt{n} (\hat b_{z, t, n}' \hat M_{n} y - b_{z, t}(P)' M(P) y) - \frac{1}{\sqrt n} \sum_{i = 1}^n \xi_{z, t}(Z_i, P)^\prime y| = O_P(r_n)~ .
\end{equation}
Next note that $\hat \sigma_{z, t, n} (y) \geq \underline{\sigma} > 0$ by construction. Lemma \ref{lem:lipschitz} implies
\begin{multline}\label{eq:studentized2}
\sup_{\|y\|_1 \leq 1} \Big|\frac{1}{\hat \sigma_{z, t, n}(y)} - \frac{1}{\sigma_{z, t}(P;y) \vee \underline{\sigma}}\Big| \leq \sup_{\|y\|_1 \leq 1} \frac{1}{2\underline{\sigma}^3}|\hat \sigma^2_{z, t, n}(y) - (\sigma_{z, t}^2(P;y)\vee \underline{\sigma}^2)| \\ \leq \sup_{\|y\|_1 \leq 1} \frac{1}{2\underline{\sigma}^3}| \hat D_{z, t, n}(y)^\prime \hat V_{z, t, n} \hat D_{z, t, n}(y) - \sigma_{z, t}^2(P;y)| = O_P\Big(\frac{q_n}{\underline{\sigma}^3}\Big), 
\end{multline}
where the second inequality follows from $|(x_1 \vee \underline{\sigma}^2) - (x_2 \vee \underline{\sigma}^2)| \leq |x_1 - x_2|$ for any $x_1, x_2 \in \mathbb R$, and the final result holds uniformly in $P \in \mathbf P_0$ and $(z, t) \in \mathcal H$ by Lemma \ref{lem:var_rate}. Moreover, note that we have
\begin{equation} \label{eq:studentized3}
\frac{1}{\sqrt n} \sum_{i = 1}^n y^\prime \xi_{z, t}(Z_i, P) = O_P(1)
\end{equation}
uniformly in $P \in \mathbf P_0$, $(z, t) \in \mathcal H$, and $y \in \mathbb R^p$ with $\|y\|_1 \leq 1$ due to Markov's inequality and the condition $\sup_{\|y\|_1 \leq 1} \mathrm{Var}_P[y^\prime \xi_{z, t}(Z, P)] < B < \infty$. Since $\hat \sigma_{z, t, n} (y) \geq \underline{\sigma}$, results \eqref{eq:studentized1}, \eqref{eq:studentized2}, and \eqref{eq:studentized3} then allow us to conclude uniformly in $P \in \mathbf P_0$, $(z, t) \in \mathcal H$, and $y \in \mathbb R^p$ with $\|y\|_1 \leq 1$ that
\begin{multline}
\frac{\sqrt n(\hat b_{z, t, n}^\prime \hat M_{n} y - b_{z, t}(P)^\prime M(P) y)}{\hat \sigma_{z, t, n} (y)} - \frac{1}{\sqrt n} \sum_{i = 1}^n \frac{y^\prime \xi_{z, t}(Z_i, P)}{(\sigma_{z, t} (P; y) \vee \underline{\sigma})} \\
 = \Big ( \frac{1}{\hat \sigma_{z, t, n} (y)} - \frac{1}{\sigma_{z, t} (P; y) \vee \underline{\sigma}} \Big ) \frac{1}{\sqrt n} \sum_{i = 1}^n \xi_{z, t}(Z_i, P)^\prime y + O_P \Big ( \frac{r_n}{\underline{\sigma}} \Big ) = O_P \Big ( \frac{r_n}{\underline{\sigma}} + \frac{q_n}{\underline{\sigma}^3} \Big ),
\end{multline}
which establishes the claim of the lemma.
\end{proof}

\begin{lemma} \label{lem:rate_final}
Let Assumptions \ref{ass:prelim-estimators} and \ref{ass:singular-value} hold, and suppose $(K_{0, p} \vee K_{1, p}) \log(1 + p) = o(n)$. Then, it follows that uniformly in $P \in \mathbf P_0$ and $(z, t) \in \mathcal H$ we have
\begin{multline*}
\sup_{ \|y\|_2 \leq 1} \Big | \sqrt{n} (\hat b_{z, t, n}' \hat M_{n} y - b_{z, t}(P)' M(P) y) - \frac{1}{\sqrt{n}} \sum_{1 \le i \le n} \xi_{z, t}(Z_i, P) ' y \Big | \\
= O_P \Big ( \frac{K_{2, p} (\log(1 + p) (K_{0, p} \vee K_{1, p}) + a_n)}{\sqrt n} \Big )~.
\end{multline*}
\end{lemma}
\begin{proof}
Let $D(A_0(P))[H] := -(M(P) H A_0^\dagger(P) + A_0^\dagger(P)' H^\prime M(P))$ for any $p\times \ell$ matrix $H$ and define
\begin{align*}
\hat S_{z, t, n}(P) & := M(P) \sqrt{n}(\hat b_{z, t, n} - b_{z, t}(P)) + D(A_0(P))[\sqrt n(\hat A_{0,n}-A_{0}(P))] b_{z, t}(P)\\
S_{z, t, n}^*(P) & := \frac{1}{\sqrt n} \sum_{i=1}^n \xi_{z, t}(Z_i,P)~.
\end{align*}
Next note that Lemma \ref{lem:rates} and $(K_{0, p} \vee K_{1, p}) \log(1 + p) = o(n)$ imply that uniformly in $P \in \mathbf P_0$ we have
\begin{equation} \label{lem:rate-final1}
\|\hat A_{0, n} - A_0(P)\|_{2, 2} = O_P \Big ( \sqrt{\frac{(K_{0, p} \vee K_{1, p}) \log(1 + p)}{n}} + \frac{a_n}{n} \Big ) = o_P(1)~.
\end{equation}
Together with Assumption \ref{ass:singular-value}, result \eqref{lem:rate-final1} implies that $\|\hat A_{0, n} - A_0(P)\|_{2, 2} \leq \underline s/2$ with probability tending to one uniformly in $P \in \mathbf P_0$. We may therefore apply Lemma \ref{lem:lin_error} together with Lemma \ref{lem:rates} and Assumption \ref{ass:singular-value} to conclude
\begin{equation} \label{lem:rate-final2}
\sup_{y \in \mathbb R^p : \|y\|_2 \leq 1} |\sqrt n(\hat b_{z, t, n}^\prime \hat M_{n} y - b_{z, t}(P)^\prime M(P) y) - \hat S_{z, t, n} (P)^\prime y| = O_P \Big ( \frac{K_{2, p}}{\sqrt n} ( (K_{0, p} \vee K_{1, p}) \log(1 + p) + \frac{a_n^2}{n}) \Big )
\end{equation}
uniformly in $P \in \mathbf P_0$ and $(z, t) \in \mathcal H$. Furthermore, by definition of $\xi_{z, t}(Z_i, P)$ we also have that
\begin{align}
\|\hat S_{z, t, n}(P) - S_{z, t, n}^*(P)\|_2 & \leq \Big \|M(P)(\sqrt{n}(\hat b_{z, t, n} - b_{z, t}(P)) - \frac{1}{\sqrt{n}}\sum_{1 \leq i \leq n} \varphi_{z, t}(Z_i, P)) \Big \|_2 \nonumber \\
& \hspace{2em} + \Big \|M(P) \Big ( \sqrt{n}(\hat A_{0, n} - A_0(P)) - \frac{1}{\sqrt{n}} \sum_{1 \leq i \leq n} \Psi_0(Z_i, P) \Big ) A_0^\dagger(P) b_{z, t}(P) \Big \|_2 \nonumber \\
& \hspace{2em} + \Big \| A_0^\dagger(P)' \Big ( \sqrt{n}(\hat A_{0, n} - A_0(P)) - \frac{1}{\sqrt{n}} \sum_{1 \leq i \leq n} \Psi_0(Z_i, P) \Big )'M(P) b_{z, t}(P) \Big \|_2~.\label{lem:rate-final3}
\end{align}
Next, note that using Lemma \ref{lem-auxAdagger}, $\|M(P) \|_{2, 2} \leq 1$ because $M(P)$ is a projection matrix, result \eqref{lem:rate-final3}, and Assumptions \ref{ass:prelim-estimators}(a) and \ref{ass:singular-value} we can conclude that uniformly in $P \in \mathbf P_0$ and $(z, t) \in \mathcal H$
\begin{equation} \label{lem:rate-final4}
\|\hat S_{z, t, n} (P) - S_{z, t, n}^*(P) \|_2 = O_P \Big ( \frac{a_n K_{2, p}}{\sqrt{n}} \Big ) ~.
\end{equation}
The claim of the lemma therefore follows from results \eqref{lem:rate-final2}, \eqref{lem:rate-final4}, the Cauchy-Schwarz inequality and $a_n/\sqrt n = o(1)$ by Assumption \ref{ass:prelim-estimators}.
\end{proof}

\begin{lemma} \label{lem:variance-xi}
Suppose $\mathrm{rank}(A_0(P)) = \ell$ and $V_{z, t}(P) < \infty$. Then, $\var_P[\xi_{z, t}(Z_i, P)' y] = \sigma_{z, t}^2(P; y)$.
\end{lemma}

\begin{proof}
Using $\mvec(A B C) = (C' \otimes A) \mvec(B)$, $(A \otimes B)' = A' \otimes B'$, and the definition of $D_{z, t}(P; y)$ we obtain
\begin{align*}
\xi_{z, t}& (Z_i, P)' y \\ & = y' M(P) \varphi_{z, t}(Z_i, P) - y' M(P) \Psi_0(Z_i, P) A_0^\dagger(P) b_{z, t}(P) - y' A_0^\dagger(P)' \Psi_0(Z_i, P)' M(P) b_{z, t}(P) \\
& = y' M(P) \varphi_{z, t}(Z_i, P) - \mvec \big ( y' M(P) \Psi_0(Z_i, P) A_0^\dagger(P) b_{z, t}(P) \big ) - \mvec \big ( b_{z, t}(P)' M(P) \Psi_0(Z_i, P) A_0^\dagger(P) y \big ) \\
& = y' M(P) \varphi_{z, t}(Z_i, P) - \{(A_0^\dagger(P) b_{z, t}(P))' \otimes (y' M(P))+ (A_0^\dagger(P)y)' \otimes (b_{z, t}(P)' M(P))\} \mvec ( \Psi_0(Z_i, P) ) \\
& = D_{z, t}(P; y)' \begin{pmatrix} \mvec \Psi_0(Z_i, P) \\ \varphi_{z, t}(Z_i, P) \end{pmatrix}~,
\end{align*}
which establishes the claim of the lemma.
\end{proof}

\begin{lemma} \label{lem:var_rate}
Suppose Assumptions \ref{ass:prelim-estimators} and \ref{ass:singular-value} hold, and suppose $\sup_{\|y\|_1 \leq 1} {\rm Var}_P[y^\prime \xi_{z, t}(Z, P)] < B < \infty$ for all $P \in \mathbf P_0$, $(z, t) \in \mathcal H$ and $p \geq 1$. If $\hat V_{z, t, n}$ is such that uniformly in $P \in \mathbf P_0$ and $(z, t) \in \mathcal H$ we have
\[ \|\hat V_{z, t, n} - V_{z, t}(P) \|_{2, 2} = O_P(\delta_n), \]
$K_{1, p} K_{2, p}^2 (K_{0, p} \vee K_{1, p}) \log(1 + p) = o(n)$, and $K_{1, p}^{1/2} K_{2, p} a_n = o(n)$, then it follows that
\begin{multline*}
\sup_{y \in \mathbb R^p: \|y\|_1 \leq 1} |\hat D_{z, t, n}(y)' \hat V_{z, t, n} \hat D_{z, t, n}(y) - \sigma_{z, t}^2(P; y)| \\
= O_P \Big ( \sqrt{K_{1, p}} K_{2, p} (\sqrt{\frac{(K_{0, p} \vee K_{1, p}) \log(1 + p)}{n}} + \frac{a_n}{n}) + K_{2, p}^2 \delta_n \Big )
\end{multline*}
uniformly in $P \in \mathbf P_0$ and $(z, t) \in \mathcal H$.
\end{lemma}

\begin{proof}
To begin the proof, first note that direct calculation yields the following decomposition:
\begin{align}\label{lem:var_rate1}
	& \hat D_{z, t, n}(y)'\hat V_{z, t, n}\hat D_{z, t, n}(y) - D_{z, t}(P; y)'V_{z, t}(P) D_{z, t}(P; y) \nonumber \\
    &= 2D_{z, t}(P; y)'V_{z, t}(P)(\hat D_{z, t, n}(y) - D_{z, t}(P; y)) + D_{z, t}(P; y)'(\hat V_{z, t, n} - V_{z, t}(P))D_{z, t}(P; y) \nonumber\\
	& \hspace{2em} + (\hat D_{z, t, n}(y) - D_{z, t}(P; y))'V_{z, t}(P) (\hat D_{z, t, n}(y) - D_{z, t}(P; y)) \nonumber \\
    & \hspace{2em} + 2D_{z, t}(P; y)'(\hat V_{z, t, n} - V_{z, t}(P))(\hat D_{z, t, n}(y) - D_{z, t}(P; y)) \nonumber \\
	& \hspace{2em} + (\hat D_{z, t, n}(y) - D_{z, t}(P; y))'(\hat V_{z, t, n} - V_{z, t}(P)) (\hat D_{z, t, n}(y) - D_{z, t}(P; y))~. 
\end{align}
Next, let $V_{z, t}^{1/2} (P)$ denote the unique positive semi-definite square root of $V_{z, t}(P)$. Then note that \eqref{lem:var_rate1} and $V_{z, t}(P) = V_{z, t}^{1/2} (P) V_{z, t}^{1/2} (P)$ imply that uniformly in $P$ and $(z, t) \in \mathcal H$ we have
\begin{align}\label{eq:sigma_diff}
 |\hat D_{z, t, n}(y)' &\hat V_{z, t, n}\hat D_{z, t, n}(y) - D_{z, t}(P; y)'V_{z, t}(P) D_{z, t}(P; y)| \nonumber \\
 \lesssim ~ & \|V_{z, t}^{1/2}(P)(\hat D_{z, t, n}(y) - D_{z, t}(P;y))\|_{2} \nonumber \\
 & \hspace{2em} \cdot (\|V_{z, t}^{1/2}(P) D_{z, t}(P; y)\|_2 + \|V_{z, t}^{1/2}(P)(\hat D_{z, t, n}(y) - D_{z, t}(P;y))\|_2)\nonumber \\
 & + \|\hat V_{z, t, n} - V_{z, t}(P)\|_{2,2}\cdot (\|D_{z, t}(P;y)\|^2_2 +   \|\hat D_{z, t, n}(y) - D_{z, t}(P;y)\|^2_2) ~.
\end{align}

To control the terms on the right-hand side of \eqref{eq:sigma_diff}, first use that $\|a \otimes b\|_2 = \|a\|_2 \cdot \|b\|_2$ for any vectors $a$ and $b$, the triangle inequality, and the definition of $D_{z, t}(P; y)$ to obtain for any $y$ with $\|y\|_1 \leq 1$ that
\begin{align} \label{eq:D_bound}
	\|D_{z, t}(P; y)\|_2   & \leq \|M(P)y\|_2 + \|A_0^\dagger(P) y\|_2 \cdot \|M(P) b_{z, t}(P)\|_2 + \|A_0^\dagger(P) b_{z, t}(P)\|_2 \cdot\|M(P) y\|_2 \nonumber \\
	& \leq 1 + 2 \frac{\|b_{z, t}(P)\|_2}{\underline s(A_0(P))}~,
\end{align}
where the final inequality follows from Lemma \ref{lem-auxAdagger}, $\|M(P) \|_{2, 2} \leq 1$ because $M(P)$ is a projection matrix, and $\|y\|_2 \leq \|y\|_1 \leq 1$. Furthermore, by similar arguments we also obtain the following inequality
\begin{multline}\label{lem:var_rate2}
\|\hat D_{z, t, n}(y) - D_{z, t}(P;y)\|_2 \leq \|(\hat M_{n} - M(P))y\|_{2} + \|(\hat A_{0,n}^\dagger y)\otimes (\hat M_{n}\hat b_{z, t, n}) - (A_0^\dagger(P) y) \otimes (M(P) b_{z, t}(P))\|_2 \\ +  \|(\hat A_{0, n}^{\dagger}\hat b_{z, t, n}) \otimes (\hat M_{n}y) - (A_0^\dagger(P) b_{z, t}(P)) \otimes (M(P)y)\|_2~.
\end{multline}
Lemma \ref{lem:weyl} applied to $\hat A_{0, n}$ and $A_0(P)$ implies
\begin{equation} \label{lem:var_rate3}
\underline s(\hat A_{0, n}) \geq \underline s(A_0(P)) - \|\hat A_{0, n} - A_0(P)\|_{2, 2}~.
\end{equation}
In particular, since $\|\hat A_{0, n} - A_0(P)\|_{2, 2} = o_P(1)$ uniformly in $P \in \mathbf P_0$ by Lemma \ref{lem:rates}, $K_{1, p} K_{2, p}^2 (K_{0, p} \vee K_{1, p}) \log(1 + p) = o(n)$, and $a_n = o(n)$, it follows from \eqref{lem:var_rate3} and Assumption \ref{ass:singular-value} that
\begin{equation} \label{lem:var_rate4}
\frac{1}{\underline s(\hat A_{0, n})} = O_P(1)
\end{equation}
uniformly in $P \in \mathbf P_0$. Therefore, applying Theorem 2.5 in \citet{chen2016perturbation} and Lemma \ref{lem:rates} yields that
\begin{equation} \label{lem:var_rate5}
\|\hat M_{n} - M(P) \|_{2, 2} = O_P \Big ( \sqrt{\frac{(K_{0, p} \vee K_{1, p}) \log(1 + p)}{n}} + \frac{a_n}{n} \Big )
\end{equation}
uniformly in $P \in \mathbf P_0$. Similarly, result \eqref{lem:var_rate4}, Assumption \ref{ass:singular-value}(a), Lemma \ref{lem:rates}, and Theorem 4.1 in \citet{wedin1973perturbation} allow us to conclude that uniformly in $P \in \mathbf P_0$ we have
\begin{equation} \label{lem:var_rate6}
\|\hat A_{0, n}^\dagger - A_0^\dagger(P)\|_{2, 2} = O_P \Big ( \sqrt{\frac{(K_{0, p} \vee K_{1, p}) \log(1 + p)}{n}} + \frac{a_n}{n} \Big )~.
\end{equation}
Next, again use that $\|a \otimes b\|_2 = \|a\|_2 \cdot \|b\|_2$ for any vectors $a$ and $b$, Lemma \ref{lem-auxAdagger}, Assumption \ref{ass:singular-value}, and that $\hat M_{n}$ is a projection matrix to obtain that for any $y$ with $\|y\|_2 \leq 1$ we have
\begin{align}\label{lem:var_rate7}
	 \|\hat A_{0, n}^{\dagger}\hat b_{z, t, n} &\otimes  \hat M_{n}y - A_0^\dagger(P) b_{z, t}(P) \otimes M(P)y\|_2 \nonumber \\
    & \leq \|\hat A_{0, n}^{\dagger}\hat b_{z, t, n}  - A_0^\dagger(P) b_{z, t}(P)\|_2 \cdot \|\hat M_{n} y\|_2 + \|A_0^\dagger(P) b_{z, t}(P)\|_2 \cdot \|(\hat M_{n} - M(P))y\|_2 \nonumber \\
	&\lesssim  \|\hat A_{0, n}^\dagger - A_0^\dagger(P)\|_{2,2} \cdot \|b_{z, t}(P)\|_2 + \frac{1}{\underline{s}(\hat A_{0,n})} \|\hat b_{z, t, n} - b_{z, t}(P)\|_{2}  +  \|b_{z, t}(P)\|_2\cdot \|\hat M_{n} - M(P)\|_{2,2} ~.
\end{align}
Furthermore, by similar arguments it also follows that for any $y$ with $\|y\|_2 \leq 1$ we have the upper bound
\begin{align}\label{lem:var_rate8}
	 \|(\hat A_{0, n}^{\dagger}y) & \otimes  (\hat M_{n}\hat b_{z, t, n}) - (A_0^\dagger(P) y) \otimes (M(P)b_{z, t}(P))\|_2 \nonumber \\
    & \leq \|(\hat A_{0, n}^{\dagger}  - A_0^\dagger(P)) y\|_2 \cdot \|M(P) b_{z, t}(P)\|_2 + \|\hat A_{0,n}^\dagger y\|_2 \cdot \|\hat M_{n} \hat b_{z, t, n} - M(P)b_{z, t}(P)\|_2 \nonumber \\
	&\leq \|\hat A_{0, n}^{\dagger} - A_0^\dagger(P)\|_{2,2}  \cdot\|b_{z, t}(P)\|_2 + \frac{1}{\underline{s}(\hat A_{0,n})} (\|\hat M_{n} - M(P)\|_{2,2} \cdot \|b_{z, t}(P)\|_2 + \|\hat b_{z, t, n} - b_{z, t}(P)\|_2).
\end{align}
Therefore, combining result \eqref{lem:var_rate2} with the bounds in \eqref{lem:var_rate7} and \eqref{lem:var_rate8}, and using results \eqref{lem:var_rate4}, \eqref{lem:var_rate5}, \eqref{lem:var_rate6}, Assumption \ref{ass:singular-value}, and Lemma \ref{lem:rates} we obtain uniformly in $P \in \mathbf P_0$ and $(z, t) \in \mathcal H$ that
\begin{equation} \label{lem:var_rate9}
\sup_{\|y\|_2 \leq 1} \|\hat D_{z, t, n}(y) - D_{z, t}(P; y) \|_2 = O_P \Big ( K_{2, p} (\sqrt{\frac{(K_{0, p} \vee K_{1, p}) \log(1 + p)}{n}} + \frac{a_n}{n}) + \sqrt{\frac{K_{1, p}}{n}} \Big ) .
\end{equation}

Next, note that $\mvec(A B C) = (C' \otimes A) \mvec(B)$, $(A \otimes B)' = A' \otimes B'$ for any conformable matrices $A$, $B$, and $C$, the definitions of $\hat D_{z, t, n}(y)$ and $D_{z, t}(P; y)$, and direct calculation allow us to conclude that
\begin{align}\label{lem:var_rate10}
(D_{z, t}(P;y)-& \hat D_{z, t, n}(y))^\prime  \begin{pmatrix} \mvec \Psi_0(Z_i, P) \\ \varphi_{z, t}(Z_i, P) \end{pmatrix} \nonumber \\
= ~& y^\prime (\hat A^\dagger_{0,n}-A_{0}^\dagger(P))^\prime \Psi_0(Z_i,P)^\prime M(P) b_{z, t}(P) + y^\prime (\hat A_{0,n}^\dagger)^\prime \Psi_0(Z_i,P)^\prime (\hat M_{n}\hat b_{z, t, n} - M(P)b_{z, t}(P)) \nonumber \\
& + y^\prime (\hat M_{n} - M(P))\Psi_0(Z_i,P) A^\dagger_0(P) b_{z, t}(P) + y^\prime \hat M_{n} \Psi_0(Z_i,P)(\hat A_{0,n}^\dagger \hat b_{z, t, n} - A_0^\dagger(P) b_{z, t}(P)) \nonumber \\
& + y^\prime (M(P)-\hat M_{n})\varphi_{z, t}(Z_i,P) ~.
\end{align}
Also note that for any $v \in \mathbb R^p$ and $u \in \mathbb R^{\ell}$, it follows from $\|u u^\prime\|_{2, 2} \leq \|u\|_2^2$ and Assumption \ref{ass:prelim-estimators}(b) that
\begin{equation} \label{lem:var_rate11}
E_P[(v^\prime \Psi_0(Z, P) u)^2] = E_P[v^\prime \Psi_0(Z, P) u u^\prime \Psi_0(Z, P)^\prime v] \leq \|u\|_2^2 (v^\prime E_P[\Psi_0(Z, P) \Psi_0(Z, P)^\prime ] v) \leq K_{1, p} \|u\|_2^2 \|v\|_2^2.
\end{equation}
Therefore, using \eqref{lem:var_rate10}, \eqref{lem:var_rate11}, the definition of $V_{z, t}(P)$, Lemma \ref{lem:eigen-cs}, Assumption \ref{ass:prelim-estimators}(b), and that $\|\hat M_{n}\|_{2,2} \leq 1$ due to $\hat M_{n}$ being a projection matrix, we obtain for any $y\in \mathbb R^p$ satisfying $\|y\|_1 \leq 1$ that
\begin{align}\label{lem:var_rate12}
 \|V_{z, t}^{1/2}(P)& (\hat D_{z, t, n}(y) - D_{z, t}(P;y))\|_2^2 \nonumber \\
  \lesssim ~ & K_{1,p}\|(\hat A_{0,n}^\dagger - A_0^\dagger(P))^\prime\|_{2,2}^2 \|b_{z, t}(P)\|_2^2 + K_{1,p}\|(\hat A_{0,n}^\dagger)^\prime\|_{2,2}^2\|\hat M_{n}\hat b_{z, t, n} - M(P) b_{z, t}(P)\|_{2}^2  \nonumber \\ & + K_{1,p}\|\hat M_{n} - M(P)\|_{2,2}^2 (1+\|A_0^\dagger(P) b_{z, t}(P)\|_2^2) + K_{1,p} \|\hat A_{0,n}^\dagger \hat b_{z, t, n} - A_0^\dagger(P) b_{z, t}(P)\|_2^2  ~ .
\end{align}
Further note that the arguments employed in \eqref{lem:var_rate7} and \eqref{lem:var_rate8} imply that uniformly in $P \in \mathbf P_0$ and $(z, t) \in \mathcal H$
\begin{align}\label{lem:var_rate13}
\|\hat M_{n} \hat b_{z, t, n} - M(P)b_{z, t}(P)\|_2 & = O_P\Big( \sqrt{\frac{K_{1,p}}{n}} + K_{2,p}(\sqrt{\frac{(K_{0,p}\vee K_{1,p})\log(1+p)}{n}} + \frac{a_n}{n})\Big) \nonumber \\
\|\hat A_{0,n}^\dagger \hat b_{z, t, n} - A_0^\dagger(P) b_{z, t}(P)\|_2 & = O_P\Big( \sqrt{\frac{K_{1,p}}{n}} + K_{2,p}(\sqrt{\frac{(K_{0,p}\vee K_{1,p})\log(1+p)}{n}} + \frac{a_n}{n}) \Big) ~.
\end{align}
Combining \eqref{lem:var_rate12} and \eqref{lem:var_rate13} and applying Lemma \ref{lem-auxAdagger} and Assumption \ref{ass:singular-value}, we have
\begin{equation} \label{lem:var_rate14}
\|V_{z, t}^{1/2} (P) (\hat D_{z, t, n}(y) - D_{z, t}(P; y)) \|_2 = O_P \Big ( \sqrt{K_{1, p}} K_{2, p} (\sqrt{\frac{(K_{0, p} \vee K_{1, p}) \log(1 + p)}{n}} + \frac{a_n}{n}) \Big )
\end{equation}
uniformly in $P \in \mathbf P_0$ and $(z, t) \in \mathcal H$. Finally, note that Lemma \ref{lem:variance-xi} implies $\|V_{z, t}^{1/2} (P) D_{z, t}(P; y) \|_2^2 = \mathrm{Var}_P[y^\prime \xi_{z, t}(Z, P)]$ for any $y$. Since $\mathrm{Var}_P[y^\prime \xi_{z, t}(Z, P)]$ is uniformly bounded in $P \in \mathbf P_0$, $(z, t) \in \mathcal H$, and $y$ satisfying $\|y\|_1 \leq 1$, the lemma then follows from \eqref{eq:sigma_diff}, \eqref{eq:D_bound}, Assumption \ref{ass:singular-value}, and results \eqref{lem:var_rate9}, \eqref{lem:var_rate14}, and our rate conditions implying $\|\hat D_{z, t, n}(y) - D_{z, t}(P; y) \|_2 \vee \|V_{z, t}^{1/2} (P) (\hat D_{z, t, n}(y) - D_{z, t}(P; y)) \|_2 = o_P(1)$ uniformly in $P \in \mathbf P_0$, $(z, t) \in \mathcal H$ and $y$ with $\|y\|_1 \leq 1$.
\end{proof}

\begin{lemma} \label{lem:rates}
Suppose Assumptions \ref{ass:prelim-estimators} and \ref{ass:singular-value}(a) hold. Then,
\begin{packed_enum}
\item[(a)] If $(K_{0, p} \vee K_{1, p}) \log(p + \ell)/n = o(1)$, then it follows that uniformly in $P \in \mathbf P_0$ we have
\[ \|\hat A_{0, n} - A_0(P)\|_{2, 2} = O_P \left ( \sqrt{\frac{(K_{0, p} \vee K_{1, p}) \log (1 + p)}{n}} + \frac{a_n}{n} \right )~. \]
\item[(b)] Uniformly in $P \in \mathbf P_0$ and $(z, t) \in \mathcal H$,
\[ \|\hat b_{z, t, n} - b_{z, t}(P)\|_{2} = O_P \left ( \sqrt{\frac{K_{1, p}}{n}} + \frac{a_n}{n} \right )~. \]
\end{packed_enum}
\end{lemma}
\begin{proof}
To establish part (a), note that the triangle inequality and Assumption \ref{ass:prelim-estimators}(a) imply that
\[ \|\hat A_{0, n} - A_0(P)\|_{2, 2} \leq \Big \| \frac{1}{n} \sum_{1 \leq i \leq n} \Psi_0(Z_i, P) \Big \| _{2, 2} + O_P \Big ( \frac{a_n}{n} \Big ) \]
uniformly in $P \in \mathbf P_0$. Moreover, Assumption \ref{ass:prelim-estimators}(b) and Theorem 1.6 in \citet{tropp2012user} yield for any $\tau \geq 0$,
\begin{equation} \label{eq:tropp0}
\sup_{P \in \mathbf P_0} P \Big \{ \Big \| \frac{1}{n} \sum_{1 \leq i \leq n} \Psi_0(Z_i, P) \Big \| _{2, 2} \geq \tau \Big \} \leq (p + \ell) \exp \bigg ( \frac{-\tau^2/2}{K_{1, p} / n + K_{0, p} \tau/(3 n)} \bigg )~.
\end{equation}
Next set $\tau = C((K_{0, p} \vee K_{1, p}) \log(p + \ell) / n)^{1/2}$ for any $0 < C < \infty$ and note that $\tau = o(1)$ by hypothesis. Therefore, for $n$ large enough we have that $K_{1, p}/n + K_{0, p} \tau/(3 n) \leq 2 (K_{1, p} \vee K_{0, p})/n$, which yields
\begin{multline} \label{eq:tropp1}
\sup_{P \in \mathbf P_0} P \Big \{ \Big \| \frac{1}{n} \sum_{1 \leq i \leq n} \Psi_0(Z_i, P) \Big \| _{2, 2} \geq C\frac{((K_{0, p} \vee K_{1, p}) \log(p + \ell))^{1/2}}{\sqrt n} \Big \} \\
  \leq (p + \ell) \exp \left ( -\frac{C^2}{4} \log(p + \ell) \right ) = \exp \left ( - \Big ( \frac{C^2}{4} - 1 \Big ) \log(p + \ell) \right )
\end{multline}
for $n$ sufficiently large. Finally, note that $\ell \leq p$ because of Assumption \ref{ass:singular-value}, which implies $\log(p + \ell) \leq 2 \log(1 + p)$. Part (a) then follows from \eqref{eq:tropp0} and letting $C \to \infty$ in \eqref{eq:tropp1}.

To establish part (b), we once again apply the triangle inequality and Assumption \ref{ass:prelim-estimators} to obtain
\[ \|\hat b_{z, t, n} - b_{z, t}(P)\|_{2} \leq \Big \| \frac{1}{n} \sum_{1 \leq i \leq n} \varphi_{z, t}(Z_i, P) \Big \| _{2} + O_P \Big ( \frac{a_n}{n} \Big )~. \]
By independence and $E_P[\varphi_{z, t}(Z_i, P)] = 0$, we have
\begin{equation} \label{eq:orlicz-varphi}
E_P \Big [ \Big \| \frac{1}{ n} \sum_{1 \leq i \leq n} \varphi_{z, t}(Z_i, P) \Big \| _2^2 \Big ]^{1/2} = \frac{1}{\sqrt n} E_P[\varphi_{z, t}(Z_i, P)' \varphi_{z, t}(Z_i, P)]^{1/2}~.
\end{equation}
Therefore, Markov's inequality, result \eqref{eq:orlicz-varphi}, and Assumption \ref{ass:prelim-estimators}(b) imply the bound
\[ P \Big \{ \Big \| \frac{1}{\sqrt n} \sum_{1 \leq i \leq n} \varphi_{z, t}(Z_i, P) \Big \| _2 \geq \tau \Big \} \leq \frac{K_{1, p}}{\tau^2}~. \]
The claim then follows by setting $\tau = CK_{1, p}^{1/2}$ for some large $C$, as desired.
\end{proof}

\begin{lemma} \label{lem:lin_error}
For $\ell \geq 1$, define $D(A_0(P))[H] = - \big( M(P) H A_0^\dagger(P) + (A_0^\dagger(P))'H'M(P) \big)$ and for $(z,t)\in\mathcal H$ set
\begin{align*}
\hat\varepsilon_{z,t,n}(P;y) & := \sqrt n(\hat b_{z,t,n}'\hat M_{n}y-b_{z,t}(P)'M(P)y) - b_{z,t}(P)'D(A_0(P))[\sqrt n(\hat A_{0,n}-A_0(P))]y \\
& \hspace{5em} - \sqrt n(\hat b_{z,t,n}-b_{z,t}(P))'M(P)y~.
\end{align*}
Suppose that $A_0(P)$ has full column rank, $\inf_{P\in\mathbf P_0}\underline s(A_0(P))\geq\underline s>0$, and $\|\hat A_{0,n}-A_0(P)\|_{2,2}\leq\underline s/2$ with probability approaching one uniformly in $P\in\mathbf P_0$.
Then, with probability approaching one uniformly in $P\in\mathbf P_0$ and $(z,t)\in\mathcal H$, it follows that
\[
\sup_{y\in\mathbb R^p:\|y\|_2\leq1}|\hat\varepsilon_{z,t,n}(P;y)| \lesssim \frac{\sqrt n\|\hat A_{0,n}-A_0(P)\|_{2,2}^2\|b_{z,t}(P)\|_2}{\underline s(A_0(P))^2} + \frac{\sqrt n\|\hat A_{0,n}-A_0(P)\|_{2,2}\|\hat b_{z,t,n}-b_{z,t}(P)\|_2}{\underline s(A_0(P))}~.
\]
\end{lemma}

\begin{proof}
For notational convenience, in the proof we suppress the dependence of $A_0(P)$, $M(P)$, and $b_{z, t}(P)$ on $P$ and the subscript $n$ in $\hat A_{0, n}$ and $\hat M_{n}$. We apply the decomposition
\begin{equation} \label{eq:lin_first}
	\hat b_{z, t, n}'\hat M y - b_{z, t}' M y =  b_{z, t}'(\hat M - M)y + (\hat b_{z, t, n} - b_{z, t})'M y + (\hat b_{z, t, n} - b_{z, t})'(\hat{M} - M)y~. 
\end{equation}

Consider the first summand in \eqref{eq:lin_first}. According to Lemma \ref{lem:derivative}, the first differential or directional derivative of a function $f: \mathbb R^{p \times \ell} \to \mathbb R$ defined as $f(A_0) = b_{z, t}'M y$ with respect to $A_0$ in the direction $H$ is $b_{z, t}' D(A_0)[H] y$. By the Mean-Value Theorem (see, e.g., Theorem 5.10 in \cite{magnus2019matrix}), $f(\hat A_0) - f(A_0) = b_{z, t}^\prime D(\tilde{A}_0)[\hat A_0 - A_0]y,$ where $\tilde{A}_0 = \tau\hat A_0 + (1-\tau)A_0$ for some $\tau \in [0, 1]$, so we can write
\begin{equation}\label{eq:lin_firstp1}
 b_{z, t}'(\hat M - M)y = b_{z, t}'D(A_0)[\hat A_0 - A_0]y + \underbrace{b_{z, t}'(D(\tilde{A}_0)[\hat A_0 - A_0] - D(A_0)[\hat A_0 - A_0])y}_{(R)}.
\end{equation}
Let $\tilde M := \mathbf I_p - \tilde A_0 \tilde A_0^\dagger$. To bound the term $(R)$ in result \eqref{eq:lin_firstp1} uniformly in $P\in \mathbf P_0$ we decompose it into the terms
\begin{align*}
(R) &= -b_{z, t}'\{ \tilde{M}(\hat A_0 - A_0) \tilde{A}_0^{\dagger} - M(\hat A_0 - A_0) A_0^{\dagger} \} y && (R.1) \\
	& \hspace{2em} - b_{z, t}'\{ (\tilde{A}_0^{\dagger})' (\hat A_0 - A_0)' \tilde{M} - (A_0^{\dagger})'(\hat A_0 - A_0)'M \} y~. && (R.2)
\end{align*}
First examining the term $(R.1)$, we expand it further to obtain the decomposition:
\begin{align*}
	(R.1) &= - b_{z, t}' (\tilde{M} - M) (\hat A_0 - A_0) A_0^{\dagger} y && (R.1.A) \\
	& \hspace{2em} - b_{z, t}'M(\hat A_0 - A_0)(\tilde{A}_0^{\dagger} - A_0^{\dagger})y && (R.1.B) \\
	& \hspace{2em} - b_{z, t}'(\tilde{M} - M)(\hat A_0 - A_0)(\tilde{A}_0^{\dagger} - A_0^{\dagger})y && (R.1.C) ~.  
\end{align*}
Next note that the Cauchy-Schwarz inequality implies that $|a'b| \leq \|a\|_{2} \|b\|_2$ for any $a$ and $b$.
We therefore obtain that
\begin{equation} \label{eq:R1A}
\begin{aligned}
	|(R.1.A)| &\leq\|(A_0^{\dagger})' (\hat A_0 - A_0)'(\tilde{M} - M)b_{z, t}\|_{2}\cdot \|y\|_2\\
	&\leq \|(A_0^{\dagger})'\|_{2, 2} \cdot\|(\hat A_0 - A_0)'\|_{2, 2} \cdot \|\tilde{M} - M\|_{2, 2} \cdot \|b_{z, t}\|_{2} \cdot \|y\|_2~.
\end{aligned}
\end{equation}
Similarly, we have
\begin{equation} \label{eq:R1B}
\begin{aligned}
	|(R.1.B)| &\leq \|(\tilde{A}_0^{\dagger} - A_0^{\dagger})' (\hat A_0 - A_0)'M b_{z, t}\|_{2}\cdot \|y\|_2\\
	&\leq \|(\tilde{A}_0^{\dagger} - A_0^{\dagger})'\|_{2, 2} \cdot \|(\hat A_0 - A_0)'\|_{2, 2} \cdot \|M\|_{2, 2} \cdot \|b_{z, t}\|_{2} \cdot \|y\|_2~,
\end{aligned}
\end{equation} 
and
\begin{equation}\label{eq:R1C}
\begin{aligned}
    |(R.1.C)| &\leq \|(\tilde{A}_0^{\dagger} - A_0^{\dagger})' (\hat A_0 - A_0)'(\tilde M - M) b_{z, t}\|_{2}\cdot \|y\|_2\\
	&\leq \|(\tilde{A}_0^{\dagger} - A_0^{\dagger})'\|_{2, 2} \cdot \|(\hat A_0 - A_0)'\|_{2, 2} \cdot \|\tilde M - M\|_{2, 2} \cdot \|b_{z, t}\|_{2} \cdot \|y\|_2~.    
\end{aligned}
\end{equation}
Proceeding to bound the term $(R.2)$ in our expansion for $(R)$, we first decompose $(R.2)$ into the three terms
\begin{align*}
	(R.2) &= - b_{z, t}'  (\tilde{A}_0^{\dagger} - A_0^{\dagger})'(\hat A_0 - A_0)' M y && (R.2.A) \\
	& \hspace{2em} - b_{z, t}'(A_0^{\dagger})'(\hat A_0 - A_0)'(\tilde{M} - M)y  && (R.2.B) \\
	& \hspace{2em} - b_{z, t}'(\tilde{A}_0^{\dagger} - A_0^{\dagger})'(\hat A_0 - A_0)' (\tilde{M} - M)y && (R.2.C)  ~.
\end{align*}
Using similar arguments to those employed in deriving \eqref{eq:R1A}, \eqref{eq:R1B}, and \eqref{eq:R1C}, we then obtain the bounds
\begin{equation} \label{eq:R2A}
\begin{aligned}
	|(R.2.A)| & \leq \|M (\hat A_0 - A_0)(\tilde{A}_0^\dagger - A_0^\dagger) b_{z, t}\|_{2} \cdot \|y\|_2\\
	& \leq \|M\|_{2, 2} \cdot \|\hat A_0 - A_0\|_{2, 2} \cdot \|\tilde{A}_0^\dagger - A_0^\dagger\|_{2, 2} \cdot \|b_{z, t}\|_{2} \cdot \|y\|_2
\end{aligned}
\end{equation}
and
\begin{equation} \label{eq:R2B}
\begin{aligned}
	|(R.2.B)| &\leq\| (\tilde{M} - M) (\hat A_0 - A_0) A_0^\dagger b_{z, t}\|_{2} \|y\|_2\\
	& \leq \|\tilde{M} - M\|_{2, 2} \cdot \|\hat A_0 - A_0\|_{2, 2} \cdot \|A_0^\dagger\|_{2, 2} \cdot \|b_{z, t}\|_{2} \cdot \|y\|_2~,
\end{aligned}
\end{equation}
and
\begin{equation} \label{eq:R2C}
\begin{aligned}
	|(R.2.C)| &\leq\| (\tilde{M} - M) (\hat A_0 - A_0) (\tilde A_0^\dagger - A_0^\dagger) b_{z, t}\|_{2} \|y\|_2\\
	& \leq \|\tilde{A}_0^\dagger - A_0^\dagger\|_{2, 2} \cdot \|\hat A_0 - A_0\|_{2, 2} \cdot \|\tilde M - M\|_{2, 2} \cdot \|b_{z, t}\|_{2} \cdot \|y\|_2~.
\end{aligned}
\end{equation}
Next note that for every $\tilde A_0 = \tau\hat A_0 + (1 - \tau)A_0$ with $\tau \in [0, 1]$, Lemma \ref{lem:weyl} applied to $\tilde A_0$ and $A_0$ implies
\begin{equation}\label{eq:perturb}
\underline s(\tilde A_0) \geq \underline s(A_0) - \|\tilde A_0 - A_0\|_{2,2} \geq \underline s(A_0) - \|\hat A_0 - A_0\|_{2,2}~,
\end{equation}
where in the final inequality we used that $\|\tilde A_0 - A_0\|_{2,2} \leq \|\hat A_0 - A_0\|_{2,2}$ due to $\tilde A_0$ being a convex combination of $\hat A_0$ and $A_0$.
Since, by hypothesis, $\|\hat A_0 - A_0\|_{2,2} \leq \underline s/2 \leq \underline s(A_0)/2$ with probability tending to one uniformly in $P\in \mathbf P_0$, we obtain from \eqref{eq:perturb} that with probability tending to one
\begin{equation}\label{eq:lower}
\underline s(\tilde A_0) \geq \frac{\underline s(A_0)}{2} > 0
\end{equation}
uniformly in $P\in \mathbf P_0$ and $\tau \in [0,1]$.
In particular, result \eqref{eq:lower} implies that $\mathrm{rank}\{\tilde A_0\} = \mathrm{rank}\{A_0\} = \ell$ with probability tending to one, which together with Theorem 4.1 in \cite{wedin1973perturbation} gives us
\begin{align}
\nonumber \|\tilde{A}_0^\dagger - A_0^\dagger\|_{2, 2} & \lesssim \|\tilde{A}_0^\dagger\|_{2, 2} \cdot \|A_0^\dagger\|_{2, 2} \cdot \|\tilde{A}_0 - A_0\|_{2, 2} \\
\label{eq:daggerdiff22} & \leq \frac{1}{\underline{s}(\tilde{A}_0)} \frac{1}{\underline s(A_0)} \|\hat A_0 - A_0\|_{2, 2} \lesssim \frac{1}{\underline s(A_0)^2} \|\hat A_0 - A_0\|_{2, 2}~,
\end{align}
where the final inequality holds with probability tending to one uniformly in $P \in \mathbf P_0$ by \eqref{eq:lower}.
Furthermore, Theorem 2.5 in \cite{chen2016perturbation} and $\|\tilde A_0 - A_0\|_{2, 2} \leq \|\hat A_0 - A_0\|_{2, 2}$ imply that
\begin{equation} \label{eq:Mdiff22}
\|\tilde{M} - M\|_{2, 2} \leq \min\{\|\tilde{A}_0^\dagger\|_{2, 2}, \|A_0^{\dagger}\|_{2, 2}\} \cdot \|\hat A_0 - A_0\|_{2, 2} \lesssim \frac{1}{\underline s(A_0)} \|\hat A_0 - A_0\|_{2, 2},
\end{equation}
where the final inequality holds with probability tending to one uniformly in $P \in \mathbf P_0$ by Lemma \ref{lem-auxAdagger} and result \eqref{eq:lower}.
Further note that $\|S\|_{2,2} = \|S^\prime\|_{2,2}$ for any matrix $S$ by Theorem 6.5.1 in \cite{luenberger1969optimization}.
Since $\|M\|_{2,2}\leq 1$ due to $M$ being a projection matrix, and $\|\hat A_0 - A_0\|_{2,2} \leq \underline{s}(A_0)/2$ with probability tending to one uniformly in $P\in \mathbf P_0$ by hypothesis, we can then combine the bounds in results \eqref{eq:R1A}, \eqref{eq:R1B}, \eqref{eq:R1C}, \eqref{eq:R2A}, \eqref{eq:R2B}, and \eqref{eq:R2C} with results \eqref{eq:daggerdiff22} and \eqref{eq:Mdiff22} to conclude that
\begin{equation}\label{eq:linrem1}
|(R)| \lesssim \frac{\|\hat A_0 - A_0\|_{2, 2}^2 \|b_{z, t}\|_{2}}{\underline s(A_0)^2}
\end{equation}
with probability approaching one uniformly in $P \in \mathbf P_0$ and $(z, t) \in \mathcal H$.
Moreover, by similar arguments
\begin{multline}\label{eq:linrem2}
	|(\hat b_{z, t, n} - b_{z, t})'(\hat M - M)y| \leq \|(\hat{M} - M)(\hat b_{z, t, n} - b_{z, t})\|_{2} \|y\|_2\\
	\leq \|\hat M - M\|_{2, 2} \|\hat b_{z, t, n} - b_{z, t}\|_{2} \|y\|_2
	\lesssim \frac{\|\hat A_0 - A_0\|_{2, 2} \|\hat b_{z, t, n} - b_{z, t}\|_2}{\underline s(A_0)},
\end{multline}
where the final inequality holds uniformly in $P\in \mathbf P_0$ by \eqref{eq:Mdiff22} with $\tau=1$.
The claim of the lemma therefore follows from the decomposition in equations \eqref{eq:lin_first} and \eqref{eq:lin_firstp1}, and the bounds in \eqref{eq:linrem1} and \eqref{eq:linrem2}.
\end{proof}

\begin{lemma} \label{lem:coupling}
Let $n \geq 2$, let $\{Z_i\}_{i = 1}^n$ be i.i.d.\ with marginal distribution $P$, and let $\mathbb Z \sim N(0, 1)$. Suppose Assumption \ref{ass:bdd-moments-IF} holds and $E_P[\xi_{z, t}(Z_i, P)] = 0$ for every $P \in \mathbf P_0$ and $(z, t) \in \mathcal H$. Define $S_{z, t, n}^*(P) := n^{-1/2} \sum_{1 \leq i \leq n} \xi_{z, t}(Z_i, P)$. Then, it follows that for some universal $C < \infty$ and any $\delta > 0$, $x \in \mathbb R$
\begin{align*}
\sup_{P \in \mathbf P_0} \sup_{y \in \mathbb R^p : \|y\|_1 \leq 1} \sup_{(z, t) \in \mathcal H} \Big \{ P \left \{ S_{z, t, n}^*(P)' y \leq x \right \} - P \left ( \sigma_{z, t} (P; y) \mathbb Z \leq x + 3\delta \right ) \Big \} & \leq \Big ( \frac{K_\xi}{\delta} \vee 1 \Big )^3 \frac{C\log(n)}{\sqrt n} \\
\sup_{P \in \mathbf P_0} \sup_{y \in \mathbb R^p : \|y\|_1 \leq 1} \sup_{(z, t) \in \mathcal H} \Big \{ P \left ( \sigma_{z, t} (P; y) \mathbb Z \leq x - 3\delta \right ) - P \left \{ S_{z, t, n}^*(P)' y \leq x \right \} \Big \} & \leq \Big ( \frac{K_\xi}{\delta} \vee 1 \Big )^3 \frac{C\log(n)}{\sqrt n}.
\end{align*}
\end{lemma}

\begin{proof}
First note that for any $y \in \mathbb R^p$, $\eta_{z, t, i} (P; y) := \xi_{z, t}(Z_i, P)^\prime y$ is an i.i.d.\ sequence satisfying $E_P[\eta_{z, t, i} (P; y)] = 0$ and $\mathrm{Var}_P[\eta_{z, t, i} (P; y)] = \sigma_{z, t}^2(P; y)$ by Lemma \ref{lem:variance-xi}. Further note that Assumption \ref{ass:bdd-moments-IF} implies that
\begin{equation} \label{lem:couplin1}
\sup_{P \in \mathbf P_0} \sup_{y \in \mathbb R^p : \|y\|_1 \leq 1} \sup_{(z, t) \in \mathcal H} \sum_{i = 1}^n E_P \Big [ \Big | \frac{\eta_{z, t, i} (P; y)}{\sqrt n} \Big | ^3 \Big ] \leq \frac{K_\xi^3}{\sqrt n} .
\end{equation}
Letting $g_{z, t, i}(P; y)$, $1 \leq i \leq n$ be independent random variables with $g_{z, t, i}(P; y) \sim N(0, \sigma_{z, t}^2(P; y)/n)$ and $\mathbb Z \sim N(0, 1)$, next use that $E[|\mathbb Z|^3] \leq 2$ to conclude that
\begin{equation} \label{lem:couplin2}
\sup_{P \in \mathbf P_0} \sup_{y \in \mathbb R^p : \|y\|_1 \leq 1} \sup_{(z, t) \in \mathcal H} \sum_{i = 1}^n E[|g_{z, t, i}(P; y)|^3] \leq \sup_{P \in \mathbf P_0} \sup_{y \in \mathbb R^p : \|y\|_1 \leq 1} \sup_{(z, t) \in \mathcal H} \frac{2 \sigma_{z, t}^3(P; y)}{\sqrt n} \leq \frac{2 K_\xi^3}{\sqrt n},
\end{equation}
where in the final inequality we used that $\sigma_{z, t}(P; y) = (E_P[\eta_{z, t, i}^2(P; y)])^{1/2} \leq (E_P[|\eta_{z, t, i}(P; y)|^3])^{1/3}$ and Assumption \ref{ass:bdd-moments-IF}.
Combining results \eqref{lem:couplin1} and \eqref{lem:couplin2} and applying Lemma 38 in \cite{belloni2019conditional}, we then obtain that for each $\delta > 0$ there exists a random variable $T_{z, t}(P; y) \sim N(0, \sigma_{z, t}^2(P; y))$ satisfying
\begin{multline} \label{lem:coupling3}
\sup_{P \in \mathbf P_0} \sup_{y \in \mathbb R^p : \|y\|_1 \leq 1} \sup_{(z, t) \in \mathcal H} P \Big \{ \Big | \frac{1}{\sqrt n} \sum_{i = 1}^n \eta_{z, t, i}(P; y) - T_{z, t}(P; y) \Big | > 3\delta \Big \} \\
\leq \inf_{u \geq 0} \Big ( 2 P \{|\mathbb Z| > u\} + \frac{3 K_\xi^3}{\delta^3 \sqrt n} u^2 \Big ) \leq \Big ( \frac{K_\xi}{\delta} \vee 1 \Big )^3 \frac{C \log(n)}{\sqrt n},
\end{multline}
where the final inequality holds for some $C < \infty$ by setting $u = \sqrt{\log(n)}$ and using the bound $P \{|\mathbb Z| > u\} \leq 2 \exp\{-u^2/2\}$.
Since $S_{z, t, n}^*(P)^\prime y = \sum_{1 \leq i \leq n} \eta_{z, t, i}(P; y)/\sqrt n$, we obtain from result \eqref{lem:coupling3} and Strassen's theorem (see, e.g., Theorem 10.3.8 in \cite{pollard2002user}) that for any $\delta > 0$ and any $x \in \mathbb R$ we have that
\begin{align*}
\sup_{P \in \mathbf P_0} \sup_{y \in \mathbb R^p : \|y\|_1 \leq 1} \sup_{(z, t) \in \mathcal H} \Big \{ P \left \{ S_{z, t, n}^*(P)' y \leq x \right \} - P \left \{ T_{z, t} (P; y) \leq x + 3\delta \right \} \Big \} & \leq \Big ( \frac{K_\xi}{\delta} \vee 1 \Big )^3 \frac{C\log(n)}{\sqrt n} \\
\sup_{P \in \mathbf P_0} \sup_{y \in \mathbb R^p : \|y\|_1 \leq 1} \sup_{(z, t) \in \mathcal H} \Big \{ P \{T_{z, t} (P; y) \leq x - 3\delta\} - P \left \{ S_{z, t, n}^*(P)' y \leq x \right \} \Big \} & \leq \Big ( \frac{K_\xi}{\delta} \vee 1 \Big )^3 \frac{C\log(n)}{\sqrt n}.
\end{align*}
The claim of the lemma then follows from $T_{z, t} (P; y) \sim N(0, \sigma_{z, t}^2(P; y))$.
\end{proof}

\begin{lemma} \label{lem:derivative}
For any $A\in \mathbb R^{p\times \ell}$ with $\mathrm{rank}(A) = \ell$ define the function $A\mapsto M(A) := \mathbf I_p - A (A' A)^{-1} A'$.
Then, $M$ is differentiable at $A$ and its derivative in the direction $H\in \mathbb R^{p\times \ell}$ equals
\[ D(A)[H] = - M(A) H (A' A)^{-1} A' - A (A' A)^{-1} H' M(A)~. \]
Therefore, for any $b, y \in \mathbb R^p$, $A \mapsto b' M(A) y$ is differentiable at any $A$ with ${\rm rank}(A) =\ell$ and its derivative is $H\mapsto b' D(A)[H] y$.
\end{lemma}

\begin{proof}
For any invertible $\ell\times \ell$ matrix $S$, Theorem 8.3 in \cite{magnus2019matrix} implies that the derivative of the inverse map $S\mapsto S^{-1}$ is given by $H\mapsto -S^{-1}H S^{-1}$.
Moreover, by the chain rule, the derivative of $A\mapsto A'A$ is given by $H\mapsto A'H+ H'A$.
Therefore, once again applying the chain rule we obtain that the derivative of $A\mapsto M(A)$ at any $A$ with ${\rm rank}(A) = \ell$ is given by
\begin{align*}
D(A)[H] & = -H(A'A)^{-1}A^\prime - A(A'A)^{-1}H' + A(A'A)^{-1}(H'A + A'H)(A'A)^{-1}A'\\
& = -M(A)H(A'A)^{-1}A' - A (A^\prime A)^{-1} H^\prime M(A),
\end{align*}
as desired. The fact that the derivative of $b^\prime M(A)y$ equals $H\mapsto b^\prime D(A)[H] y$ follows from a second application of the chain rule. 
\end{proof}

\begin{lemma} \label{lem:lipschitz}
If $c_1 \ge \underline \sigma^2 > 0$ and $c_2 \ge \underline \sigma^2$, then
\[ |c_1^{-1/2} - c_2^{-1/2}| \leq \frac{1}{2\underline \sigma^3}|c_1 - c_2|. \]
\end{lemma}

\begin{proof}
Consider the function $f(x) = x^{-1/2}$. For $x \ge \underline \sigma^2$, we have $|f'(x)| = \frac{1}{2} x^{-3/2} \le \frac{1}{2} (\underline \sigma^2)^{-3/2} = 1/(2\underline \sigma^3)$. The conclusion now follows from the mean-value theorem.
\end{proof}

\begin{lemma} \label{lem:eigen-cs}
Let $X \in \mathbb R^p$ be a random vector and suppose $E[X X'] < \infty$. Then, $\|E[X X'] \|_{2, 2} \leq E[X' X]$.
\end{lemma}

\begin{proof}
Because $E[X X']$ is symmetric and positive semi-definite, all its eigenvalues are non-negative and $\|E[X X^\prime] \|_{2, 2}$ equals the largest eigenvalue of $E[X X']$. Using that $\mathrm{trace}\{E[X X'] \}$ equals the sum of the eigenvalues of $E[X X']$, we can then conclude that $\|E[X X'] \|_{2, 2} \leq \mathrm{trace}\{E[X X'] \} = E[X' X]$.
\end{proof}

\begin{lemma} \label{lem-auxAdagger}
Let $A_0$ be an arbitrary $p\times \ell$ matrix of rank $\ell \geq 1$. Then, $\|A_0^\dagger\|_{2, 2} = \|(A_0^\dagger)^\prime\|_{2, 2} = 1/\underline{s} (A_0)$.
\end{lemma}

\begin{proof}
The claim that $\|A_0^\dagger\|_{2, 2} = \|(A_0^\dagger)^\prime\|_{2, 2}$ follows from Theorem 6.5.1 in \cite{luenberger1969optimization}. Since $A_0$ has rank $\ell$ it follows from Proposition 6.11.1 in \cite{luenberger1969optimization} that $A_0^\dagger = (A_0^\prime A_0)^{-1} A_0^\prime$. Hence, we have
\begin{equation}
\|(A_0^{\dagger})^\prime\|_{2, 2} = \sup_{\|x\|_2 \leq 1} (x'(A_0' A_0)^{-1} A_0' A_0(A_0' A_0)^{-1} x)^{1/2} = \sup_{\|x\|_2 \leq 1} (x'(A_0' A_0)^{-1} x)^{1/2} = \frac{1}{\underline s(A_0)}~,
\end{equation}
where the first equality follows by definition of $\| \cdot \|_{2, 2}$ and the final one from $\|(A_0^\prime A_0)^{-1} \|_{2, 2} = 1/\underline{s} (A_0)^2$.
\end{proof}

\begin{lemma} \label{lem:weyl}
Suppose $B$ and $D$ are $p \times \ell$ matrices, where $1 \leq \ell \leq p$. Then, $|\underline s(B) - \underline s(D)| \leq \|B - D\|_{2, 2}$.
\end{lemma}

\begin{proof}
Since $\ell \le p$,
\[
\underline{s}(B)=\min_{\|x\|_2=1}\|Bx\|_2~.
\]
For every $x \in \mathbb{R}^{\ell}$ with $\|x\|_2=1$, 
\[
\|Bx\|_2
\ge \|Dx\|_2-\|(B-D)x\|_2
\ge \underline{s}(D)-\|B-D\|_{2,2}~.
\]
Taking the minimum over such $x$ yields
\[
\underline{s}(B)-\underline{s}(D)
\ge -\|B-D\|_{2,2}~.
\]
Interchanging $B$ and $D$ gives the reverse bound, and hence
\[
\left|\underline{s}(B)-\underline{s}(D)\right|
\le \|B-D\|_{2,2}~,
\]
as desired.
\end{proof}

\section{Computation of the Test} \label{sec:LP}

\subsection{Computing the Direction}
The following lemma shows that a solution to \eqref{eq:yhat-problem} can be computed using a quadratic program:

\begin{lemma} \label{lem:qp}
The problem
\begin{equation} \label{eq:vhat}
\min_{v \in \{\hat M_{n}^{(1)} \hat b_{z, t, n}^{(1)}: (z, t) \in \mathcal H\}} \|v\|_2^2
\end{equation}
has a unique solution. Denoting it by $\hat v_n^{(1)}$ and defining
\begin{equation} \label{eq:ytilde}
\tilde y_n^{(1)} :=
\begin{cases}
\hat v_n^{(1)} / \|\hat v_n^{(1)}\|_2, &\text{ if } \hat v_n^{(1)} \neq 0 \\
0, &\text{ if } \hat v_n^{(1)} = 0~,
\end{cases}
\end{equation}
we have that $\tilde y_n^{(1)}$ solves \eqref{eq:yhat-problem}.
\end{lemma}

\begin{proof}
First recall from Lemma \ref{lem:normalization} that $\mathcal H$ is nonempty, compact, and a polyhedron and hence convex. Therefore, the set $\big \{ \hat M_{n}^{(1)} \hat b_{z, t, n}^{(1)}: (z, t) \in \mathcal H \big \}$ as a linear image of $\mathcal H$ is also nonempty, compact, and convex. Because $\mathcal H$ is compact and $v \mapsto \|v\|_2^2$ is continuous, \eqref{eq:vhat} has a solution. Moreover, since $v \mapsto \|v\|_2^2$ is strictly convex, the solution is unique. As in the statement of the lemma, we can denote it by $\hat v_n^{(1)}$. Next, we show
\begin{equation} \label{eq:min-norm-duality}
\max_{y \in \mathbb R^p: \|y\|_2 \leq 1} \min_{(z, t) \in \mathcal H} (\hat b_{z, t, n}^{(1)})' \hat M_{n}^{(1)} y = \min_{(z, t) \in \mathcal H} \|\hat M_{n}^{(1)} \hat b_{z, t, n}^{(1)}\|_2~.  
\end{equation}
If the minimum on the right-hand side of \eqref{eq:min-norm-duality} is strictly positive, then \eqref{eq:min-norm-duality} follows from Theorem 5.13.1 in \cite{luenberger1969optimization} applied to $x_1 = 0$ and $K = \big \{ \hat M_{n}^{(1)} \hat b_{z, t, n}^{(1)}: (z, t) \in \mathcal H \big \}$; if the minimum is equal to zero, then $\min_{(z, t) \in \mathcal H} (\hat b_{z, t, n}^{(1)})' \hat M_{n}^{(1)} y \leq 0 \cdot y = 0$ for all $y$, and hence the left-hand side of \eqref{eq:min-norm-duality} equals zero by setting $y = 0$. In either case, \eqref{eq:min-norm-duality} holds.

Suppose first that $\hat v_n^{(1)} \neq 0$. Theorem 5.13.1 in \cite{luenberger1969optimization} implies that the maximum on the left-hand side of \eqref{eq:min-norm-duality} is attained by some $y^\ast$ with $\|y^\ast\|_2 \leq 1$. Therefore, \eqref{eq:min-norm-duality}, $\hat v_n^{(1)} \in \big \{ \hat M_{n}^{(1)} \hat b_{z, t, n}^{(1)}: (z, t) \in \mathcal H \big \}$, and the Cauchy--Schwarz inequality imply
\[ \|\hat v_n^{(1)}\|_2 = \min_{(z, t) \in \mathcal H} (\hat b_{z, t, n}^{(1)})' \hat M_{n}^{(1)} y^\ast \leq (\hat v_n^{(1)})' y^\ast \leq \|\hat v_n^{(1)}\|_2 \|y^\ast\|_2 \leq \|\hat v_n^{(1)}\|_2~, \]
which implies $\|y^\ast\|_2 = 1$ and $y^\ast = \hat v_n^{(1)} / \|\hat v_n^{(1)}\|_2 = \tilde y_n^{(1)}$ because $\hat v_n^{(1)} \neq 0$.
If instead $\hat v_n^{(1)} = 0$, then \eqref{eq:min-norm-duality} shows its left-hand side is zero, which is attained by $\tilde y_n^{(1)} = 0$. In either case, $\tilde y_n^{(1)}$ solves \eqref{eq:yhat-problem}.
\end{proof}

Expressing the constraints in matrix form, we have $\mathcal H = \{(z, t): B^{\mathsf e} (z', t)' = c, B^{\mathsf i} (z', t)' \leq 0\}$ for
\[
B^{\mathsf e}  := \begin{pmatrix}A^{\mathsf{k}, \mathsf{e}} & -\beta^{\mathsf{k}, \mathsf{e}} \\
L' & 0 \\
-\mathbf 1' A^{\mathsf{k}, \mathsf{i}} & 1 + \mathbf 1' \beta^{\mathsf{k}, \mathsf{i}}\end{pmatrix}~,
\qquad
c := \begin{pmatrix}
    0 \\ 0 \\ 1
\end{pmatrix}~,
\qquad
B^{\mathsf i} := \begin{pmatrix}0 & -1 \\
A^{\mathsf{k}, \mathsf{i}} & -\beta^{\mathsf{k}, \mathsf{i}}\end{pmatrix}~.
\]
Lemma \ref{lem:qp} implies that we can use the following quadratic program to get $\hat v_n^{(1)}$:
\begin{align}
\label{eq:direction-QP} \min_{v, z, t}~ & \|v\|_2^2 \\
\nonumber \text{s.t.}~ & v = \hat M_n^{(1)} (\hat A_n^{(1)}, - \hat \beta_n^{(1)}) \begin{pmatrix}
    z \\ t
\end{pmatrix} \\
\nonumber & B^{\mathsf e} \begin{pmatrix}
    z \\ t
\end{pmatrix} = c \\
\nonumber & B^{\mathsf i} \begin{pmatrix}
    z \\ t
\end{pmatrix} \leq 0~.
\end{align}
The definitions in \eqref{eq:ytilde} and \eqref{eq:yhat} imply that
\[ \hat y_n^{(1)} = \begin{cases}
\hat v_n^{(1)} / \|\hat v_n^{(1)}\|_1, &\text{ if } \hat v_n^{(1)} \neq 0 \\
0, &\text{ if } \hat v_n^{(1)} = 0~.
\end{cases} \]

\subsection{Computing the Test}
If $\hat y_n^{(1)} = 0$, then $\phi_n = 0$. Otherwise, fix $\hat y_n^{(1)} \neq 0$ and let $\mathcal V$ be the set of vertices of $\mathcal H$. Since $\hat\sigma_{z, t, n}^{(2)}(\hat y_n^{(1)}) > 0$, we have
\[ T_n > z_{1 - \alpha} \iff \min_{(z, t) \in \mathcal H} \big \{ \sqrt{n_2} (\hat b_{z, t, n}^{(2)})' \hat M_n^{(2)} \hat y_n^{(1)} - z_{1 - \alpha} \hat\sigma_{z, t, n}^{(2)}(\hat y_n^{(1)}) \big \} > 0~. \]
Since $\hat\sigma_{z, t, n}^{(2)}(\hat y_n^{(1)})$ is convex in $(z, t)$, the expression in braces is concave.
By Lemma \ref{lem:normalization}, $\mathcal H$ is a compact polyhedron and hence is the convex hull of $\mathcal V$ \citep[Theorems 18.5 and 19.1]{rockafellar1970convex}. By concavity, the minimum is attained at a vertex, and the vertices can be enumerated using algorithms such as \cite{avis1991pivoting}. Therefore, $T_n > z_{1 - \alpha}$ if and only if
\[ \min_{(z, t) \in \mathcal V} \big \{ \sqrt{n_2} (\hat b_{z, t, n}^{(2)})' \hat M_n^{(2)} \hat y_n^{(1)} - z_{1 - \alpha} \hat\sigma_{z, t, n}^{(2)}(\hat y_n^{(1)}) \big \} > 0~. \]

\section{Preprocessing Using Known Coefficient Restrictions} \label{sec:coefficient-preprocessing}

Known restrictions on $A$ can sometimes be used to rewrite the system $A x = b$, $x \in \mathcal X$, as an equivalent system with variables ranging over a known compact polyhedron.

\begin{assumption} \label{ass:pp}
Suppose the following hold:
\begin{packed_enum}
\item[(a)] There is a known point $x_0 \in \mathcal X \cap \lin(\mathcal X)^\perp$ such that $\mathcal X \cap \lin(\mathcal X)^\perp = x_0 + (\rec(\mathcal X) \cap \lin(\mathcal X)^\perp)$.
\item[(b)] $M_A A v \neq 0$ for every $v \in \rec(\mathcal X) \cap \lin(\mathcal X)^\perp \setminus \{0\}$.
\item[(c)] There is a known closed pointed convex cone $\mathcal K \subseteq \mathbb R^p$ with $\mathcal K \neq \{0\}$ such that $M_A A v \in \mathcal K$ for every $v \in \rec(\mathcal X) \cap \lin(\mathcal X)^\perp$.
\item[(d)] There is a known $\nu \in \mathbb R^p$ such that $\nu' w > 0$ for every $w \in \mathcal K \setminus \{0\}$.
\end{packed_enum}
\end{assumption}

Assumption \ref{ass:pp}(a) holds if all known inequalities bind at a known point $x_0 \in \mathcal X \cap \lin(\mathcal X)^\perp$, i.e., $A^{\mathsf{k}, \mathsf{i}} x_0 = \beta^{\mathsf{k}, \mathsf{i}}$. Indeed, \eqref{eq:X} and \eqref{def:rec} show that $x \in \mathcal X \cap \lin(\mathcal X)^\perp$ if and only if $x - x_0 \in \rec(\mathcal X) \cap \lin(\mathcal X)^\perp$. If $\mathcal X \cap \lin(\mathcal X)^\perp$ is already a cone, then we can take $x_0 = 0$.
Assumption \ref{ass:pp}(b) is equivalent to requiring $A u \neq 0$ for every $u \in \rec(\mathcal X) \setminus \lin(\mathcal X)$, as shown in the proof of Lemma \ref{lem:auxrem} (see \eqref{lem:auxrem1} and the surrounding argument). The calculations in the following examples verify Assumption \ref{ass:pp}(b) for $A = A(P)$ for every $P \in \mathbf P$. Hence, Lemma \ref{lem:closure}(ii) implies $\mathbf P_0 = \bar{\mathbf P}_0$ in these examples, even though $\mathcal X$ is unbounded.
Assumption \ref{ass:pp}(c)--(d) can be directly verified in examples as well.

Using the normalization from Section \ref{sec:farkas}, we define the set of normalized recession directions by
\[ \mathcal D := \{z: (z, 0) \in \mathcal H\} = \{z \in \rec(\mathcal X) \cap \lin(\mathcal X)^\perp: -\mathbf 1' A^{\mathsf{k}, \mathsf{i}} z = 1\}~. \]
The normalization uses the linear part of $h$, since $h(x_0 + v) - h(x_0) = -\mathbf 1' A^{\mathsf{k}, \mathsf{i}} v$.

\begin{lemma} \label{lem:pp-cone}
Suppose Assumption \ref{ass:pp} holds. Define
\[ F(A, b) := \{\nu' M_A (b - A x_0)\} M_A A - M_A (b - A x_0) \{\nu' M_A A\}~. \]
Then,
\begin{packed_enum}
\item[(i)] If $\rec(\mathcal X) \cap \lin(\mathcal X)^\perp = \{0\}$, then $(A, b) \in \mathbf C_0$ if and only if $M_A (b - A x_0) = 0$.
\item[(ii)] If $\rec(\mathcal X) \cap \lin(\mathcal X)^\perp \neq \{0\}$, then $\mathcal D$ is a nonempty compact polyhedron, and $(A, b) \in \mathbf C_0$ if and only if
\begin{equation} \label{eq:pp-directional-cone}
M_A (b - A x_0) \in \mathcal K \text{ and } F(A, b) z = 0 \text{ for some } z \in \mathcal D~.
\end{equation}
\item[(iii)] If $\rec(\mathcal X) \cap \lin(\mathcal X)^\perp \neq \{0\}$, then \eqref{eq:pp-directional-cone} holds if and only if there exists $(z, s) \in \mathcal D \times [0, 1]$ satisfying
\begin{align}
\label{eq:pp-input} F(A, b) z & = 0 \\
\label{eq:pp-membership} \{1 + \nu' M_A (b - A x_0)\} s & = \nu' M_A (b - A x_0)~.
\end{align}
\end{packed_enum}
\end{lemma}

\begin{proof}
By Assumption \ref{ass:pp}(a) and $\mathbf C_0 = \mathbf C_1$ for $\mathbf C_1$ defined in \eqref{lem:annihilator1}, we see $(A, b) \in \mathbf C_0$ if and only if
\begin{equation} \label{eq:pp-cone-system}
M_A A v = M_A (b - A x_0) \text{ for some } v \in \rec(\mathcal X) \cap \lin(\mathcal X)^\perp~.
\end{equation}
Part (i) follows upon taking $v = 0$.

For part (ii), suppose $\rec(\mathcal X) \cap \lin(\mathcal X)^\perp \neq \{0\}$. By Lemma \ref{lem:basic}(c)--(d), $\mathcal D$ is nonempty and every $v \in \rec(\mathcal X) \cap \lin(\mathcal X)^\perp$ can be written as $v = a z$ with $a \geq 0$ and $z \in \mathcal D$, including $v = 0$ by taking $a = 0$. Since $\mathcal D \times \{0\} = \mathcal H \cap (\mathbb R^d \times \{0\})$, Lemma \ref{lem:normalization} implies that $\mathcal D$ is a compact polyhedron. For every $z \in \mathcal D$, Assumption \ref{ass:pp}(b)--(c) implies $M_A A z \in \mathcal K \setminus \{0\}$ and hence Assumption \ref{ass:pp}(d) implies $\nu' M_A A z > 0$.
First suppose $(A, b) \in \mathbf C_0$. Choose $v$ satisfying \eqref{eq:pp-cone-system} and write $v = a z$ with $a \geq 0$ and $z \in \mathcal D$. Then,
\begin{equation} \label{eq:pp-scale-in-cone}
M_A (b - A x_0) = a M_A A z~,  
\end{equation}
which combined with Assumption \ref{ass:pp}(c) implies $M_A (b - A x_0) \in \mathcal K$. Therefore, premultiplying both sides of \eqref{eq:pp-scale-in-cone} by $\nu'$ and using $\nu' M_A A z > 0$, we see
\begin{equation} \label{eq:pp-scale}
a = \frac{\nu' M_A (b - A x_0)}{\nu' M_A A z}~.
\end{equation}
Substituting \eqref{eq:pp-scale} into \eqref{eq:pp-scale-in-cone} and multiplying both sides by $\nu' M_A A z$, we have $F(A, b) z = 0$. Therefore, \eqref{eq:pp-directional-cone} holds.
Conversely, suppose $M_A (b - A x_0) \in \mathcal K$ and $F(A, b) z = 0$ for some $z \in \mathcal D$. Define $a$ by \eqref{eq:pp-scale} and note the numerator is nonnegative due to Assumption \ref{ass:pp}(d) and the denominator is strictly positive as shown above, so $a \geq 0$. Moreover,
\[ a M_A A z - M_A (b - A x_0) = \frac{F(A, b) z}{\nu' M_A A z} = 0~. \]
In particular, $v = a z \in \rec(\mathcal X) \cap \lin(\mathcal X)^\perp$ satisfies \eqref{eq:pp-cone-system}, so $(A, b) \in \mathbf C_0$. Part (ii) now follows.

For part (iii), first suppose \eqref{eq:pp-directional-cone} holds. Since $M_A (b - A x_0) \in \mathcal K$, $\nu' M_A (b - A x_0) \geq 0$ by Assumption \ref{ass:pp}(d), so $s = \{\nu' M_A (b - A x_0)\} / \{1 + \nu' M_A (b - A x_0)\} \in [0, 1)$ satisfies \eqref{eq:pp-membership}. Conversely, suppose $(z, s) \in \mathcal D \times [0, 1]$ satisfies \eqref{eq:pp-input}--\eqref{eq:pp-membership}. Equation \eqref{eq:pp-membership} rules out $s = 1$ and implies $\nu' M_A (b - A x_0) = s / (1 - s) \geq 0$. Define $a$ by \eqref{eq:pp-scale}. Since $\nu' M_A(b - A x_0) \geq 0$ and $\nu' M_A A z > 0$, $a \geq 0$. Equation \eqref{eq:pp-input} implies $M_A (b - A x_0) = a M_A A z \in \mathcal K$, so \eqref{eq:pp-directional-cone} holds.
\end{proof}

\begin{example} \label{eg:pp-interval-iv}
Consider the interval outcome IV model in Section \ref{sec:interval-sims} and the hypothesis $\theta_2 = \theta_{2, 0}$. We use the notation of that section and suppress dependence on $P$. Define $\delta := (\theta_1, \gamma')'$ and $q := (\ell', u')'$, so that $x = (\delta', q')' \in \mathcal X := \mathbb R^{d_W + 1} \times \mathbb R_+^{2 K}$. Multiplying the lower-bound equations by $-1$, we obtain the equivalent system $A x = b$, where
\[ A = \begin{pmatrix}
-E[Z X_1] & -E[Z W'] & \mathbf I_K & 0 \\
E[Z X_1] & E[Z W'] & 0 & \mathbf I_K
\end{pmatrix} \qquad \text{and} \qquad
b = \begin{pmatrix}
\theta_{2, 0} E[Z X_2] - E[Z Y^L] \\
E[Z Y^U] - \theta_{2, 0} E[Z X_2]
\end{pmatrix}~. \]
Here, $\lin(\mathcal X) = \mathbb R^{d_W + 1} \times \{0\}$ and $\mathcal X \cap \lin(\mathcal X)^\perp = \{0\} \times \mathbb R_+^{2 K}$, so Assumption \ref{ass:pp}(a) holds with $x_0 = 0$. Denote the first $d_W + 1$ columns of $A$ by $A_0$. Since $\mathbf 1' A_0 = 0$, we have $M_A \mathbf 1 = \mathbf 1$. For $q \in \mathbb R_+^{2 K} \setminus \{0\}$, we also have $\mathbf 1' M_A q = \mathbf 1' q > 0$, so $M_A q \neq 0$, as required by Assumption \ref{ass:pp}(b).
Take $\nu = \mathbf 1$ and define the closed pointed convex cone $\mathcal K := \{w \in \mathbb R^{2 K}: \|w\|_2 \leq \mathbf 1' w\}$. For every $q \geq 0$,
\[ \|M_A q\|_2 \leq \|q\|_2 \leq \mathbf 1' q = \mathbf 1' M_A q~, \]
so Assumption \ref{ass:pp}(c) holds. For $w \in \mathcal K \setminus \{0\}$, we have $\nu' w \geq \|w\|_2 > 0$, so Assumption \ref{ass:pp}(d) holds.

The known inequalities $-q \leq 0$ imply $h((\delta', q')') = 1 + \mathbf 1' q$. Since $\lin(\mathcal X)^\perp = \{0\} \times \mathbb R^{2 K}$, we obtain
\[ \mathcal D = \big \{ (0', q')': q \in \mathbb R_+^{2 K},\ \mathbf 1' q = 1 \big \}~. \]
Since the entries of $Z$ sum to one and $Y^L \leq Y^U$, we have $\mathbf 1' b = E[Y^U - Y^L] \geq 0$, so \eqref{eq:pp-membership} always has a solution $s \in [0, 1]$ and can be dropped. For $z = (0', q')' \in \mathcal D$, we have $\mathbf 1' q = 1$. Using $M_A A = (0, M_A)$ and $\mathbf 1' M_A = \mathbf 1'$, the remaining equation \eqref{eq:pp-input} becomes
\begin{equation} \label{eq:pp-interval-input}
(\mathbf 1' b) M_A q = M_A b \text{ for some } q \in \mathbb R_+^{2K} \text{ such that } \mathbf 1'q = 1~.
\end{equation}
By Lemma \ref{lem:pp-cone}(ii)--(iii), $A x = b$ has a solution $x \in \mathcal X$ if and only if \eqref{eq:pp-interval-input} holds.
\end{example}

\begin{example} \label{eg:pp-npiv}
Consider Example \ref{eg:npiv} with $x_1 < \cdots < x_H$, $S g \leq 0$ encoding that $g$ is weakly decreasing, and the hypothesis $g(x_1) = L_0$. We suppress dependence on $P$ and assume $P \{X = x_H\} > 0$. In this case, $\lin(\mathcal X) = \{0\}$ and $M_A = \mathbf I_K$. Taking $x_0 = L_0 \mathbf 1$, we have $\mathcal X = x_0 + \rec(\mathcal X)$, where
\[ \rec(\mathcal X) = \{v \in \mathbb R^H: v_1 = 0 \geq v_2 \geq \cdots \geq v_H\}~. \]
Assumption \ref{ass:pp}(a) holds, and $\rec(\mathcal X) \neq \{0\}$ because $H > K \geq 1$. For every $v \in \rec(\mathcal X) \setminus \{0\}$, we have $v_H < 0$, so
\[ \mathbf 1' \Pi' v = \sum_{1 \leq h \leq H} P \{X = x_h\} v_h < 0~. \]
Hence, $\Pi' v \neq 0$, as required by Assumption \ref{ass:pp}(b). Since $\Pi$ has nonnegative entries, Assumption \ref{ass:pp}(c) holds with $\mathcal K = \{w \in \mathbb R^K: w \leq 0\}$, and Assumption \ref{ass:pp}(d) holds with $\nu = -\mathbf 1$.

Writing the monotonicity restrictions as $g(x_{j + 1}) - g(x_j) \leq 0$ for $1 \leq j < H$, we have $h(g) = 1 + g(x_1) - g(x_H)$. For $v \in \rec(\mathcal X)$, we have $v_1 = 0$ and hence $h(x_0 + v) - h(x_0) = -v_H$, so
\[ \mathcal D = \{z \in \mathbb R^H: z_1 = 0 \geq z_2 \geq \cdots \geq z_H = -1\}~. \]
The identities $\mathbf 1' m = E[Y]$ and $\mathbf 1' \Pi' \mathbf 1 = 1$ imply $\nu' M_A (b - A x_0) = L_0 - E[Y]$ and
\[ F(\Pi', m) = (L_0 - E[Y]) \Pi' + (m - L_0 \Pi' \mathbf 1) (\mathbf 1' \Pi')~. \]
The equations in \eqref{eq:pp-input}--\eqref{eq:pp-membership} become
\begin{equation} \label{eq:pp-npiv-input}
\begin{aligned}
F(\Pi', m) z & = 0 \\
(1 + L_0 - E[Y]) s & = L_0 - E[Y] \\
(z, s) & \in \mathcal D \times [0, 1]~.
\end{aligned}
\end{equation}
By Lemma \ref{lem:pp-cone}(ii)--(iii), $\Pi' g = m$ has a solution $g \in \mathcal X$ if and only if \eqref{eq:pp-npiv-input} has a solution.
\end{example}

\section{Additional Simulations} \label{sec:additional-sims}
In this section we present a set of additional simulation exercises that are designed so that the rank conditions underlying \cite{cox2025testing} and \cite{goff2026testing} fail or are close to failing. Tests, tuning choices, and figure conventions follow Section \ref{sec:sims}.

\subsection{Rank Instability of the Coefficient Matrix} \label{sec:sim-rank}
We consider a design that violates the stable-rank condition in Assumption 2 of \cite{cox2025testing}: the population coefficient matrix is zero, whereas its sample analog has rank one with probability one at every sample size. We observe independent identically distributed pairs $(U_i, G_i) \in \mathbb R^2 \times \mathbb R^2$. At a hypothesized value $\theta \in \mathbb R$, the null hypothesis is
\[ H_0: E_P[U] + E_P[G] x = \theta \mathbf 1_2 \text{ for some } x \in [-1, 1]~. \]
This null hypothesis fits into \eqref{eq:null-Ab} with $\mathcal X = [-1, 1]$, $A(P) = E_P[G]$, and $\beta(P) = \theta \mathbf 1_2 - E_P[U]$.
Since $\mathcal X=[-1,1]$ is compact, $\lin(\mathcal X) = \{0\}$ and $M(P) = \mathbf I_2$. We do not employ the preprocessing step in Appendix \ref{sec:coefficient-preprocessing} since $\rec(\mathcal X) = \{0\}$.
The simulation design is specified by drawing independent vectors $U_i \sim N(0, \mathbf I_2)$ and $E_i \sim N(0, \mathbf I_2)$ and setting $G_i = J U_i + E_i / 4$, where
\[ J = \begin{pmatrix}0 & -1 \\ 1 & 0\end{pmatrix}~. \]
The pairs $(U_i, E_i)$ are independent across observations. Both $E_P[U]$ and $E_P[G]$ are zero, so the identified set for $\theta$ is $\{0\}$.

We estimate $A(P)$ and $\beta(P)$ by sample averages. Note that $A(P)$ has rank zero, whereas its sample analog satisfies
\[ \hat A_n = \frac{1}{n} \sum_{i = 1}^n G_i \sim N \bigg ( 0, \frac{17}{16 n} \mathbf I_2 \bigg )~. \]
Thus, $\hat A_n$ has rank one with probability one at every sample size.

\begin{figure}[t!]
\centering
\includegraphics{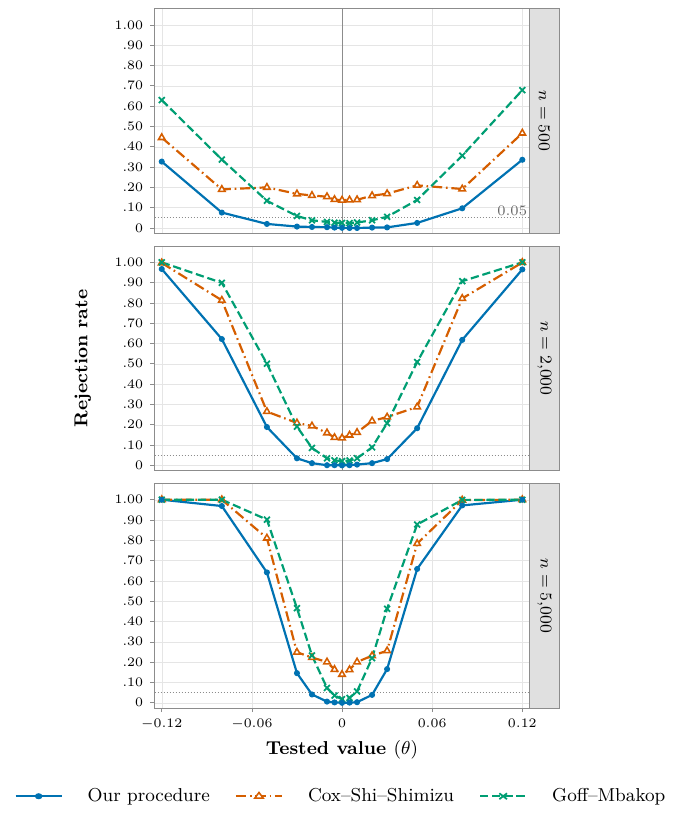}
\caption{Rank instability design: rejection rates over $1{,}000$ replications for sample sizes $n \in \{500, 2{,}000, 5{,}000\}$. Rows correspond to sample sizes. The horizontal axis shows the hypothesized value $\theta$. The vertical line marks the identified set $\{0\}$. The solid blue line with filled circles represents our procedure, the dash-dotted orange line with triangles represents Cox--Shi--Shimizu, and the dashed green line with crosses represents Goff--Mbakop. The dotted horizontal line marks the nominal rejection rate $0.05$.}
\label{fig:sim-rank}
\end{figure}

Figure \ref{fig:sim-rank} reports the results for $n \in \{500, 2000, 5000\}$. At $\theta = 0$, both our procedure and \cite{goff2026testing} reject less than $5\%$ of the time; in contrast, the procedure in \cite{cox2025testing} rejects between $13.5\%$ and $14.0\%$.

\subsection{Gaussian Wedge} \label{sec:sim-wedge}
Next, we consider a design with a large ``condition number'' in Assumption 4.17(i) of \cite{goff2026testing}. Define $\mu_Z := E_P[Z]$ and $\mu_W := E_P[W]$. At a hypothesized value $q$, we test whether some $(u, v)$ satisfies $\mu_W \leq u - 2 \mu_Z v$, $\mu_W \leq -u - 2 \mu_Z v$, and $v \geq -q$. We introduce nonnegative slack variables $s_1, s_2, s_3$ and define $x := (u, v, s_1, s_2, s_3)'$. Then, the null hypothesis fits into \eqref{eq:null-Ab} with $\mathcal X = \mathcal X(q) = \{(u, v, s_1, s_2, s_3) \in \mathbb R^2 \times \mathbb R_+^3: -v + s_3 = q\}$,
\[
A(P) = \begin{pmatrix}-1 & 2 \mu_Z & 1 & 0 & 0 \\ 1 & 2 \mu_Z & 0 & 1 & 0\end{pmatrix}~, \qquad \text{and} \qquad \beta(P) = (-\mu_W, -\mu_W)'~.
\]
The observations are Gaussian:
\[ \binom{Z}{W} \sim N \bigg ( \binom{0.005}{0.01}, 0.01 \begin{pmatrix}1 & -0.9 \\ -0.9 & 1\end{pmatrix} \bigg )~. \]
The linear system in \eqref{eq:null-Ab} has a solution if and only if $q \geq 1$; at $q = 1$, the only feasible pair is $(u, v) = (0, -1)$. All three procedures treat both means as unknown and impose no sign restriction. The coefficient rows of the first two inequalities are nearly linearly dependent, making the ``condition number'' in Assumption 4.17(i) of \cite{goff2026testing} large. 
We note that $\lin(\mathcal X(q)) = \operatorname{span}\{(1, 0, 0, 0, 0)'\}$ and the first column of $A(P)$ is known, so $M(P)$ is known. We do not apply preprocessing in Appendix \ref{sec:coefficient-preprocessing} because one of the required assumptions does not hold.

\begin{figure}[t!]
\centering
\includegraphics{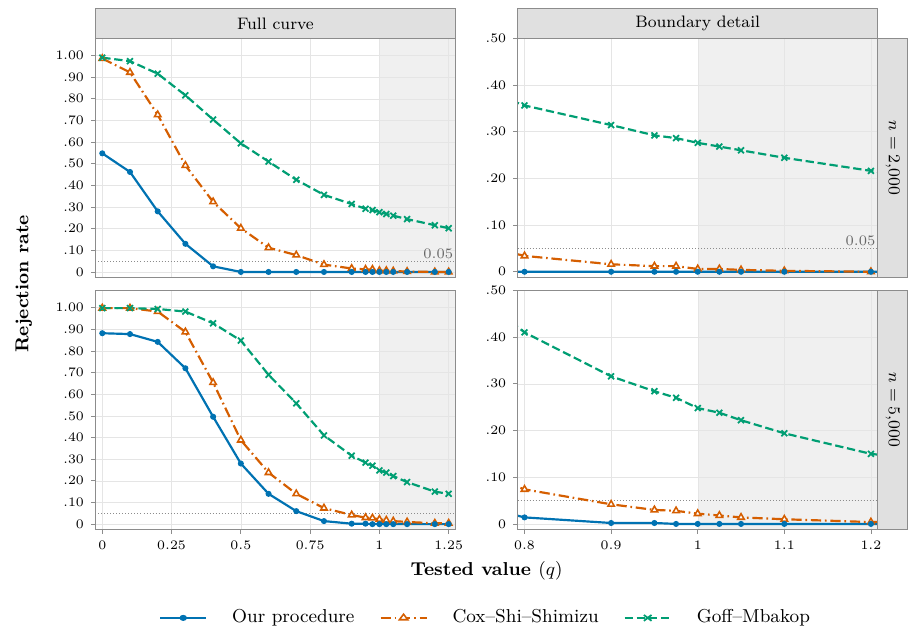}
\caption{Gaussian wedge design: rejection rates over $500$ replications for sample sizes $n \in \{2{,}000, 5{,}000\}$. Rows correspond to sample sizes. The left column shows the full curves, and the right column shows an enlarged view near $q = 1$ with an expanded rejection-rate scale. The horizontal axis shows the hypothesized value $q$. The shaded region marks the identified set $[1, \infty)$. The solid blue line with filled circles represents our procedure, the dash-dotted orange line with triangles represents Cox--Shi--Shimizu, and the dashed green line with crosses represents Goff--Mbakop. The dotted horizontal line marks the nominal rejection rate $0.05$.}
\label{fig:sim-wedge}
\end{figure}

Figure \ref{fig:sim-wedge} reports the results for $n \in \{2000, 5000\}$. At $q = 1$, both our procedure and \cite{cox2025testing} reject less than $5\%$ of the time; in contrast, the procedure in \cite{goff2026testing} rejects $27.6\%$ at $n = 2000$ and $24.8\%$ at $n = 5000$.

\subsection{Synthetic Parallel Trends} \label{sec:sim-liu}
Next, we consider a design called DGP-4 in \cite{liu2025synthetic}, which is based on Example \ref{eg:spt}. This design is calibrated to Current Population Survey data on women's log weekly earnings across $50$ states from $1979$ to $2018$. We assign placebo treatment to one state in 2018 and use the other 49 as potential donors. The true treatment effect is thus zero.

To construct the population mean earnings paths, \cite{liu2025synthetic} fits a factor model with rank four to the state-year means. We select three of the $49$ potential donors and set the treated state's mean earnings in each year equal to the average of their three fitted means. Unlike the original paper, we hold the identities of these three states and the treated state fixed across Monte Carlo replications, so that rejection rates are evaluated under a fixed population distribution. Each of the other $46$ donors has the same mean earnings as the treated state in every year through $2017$, but its mean earnings in $2018$ are lower by $\delta$. Thus, an equally weighted average of the three selected donors reproduces the treated state's entire mean earnings path, while any one of the other $46$ matches its pre-treatment path exactly and understates its $2018$ counterfactual by $\delta$. We consider $\delta \in \{1,5\}$, with $\delta=5$ corresponding to the original design of \cite{liu2025synthetic}. Within each state, the vectors of log weekly earnings over the $40$ years follow a multivariate normal distribution with the mean path specified above and a state-specific covariance matrix calibrated to the data. Observations are independent within and across states. We consider three sample sizes, with $\bar n \in \{250,500,1{,}000\}$ denoting the average number of observations per state. For each design, we allocate the $50\bar n$ observations across states in proportion to their sample sizes in the calibration. 

\begin{figure}[t!]
\centering
\includegraphics{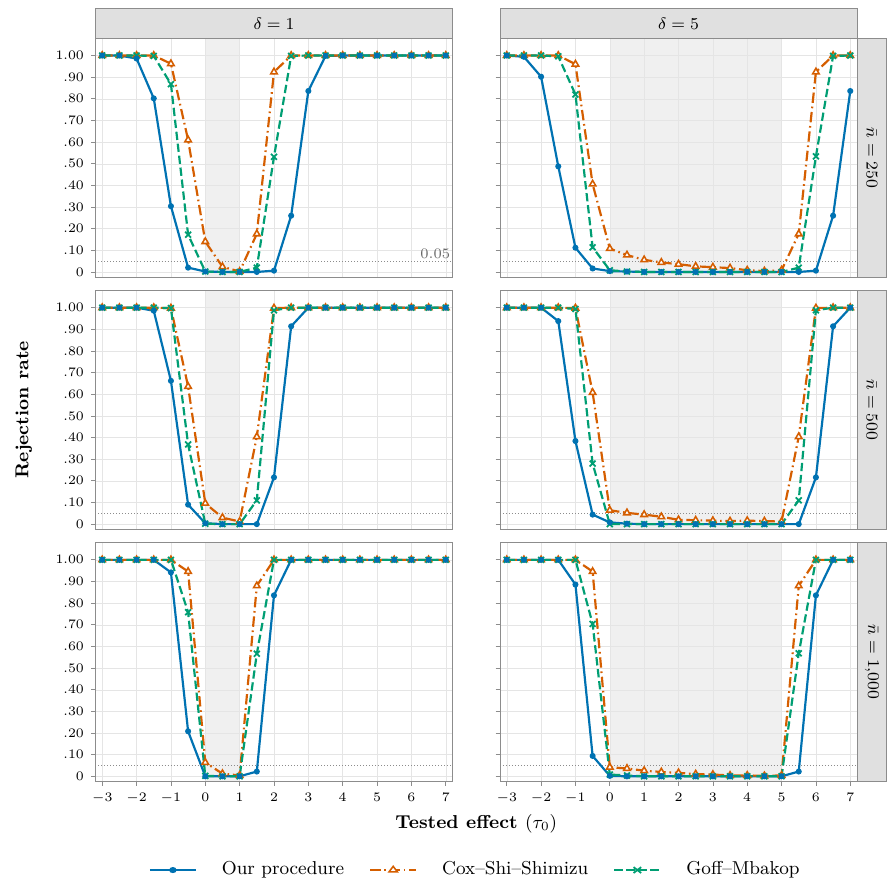}
\caption{Synthetic parallel trends design: rejection rates over $500$ replications for average sample sizes $\bar n \in \{250, 500, 1{,}000\}$ observations per state and gaps $\delta \in \{1, 5\}$. Rows correspond to average sample sizes, and columns correspond to gaps. The horizontal axis shows the hypothesized treatment effect $\tau_0$. The shaded region marks the identified set $[0, \delta]$. The solid blue line with filled circles represents our procedure, the dash-dotted orange line with triangles represents Cox--Shi--Shimizu, and the dashed green line with crosses represents Goff--Mbakop. The dotted horizontal line marks the nominal rejection rate $0.05$.}
\label{fig:sim-liu}
\end{figure}

Figure \ref{fig:sim-liu} reports the results for $\bar{n} \in \{250, 500, 1000\}$. At the original value  $\delta = 5$, the procedure in \cite{cox2025testing} rejects $10.8\%$ of the time at the lower endpoint of the identified set for $\bar{n} = 250$. At the modified value $\delta = 1$, it rejects $14\%$ at the lower endpoint. In contrast, our procedure and \cite{goff2026testing} reject less than $5\%$ of the time at every grid point in the identified set. Outside of the identified set, our procedure's average rejection rate is $13.65$ percentage points lower than that of \cite{goff2026testing}, although both procedures reject with high probability farther away, with steeper curves as $\bar n$ increases.

\end{document}